\pdfoutput=1
\newif\iffull
\fulltrue
\newif\ifproofs
\proofstrue

\iffull\documentclass[nonacm]{acmart}
\setcopyright{none}
\makeatletter                   %1
\def\mdseries@tt{m}             %1
\makeatother                    %1\usepackage[plain]{fancyref}
\usepackage[draft=true]{minted} %2
\usepackage{color}
\usepackage{hyperref}           %5
\hypersetup{
    colorlinks=true,
    linkcolor=blue,
    filecolor=red,      
    urlcolor=magenta,
    breaklinks=true,            %3
}
\usepackage{breakurl}           %3
\else\documentclass[sigplan,10pt]{acmart}
\fi

\iffull
\else
\acmConference[PLDI'20]{}{June 15-20, 2020}{London, UK}
\acmYear{2020}
\acmISBN{} % \acmISBN{978-x-xxxx-xxxx-x/YY/MM}
\acmDOI{} % \acmDOI{10.1145/nnnnnnn.nnnnnnn}
\startPage{1}

\setcopyright{acmcopyright}
\copyrightyear{2020}           %% If different from \acmYear
\fi
\usepackage{scalerel}
\usepackage{accents}
\usepackage{xspace}

\newcommand{\rulename}[1]{\textsc{#1}}
\newcommand{\Rule}[4][]{\ensuremath{\inferrule*[right={(#2)},#1]{#3}{#4}}}

\newcommand{\create}{\texttt{fcreate}\xspace}

\newcommand{\touch}{\texttt{ftouch}\xspace}
\newcommand{\spawn}{\texttt{spawn}\xspace}
\newcommand{\sync}{\texttt{sync}\xspace}

\newcommand{\SYS}{I-Cilk\xspace}
\newcommand{\ioF}{\texttt{io\_future}\xspace}
\newcommand{\cilkWrite}{\texttt{cilk\_write}\xspace}
\newcommand{\cilkRead}{\texttt{cilk\_read}\xspace}

\newcommand{\proxy}{\texttt{proxy}\xspace}
\newcommand{\emailapp}{\texttt{email}\xspace}
\newcommand{\jserver}{\texttt{jserver}\xspace}

\newcommand{\gp}{\gamma}

\newcounter{ccount}
\newenvironment{closeenum}
    {\begin{list}{\arabic{ccount}.}
    {\usecounter{ccount}
     \setlength{\topsep}{0.15\baselineskip}
     \setlength{\leftmargin}{10pt}
     \setlength{\parskip}{0pt}}}
    {\end{list}}

\newenvironment{closeitemize}
    {\begin{list}{$\bullet$}
    {%\advance\leftmargin -1em
      \setlength{\topsep}{0.15\baselineskip}
      \setlength{\parskip}{0pt}
      \setlength{\itemindent}{0pt}
      \setlength{\leftmargin}{10pt}}}
    {\end{list}}

\renewcommand{\paragraph}[1]{%
{\normalfont\normalsize\bfseries #1}%
}%

\newcommand{\altdiv}{\mathbin{\mathrlap{\mathrlap{\sim}:}\phantom{\sim}}}

\newcommand{\calcname}{\ensuremath{\lambda^{4}_i}}

\newcommand{\term}[1]{\textbf{\textit{#1}}}

\newcommand{\ctx}{\Gamma}
\theoremstyle{definition}
\newtheorem{defn}{Definition}

\newcommand{\ectx}{\cdot}

\newcommand{\dom}[1]{\mathit{dom}(#1)}

\newcommand{\defeq}{\triangleq}

\newcommand{\priowork}[1]{W_{#1}}
\newcommand{\prioworkof}[2]{\priowork{#2}(#1)}
\newcommand{\psnlt}[1]{\not\prec #1}
\newcommand{\longsp}[2]{S_{#2}(#1)}

\newcommand{\anc}[2]{#1 \sqsupseteq #2}
\newcommand{\sanc}[2]{#1 \sqsupseteq^s #2}
\newcommand{\wanc}[2]{#1 \sqsupseteq^w #2}
\newcommand{\nanc}[2]{#1 \not\sqsupseteq #2}
\newcommand{\neqanc}[2]{#1 \sqsupset #2}
\newcommand{\compwork}[2]{\ensuremath{\hspace{0.25em}\mathop{\mathclap{\nuparrow}\mathclap{\downarrow}}\hspace{0.25em}#2}}

\newcommand{\exec}[1]{\mathit{Exec}(#1)}

\newcommand{\uprio}[2]{\mathit{Prio}_{#1}(#2)}
\newcommand{\strengthen}[2]{\hat{#2}_{#1}}

\newcommand{\uthread}{\vec{u}}
\newcommand{\tscomp}{\cdot}

\newcommand{\sthread}[1]{#1}
\newcommand{\gthread}[3]{\ensuremath{#1 \xhookrightarrow[#2]{} #3}}

\newcommand{\tgraph}[3]{\dagq{\gthread{#1}{#2}{#3}}{\emptyset}{\emptyset}{\emptyset}}

\newcommand{\gthreads}{\mathcal{T}}
\newcommand{\spawns}{E^c}
\newcommand{\syncs}{E^t}
\newcommand{\reads}{E^w}
\newcommand{\dagq}[4]{\ensuremath{(#1, #2, #3, #4)}}

\newcommand{\resptimeof}[1]{\ensuremath{T(#1)}}

\newcommand{\kw}[1]{\mbox{\texttt{#1}}}
\newcommand{\cdparens}[1]{({#1})}
\newcommand{\cdsqbracks}[1]{\kw{[}{#1}\kw{]}}
\newcommand{\cd}[1]{{\lstinline!#1!}}

\newcommand{\kwunit}{\kw{unit}}
\newcommand{\kwnat}{\kw{nat}}
\newcommand{\prodsym}{\ensuremath{\times}}
\newcommand{\kwprod}[2]{\ensuremath{{#1} \prodsym {#2}}}
\newcommand{\sumsym}{\ensuremath{+}}
\newcommand{\kwsum}[2]{\ensuremath{{#1} \sumsym {#2}}}
\newcommand{\arrsym}{\ensuremath{\to}}
\newcommand{\kwarr}[2]{\ensuremath{{#1} \arrsym {#2}}}
\newcommand{\cmdsym}{\ensuremath{\kw{cmd}}}
\newcommand{\kwcmdt}[2]{\ensuremath{#1~\cmdsym \cdsqbracks{#2}}}
\newcommand{\atsym}{\ensuremath{\mathop{\kw{thread}}}}
\newcommand{\kwat}[2]{\ensuremath{#1 \atsym \cdsqbracks{#2}}}
\newcommand{\fasym}{\ensuremath{\forall}}
\newcommand{\kwforall}[3]{\ensuremath{\fasym #1\sim #2. #3}}
\newcommand{\kwreft}[1]{\ensuremath{#1~\kw{ref}}}

\newcommand{\cons}{C}
\newcommand{\cconj}[2]{\ensuremath{#1 \land #2}}
\newcommand{\prio}{\rho}
\newcommand{\prioc}{\overline{\prio}}
\newcommand{\vprio}{\pi}
\newcommand{\ple}[2]{#1 \preceq #2}
\newcommand{\plt}[2]{#1 \prec #2}

\newcommand{\nple}[2]{#1 \not\preceq #2}
\newcommand{\pnlt}[2]{#1 \not\prec #2}
\newcommand{\isprio}{~\kw{prio}}
\newcommand{\worlds}{R}
\newcommand{\prios}{R}

\newcommand{\kwnumeral}[1]{\overline{#1}}
\newcommand{\kwn}{\kwnumeral{n}}
\newcommand{\kwassn}{s}
\newcommand{\kwtriv}{\langle \rangle}
\newcommand{\kwfun}[2]{\ensuremath{\lambda #1{.}#2}}
\newcommand{\kwfix}[3]{\ensuremath{\mathop{\kw{fix}}#1{:}#2\mathbin{\kw{is}}#3}}
\newcommand{\kwifz}[4]{\ensuremath{\kw{ifz}~#1~\{#2; #3. #4\}}}
\newcommand{\kwepair}[2]{\ensuremath{\cdparens{{#1},{#2}}}}
\newcommand{\kweinl}[1]{\ensuremath{\kw{inl}~#1}}
\newcommand{\kweinr}[1]{\ensuremath{\kw{inr}~#1}}
\newcommand{\kwapply}[2]{\ensuremath{{#1}~{#2}}}
\newcommand{\kwfst}[1]{\ensuremath{\kw{fst}~#1}}
\newcommand{\kwsnd}[1]{\ensuremath{\kw{snd}~#1}}
\newcommand{\kwcase}[5]{\ensuremath{\kw{case}~#1~\{#2.#3; #4.#5\}}}
\newcommand{\kwspawn}[3]{\ensuremath{\kw{fcreate}\cdsqbracks{#1; #2}\{#3\}}}
\newcommand{\kwsync}[1]{\ensuremath{\kw{ftouch}~#1}}
\newcommand{\kwtid}[1]{\ensuremath{\kw{tid}\cdsqbracks{#1}}}
\newcommand{\kwwapp}[2]{\ensuremath{#1\cdsqbracks{#2}}}
\newcommand{\kwcmd}[2]{\ensuremath{\kw{cmd}\cdsqbracks{#1}~\{#2\}}}
\newcommand{\kwwlam}[3]{\ensuremath{\Lambda #1 \sim #2. #3}}
\newcommand{\kwlet}[3]{\ensuremath{\kw{let}~#1 = #2~\kw{in}~#3}}
\newcommand{\kwdcl}[4]{\ensuremath{\kw{dcl}~\cdsqbracks{#1}~#2 := #3~\kw{in}~#4}}
\newcommand{\kwref}[1]{\ensuremath{\kw{ref}\cdsqbracks{#1}}}
\newcommand{\kwderef}[1]{\ensuremath{!#1}}
\newcommand{\kwassign}[2]{\ensuremath{#1 := #2}}

\newcommand{\cmd}{m}
\newcommand{\kwbind}[3]{\ensuremath{#2 \leftarrow #1; #3}}
\newcommand{\kwret}[1]{\ensuremath{\kw{ret}~#1}}
\newcommand{\fresh}{~\kw{fresh}}

\newcommand{\graph}{g}
\newcommand{\ethread}{{[]}}
\newcommand{\scompsym}{\oplus}
\newcommand{\scomp}[1]{\scompsym_{#1}}
\newcommand{\egraph}{\emptyset}

\newcommand{\sig}{\Sigma}

\newcommand{\esig}{\cdot}
\newcommand{\emem}{\emptyset}
\newcommand{\hastype}[2]{#1 \mathop{:} #2}
\newcommand{\sigtype}[3]{#1 \mathord{\sim} #2 \mathord{@} #3}
\newcommand{\sigrtype}[2]{#1 \mathord{\sim} #2}
\newcommand{\etyped}[5][\worlds]{#3 \vdash^{#1}_{#2} #4 : #5}
\newcommand{\cmdtyped}[6][\worlds]{#3 \vdash^{#1}_{#2} #4 \altdiv #5 \mathord{@} #6}
\newcommand{\mtyped}[3]{\vdash^{#1} #2 : #3}
\newcommand{\meetc}[3][\worlds]{#2 \vdash^{#1} #3}

\newcommand{\sstyped}[5][\worlds]{\vdash^{#1}_{#2} #3 : #4 \mathbin{@} #5}
\newcommand{\stackaccepts}[6][\worlds]{\vdash^{#1}_{#2} #3 \mathbin{\vartriangleleft :} #4 \leadsto #5 \mathbin{@} #6}
\newcommand{\stackacceptsc}[6][\worlds]{\vdash^{#1}_{#2} #3 \mathbin{\blacktriangleleft :} #4 \leadsto #5 \mathbin{@} #6}

\newcommand{\tpcp}{\otimes}
\newcommand{\tp}{\mu}
\newcommand{\etp}{\emptyset}
\newcommand{\mem}{\sigma}

\newcommand{\memrent}[3]{(#1, #2, #3)}
\newcommand{\mement}[4]{#1 \mapsto \memrent{#2}{#3}{#4}}

\newcommand{\stack}{k}
\newcommand{\estack}{\epsilon}
\newcommand{\ssend}[2]{#1 \mathbin{\vartriangleright} #2}
\newcommand{\sreturn}[2]{#1 \mathbin{\vartriangleleft} #2}
\newcommand{\scsend}[2]{#1 \mathbin{\blacktriangleright} #2}
\newcommand{\screturn}[2]{#1 \mathbin{\blacktriangleleft} #2}
\newcommand{\shole}{\text{\---}}
\newcommand{\scp}[2]{#1; #2}
\newcommand{\stackstate}{K}

\newcommand{\lconfig}[3]{#1 \mid #2 \tpcp #3}
\newcommand{\rconfig}[5]{#1 \tpcp #2 \mid #3 \mid #4 \mid #5}
\newcommand{\gconfig}[4]{#1 \mid #2 \mid #3 \mid #4}

\newcommand{\cthread}[4]{\ensuremath{#1 \xhookrightarrow[#2; #3]{} #4}}

\newcommand{\estep}{\mapsto}
\newcommand{\mstep}{\mathbin{\mathbf{\Rightarrow}}}
\newcommand{\gstep}{\mathbin{\mathbf{\Rightarrow}}}

\newcommand{\secput}[2]{\section{#2}\label{sec:#1}}
\newcommand{\secref}[1]{Section~\ref{sec:#1}}

\newcommand{\figref}[1]{Figure~\ref{fig:#1}}

\newcommand{\thmref}[1]{Theorem~\ref{thm:#1}}

\renewcommand{\eqref}[1]{Equation~(\ref{eq:#1})}

\newif\ifnotes
\relax

\newcounter{remark}[section]

\newcommand{\code}[1]{\lstinline!#1!}

\newcommand{\subheading}[1]{\subsubsection*{\bf #1}}

\newcommand{\punt}[1]{}

\newcommand{\ContinueLineNumber}{\lstset{firstnumber=last}}

\newcommand{\mylineskip}{0.30cm}
\iffull
\newcommand{\langfigsize}{\small}
\else
\newcommand{\langfigsize}{\small}
\fi

\notesfalse %% use this to turn off notes

\usepackage{flushend}
\usepackage{relsize}
\usepackage{mathtools}
\usepackage{math-cmds}
\usepackage{mathpartir}
\usepackage{listings, multicol}
\usepackage{tikz}
\usepackage{code}

\begin{document}
\iffull
\sloppy

%% Title information
\title{Responsive Parallelism with Futures and State}         %% [Short Title] is optional;
                                        %% when present, will be used in
                                        %% header instead of Full Title.
%\titlenote{with title note}             %% \titlenote is optional;
                                        %% can be repeated if necessary;
                                        %% contents suppressed with 'anonymous'
%\subtitle{Subtitle}                     %% \subtitle is optional
%\subtitlenote{with subtitle note}       %% \subtitlenote is optional;
                                        %% can be repeated if necessary;
                                        %% contents suppressed with 'anonymous'
\iffull
%\subtitle{Expanded Version; Original Published at PLDI 2020}
\fi

\author{Stefan K. Muller}
\email{smuller@cs.cmu.edu}
\affiliation{Carnegie Mellon University}
\authornote{The first two authors contributed equally to this work.}

\author{Kyle Singer}
\email{kdsinger@wustl.edu}
\affiliation{Washington University in St. Louis}
\authornotemark[1]

\author{Noah Goldstein}
\email{goldstein.n@wustl.edu}
\affiliation{Washington University in St. Louis}

\author{Umut A. Acar}
\email{umut@cs.cmu.edu}
\affiliation{Carnegie Mellon University}

\author{Kunal Agrawal}
\email{kunal@wustl.edu}
\affiliation{Washington University in St. Louis}

\author{I-Ting Angelina Lee}
\email{angelee@wustl.edu}
\affiliation{Washington University in St. Louis}

\begin{abstract}
Motivated by the increasing shift to multicore computers, 
recent work has developed language support for responsive
parallel applications that mix compute-intensive tasks with
latency-sensitive, usually interactive, tasks.
These developments include calculi that allow assigning priorities to
threads, type systems that can rule out priority inversions, and
accompanying cost models for predicting responsiveness.
These advances share one important limitation: all of
this work assumes purely functional programming.
This is a significant restriction, because many realistic interactive
applications, from games to robots to web servers, use mutable state,
e.g., for communication between threads.

In this paper, we lift the restriction concerning the use of state.
We present {\calcname}, a calculus with implicit parallelism in the
form of prioritized futures and mutable state in the form of
references.
Because both futures and references are first-class values,
{\calcname} programs can exhibit complex dependencies,
including interaction between threads and with the external world
(users, network, etc).
To reason about the responsiveness of \calcname{} programs, we extend
traditional graph-based cost models for parallelism to account for
dependencies created via mutable state, and we present a type system
to outlaw priority inversions that can lead to unbounded blocking.
We show that these techniques are practical by implementing them in
C++ and present an empirical evaluation.

\end{abstract}

\maketitle
\renewcommand{\shortauthors}{S. K. Muller, K. Singer, N. Goldstein, U. A. Acar, K. Agrawal and I. Lee}

\section{Introduction}

Advances of the past decade have brought multicore computers to the
mainstream.  Many computers today, from a credit-card-size Raspberry
Pi with four cores to a rack server, are built with multiple processors
(cores).
These developments have led to increased interest in languages and
techniques for writing parallel programs 
with \term{cooperative threading}.
In cooperative threading, the user expresses parallelism at a high
level and a language-supplied scheduler manages
parallelism at run time.
Many languages, libraries, and systems for cooperative threading have
been developed, including MultiLisp~\cite{halstead84},
NESL~\citep{nesl-94}, dialects of Cilk~\citep{frigolera98,
  DanaherLeLe06, Leiserson10, SingerXuLe19}, OpenMP~\cite{OpenMP18},
Fork/Join Java~\citep{lea00}, dialects of Habanero~\citep{BarikBuCa09,
  CaveZhSh11, is-habanero-14}, TPL~\citep{tpl09},
TBB~\citep{threadingbuildingblocksmanual}, X10~\citep{x10-2005},
parallel ML~\citep{manticore-implicit08, manticore-implicit-11,
  szj-multimlton14, rmab-mm-2016,
  gwraf-hieararchical-2018}, and parallel
Haskell~\citep{keller+2010, kttn-zoo-2014}.
% removed ref westrick+disent-2020; not sure what it is

Cooperative threading is well suited for compute-intensive jobs and may
be used to maximize \term{throughput} by finishing a job as quickly as
possible.
But modern applications also include interactive jobs where a
thread may needed to be completed as quickly as possible.
For such interactive applications, the main optimization criterion is
\term{responsiveness} --- how long each thread takes to respond to
a user.
To meet the demands of such applications, the systems community has
developed \term{competitive threading} techniques, which focus on
hiding the latency of blocking operations by multiplexing independent
sequential threads of control~\citep{hjtww-1993, furm-2000, bdmf-2010, ggdmfb-2014}.

Historically, collaborative and competitive threading have been researched
largely separately.  With the mainstream availability of parallel
computers, this separation is now obsolete: many jobs today include
both compute-intensive tasks and interactive tasks.
In fact, applications such as games, browsers, design tools, and all
sorts of interesting interactive systems involve both compute-heavy
tasks (e.g., graphics, AI, statistics calculations) and interaction.
Researchers have therefore started bridging the two worlds.
\citet{mah-responsive-2017, mah-priorities-2018, mwa-fairness-2019}
have developed programming-language techniques
that allow programmers to write cooperatively
threaded programs and also assign priorities to threads, as in
competitive threading.
By using a type system~\cite{mah-priorities-2018} and a cost model,
the authors present techniques for reasoning about the responsiveness
of parallel interactive program.

All of this prior work has made some progress on bridging
collaborative and competitive threading, but it makes an important
assumption: pure functional programming.
%
%Even though the previously proposed calculi and related techniques
%support some effects, such as I/O operations,
Specifically, the work does not allow for memory
effects, which are crucial for allowing threads to communicate.
This restriction can be significant, because nearly all realistic
interactive applications rely on mutable state and effects.
As an example, consider a basic server consisting of two entities: a
high-priority event loop handling queries from a user and a
low-priority background thread for optimizing the server's database.
Under Muller et al.'s work, the event loop and background thread can
only communicate by synchronizing, but such a synchronization would
lead to a priority inversion.
If effects were allowed, then the threads could communicate by using a
piece of shared state.

In this paper, we overcome this restriction by developing programming language support for collaborative and competitive threading in the presence of state.
To this end, we consider {\calcname}, a core calculus for an
implicitly parallel language with mutable state in the
form of references.  The parallel portion of the calculus is based on
\term{futures}, which represent asynchronous computations as first-class
values. Futures can be created and synchronized in a very general fashion.
The calculus also allows programmers to assign priorities to futures,
which represent their computational urgency.
Because it combines futures and state, {\calcname} is very expressive
and enables writing conventional nested-parallel programs as well as
those with more complex and dynamic dependencies.
For example, we can parallelize a dynamic-programming algorithm by creating an
initially empty array of future references and then populating the array by
creating futures, which may all be executed in parallel.
Similarly, we can express rich interactive computations, e.g., a
network event can be delegated to a future that sends asynchronous
status updates via a piece of shared state implemented as a reference.
%
%% Save space
%% Finally, even though the calculus does not directly include it, \emph{polling} of a future, which is an important construct for expressing responsive interactive programs, can be implemented.  The idea is to pair a future with an ordinary reference, where it places its result; to poll the future another future can simply check the contents of the reference without needing to synchronize.
 
The high degree of expressiveness in \calcname{} makes it tricky to
reason about the cost due to priority inversions
and non-determinism due to scheduling: because of the presence of state, the computation  may depend on scheduling decisions.
For these reasons, traditional graph-based cost models of
parallel computations~\cite{BlellochGr95,BlellochGr96,SpoonhowerBlHaGi08}
do not apply to programs that mix futures and state.
Such models typically do not take priorities into account and assume
that scheduling does not change the computation graph.
%

%In \calcname{}, by contrast, the execution, and hence the computation graph,
%may depend on the schedule due to the presence of mutable state and
%the resulting determinacy races~\cite{feng97efficient}.
% 
%% For example, in an event-driven program where two futures send work to each other (represented  as a future) via a memory cell of future type, leading to a program whose computation graphs could differ vastly from one execution to another, depending on the events.
%% %
%Another fundamental difference is the presence of priorities, which, if used incorrectly, can lead to priority inversions.

% which are typically absent from traditional collaborative threading models.
%
%If priorities are used incorrectly, which is not hard to do, then
%scheduling decisions could cause a high priority future thread to wait
%on a low-priority one.

We tackle these challenges by using a combination of algorithmic and formal techniques.
On the algorithmic side, we extend traditional graph-based cost
models to include information about priorities as well as
``happens-before'' edges that capture certain dependencies by reifying
execution-dependent information flow through mutable state.
We then prove that if a computation graph has no priority inversions,
then it can be scheduled by using an extension of greedy scheduling
with priorities to obtain provable bounds on the response time of any
thread (\thmref{gen-brent}).
Priority inversions are not simple to reason about,
so we present a type system for {\calcname} that guarantees that
any well-typed program has no priority inversions.
To establish the soundness of the type system (\thmref{soundness}),
we model the structure of the computation by giving a dynamic
semantics that, in addition to evaluating the program, creates a
computation graph that captures both traditional dependencies between
threads and also non-traditionally captures certain happens-before
dependencies to model the impact of mutable state on the computation.

Because \calcname{} is a formal system, it can in principle be
implemented in many different languages.
For this paper, we chose to implement such a system in the context of
C/C++ because many real-world interactive applications with stringent
performance requirements are written in C/C++.
%, because this enables us
%to compare with highly optimized codes.
%
Specifically, we have developed \SYS,
a task parallel platform that supports interactive parallel applications.
\SYS is based on Cilk, a parallel dialect of C/C++.
As with traditional cooperative threading systems, \SYS consists of a
runtime scheduler that dynamically creates threads and maps them onto
available processing cores.  Unlike traditional task-parallel
platforms, however, \SYS supports competitive threading by allowing
the programmer to specify priorities of tasks.
Perhaps somewhat unexpectedly, \SYS also includes an implementation of the
\calcname{} type system to rule out priority inversions.  The type
system is implemented by using inheritance, template programming, and
other features of C++ to encode the restrictions necessary to prevent
priority inversions.
Because C++ is not a safe language, this implementation of the type
system expects the programmer to obey certain conventions.

The thread scheduler of \SYS aims to implement the scheduling
principle that \thmref{gen-brent} relies on.
This is challenging to do efficiently because it requires maintaining
global information within the scheduler that can only be achieved via
frequent synchronizations.
Instead, \SYS approximates optimal scheduling by utilizing a two-level
adaptive scheduling strategy that re-evaluates the scheduling decision
at a fixed scheduling quantum.

We empirically evaluate \SYS using three moderately-sized application
benchmarks (about 1K lines each).  These applications fully utilize
the features of \SYS (including I/O and prioritization of tasks).  We
will dive into one application in detail to illustrate the use of
future references and mutable states.  To demonstrate the efficiency
of \SYS, we compare the response times and execution times of tasks at
different priority levels running on \SYS and on a baseline system
that behaves like \SYS but does not account for priority.  Empirically we
demonstrate that indeed \SYS provides much better response time,
illustrating the efficacy of its scheduler.

In summary, the contributions of this paper include: 
\begin{closeitemize}
\item a cost model for imperative parallel programs that incorporates
  scheduler-dependence through mutable state (\secref{dag});
\item a calculus {\calcname} for imperative parallel programs, equipped with
  a type system that guarantees absence of priority inversions (\secref{lang});
% \item An extension of the Parallel ML language~\citep{spoonhower:phd,
% rmab-mm-2016, gwraf-hieararchical-2018, mah-priorities-2018},
% with a compiler that soundly checks for priority inversions.
\item \SYS, a C/C++-based task parallel platform that supports interactive
  parallel applications with a type system and scheduler that embody the ideas 
  of the threading model, type system, and cost model of \calcname{}
  (\secref{impl}); and 
\item an empirical evaluation of \SYS using three large case studies written
  with \SYS (\secref{eval}).
\end{closeitemize}

\secput{dag}{A DAG Model for Responsiveness}

\subsection{Preliminaries}

For the purpose of this paper, we will consider programs with first-class
threads that implement futures.
Because our models and scheduling algorithms are largely independent of the
language mechanisms by which threads are created, we will simply refer to
``threads'' here.
%To avoid confusion with our C++ implementation
%which distinguishes between ``future handles'' and ``future threads'', we will
%simply use the term ``thread'' here.
%
We assign threads a priority, written~$\prio$, drawn from a partially
ordered set~$\prios$, where ~$\ple{\prio_1}{\prio_2}$ means that
priority~$\prio_1$ is lower than priority~$\prio_2$ or $\prio_1 =
\prio_2$.
We write~$\plt{\prio_1}{\prio_2}$ for the strict partial-order
relation that does not allow for reflexivity.
Note that a total order is a partial order by definition and threads
can be given priorities from a totally ordered set, e.g., integers.

Threads interact with each other in two ways.
First, a thread~$a$ may \term{create}
a thread~$b$, after which the two threads run in parallel.
We call this operation, which returns a handle to~$b$,
``future-create'' or simply {\create}.
Second, a thread~$a$
may wait
for a thread~$b$ to complete before proceeding.
We call this operation ``future-touch'' or {\touch}.
This model subsumes the classic fork-join (spawn-sync) parallelism.

% both of which will be included in the languages and applications we consider
%in the remainder of the paper.

% \ur{We might not be using spawn-sync in apps anymore.}

%\ur{The next couple of paragraphs are preliminaries/background. I would start a subsection.}

As is traditionally done, we can represent the execution of a parallel
program with a \term{Directed Acyclic Graph} or a \term{DAG}.
A vertex of the DAG represents an operation (without
loss of generality, we will assume that a single vertex represents a uniform
unit of computation time, such as a processing core cycle).
A directed edge from~$u$ to~$u'$, written~$(u, u')$, indicates that the
operation represented by~$u'$ depends on the operation represented by~$u$.
We write~$\anc{u}{u'}$ to mean that~$u$ is an \term{ancestor} of~$u'$,
i.e., there is a (directed) path from~$u$ to~$u'$
(it may be that~$u = u'$).
If it is the case that~$\nanc{u}{u'}$ and~$\nanc{u'}{u}$, then~$u$ and~$u'$
may run in parallel.

A \term{schedule} of a DAG is an assignment of vertices to processing cores 
at each time step during the execution of a parallel program.
Schedules must obey the dependences in the DAG: a vertex may only
be assigned to a core if it is \term{ready}, that is, if all of its
(proper) ancestors have been assigned on prior time steps.
The goal of an efficient scheduler for parallel programs is to construct as
short a schedule as possible.
Constructing an optimal schedule is impossible when, as in many real programs,
the DAG unfolds dynamically during execution and is not known ahead of time
(even a relaxed \term{offline} version of the problem in which the DAG is
known ahead of time is NP-hard~\citep{ullman75}).
However, prior results have shown that schedules obeying
certain scheduling principles are within a constant factor of
optimal length while making decisions based only on information available
online (i.e., they need only know the set of ready vertices at any point in
time).
One such scheduling principle for DAGs with priorities is \term{prompt scheduling}.
At each time step, a prompt schedule assigns to a core
a ready vertex~$u$ such that no currently unassigned vertex is higher-priority
than~$u$
%\footnote{This procedure is simpler in a totally ordered setting,
%  in which it could be described by saying ``assigns a vertex~$u$ of the
%  highest possible priority''.} 
repeatedly until no cores remain or no ready vertices remain.

\subsection{Weak Edges}

Traditionally, cost models for parallel programs assume that
scheduling does not change the DAG of a parallel computation.
This assumption is reasonable for deterministic programs
% and even certain non-deterministic programs, 
and provides a nice layer of abstraction
over scheduling --- we can assume that any schedule of a DAG
corresponds to a valid execution.
This fundamental assumption breaks in our setting where threads are
first class values and state can be used to communicate in an
unstructured fashion, leading to determinacy races.

Consider as an example the following program.

%\begin{minipage}{0.49\textwidth}
\begin{minipage}{0.49\columnwidth}
\begin{lstlisting}[escapeinside=@@]
thread t = NULL;

void g() {} @\label{code1:g}@
void f() {
  t = @\create@(g);@\label{code1:f}@
}
\end{lstlisting}
\end{minipage}
\begin{minipage}{0.5\columnwidth}
\ContinueLineNumber
\begin{lstlisting}[escapeinside=@@]
void main() {
  @\create@(f); @\label{code1:spawn}@
  if (t != NULL) { @\label{code1:cond}@
    @\touch@(t); @\label{code1:join}@
  }
}
\end{lstlisting}
\end{minipage}
%\end{minipage}

%

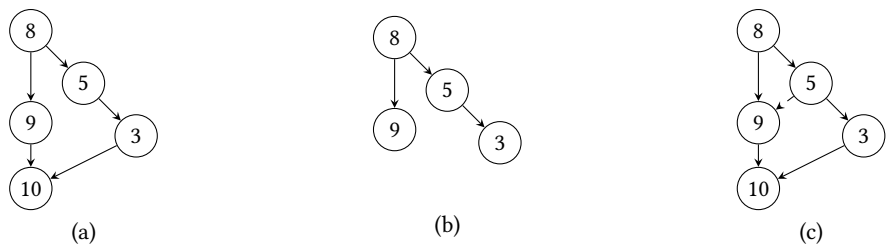
\begin{figure}
  \centering
  \begin{minipage}{0.32\columnwidth}
    \centering
    \begin{tikzpicture}[scale=0.70]
    \tikzstyle{n} = [draw,circle,fill=white,inner sep = 2pt,minimum width=16pt];
    \tikzstyle{edge} = [draw,-stealth];

    \node[n] (spawn) at (0, 0) {\ref{code1:spawn}};
    \node[n] (f) at (1, -1) {\ref{code1:f}};
    \node[n] (g) at (2, -2) {\ref{code1:g}};
    \node[n] (cond) at (0, -1.75) {\ref{code1:cond}};
    \node[n] (join) at (0, -3) {\ref{code1:join}};
    \path[edge] (spawn) -- (f);
    \path[edge] (f) -- (g);
    \path[edge] (spawn) -- (cond);
    \path[edge] (cond) -- (join);
    \path[edge] (g) -- (join);
    \end{tikzpicture}

    (a)
  \end{minipage}%
  \begin{minipage}{0.32\columnwidth}
    \centering
  \begin{tikzpicture}[scale=0.70]
    \tikzstyle{n} = [draw,circle,fill=white,inner sep = 2pt,minimum width=16pt];
    \tikzstyle{edge} = [draw,-stealth];

    \node[n] (spawn) at (0, 0) {\ref{code1:spawn}};
    \node[n] (f) at (1, -1) {\ref{code1:f}};
    \node[n] (g) at (2, -2) {\ref{code1:g}};
    \node[n] (cond) at (0, -1.75) {\ref{code1:cond}};
    \node[minimum width=16pt] (gh) at (0, -3) {};
    \path[edge] (spawn) -- (f);
    \path[edge] (f) -- (g);
    \path[edge] (spawn) -- (cond);
  \end{tikzpicture}

  (b)
  \end{minipage}%
  \begin{minipage}{0.32\columnwidth}
    \centering
  \begin{tikzpicture}[scale=0.70]
    \tikzstyle{n} = [draw,circle,fill=white,inner sep = 2pt,minimum width=16pt];
    \tikzstyle{edge} = [draw,-stealth];
    \tikzstyle{wedge} = [draw,dashed,-stealth];

    \node[n] (spawn) at (0, 0) {\ref{code1:spawn}};
    \node[n] (f) at (1, -1) {\ref{code1:f}};
    \node[n] (g) at (2, -2) {\ref{code1:g}};
    \node[n] (cond) at (0, -1.75) {\ref{code1:cond}};
    \node[n] (join) at (0, -3) {\ref{code1:join}};
    \path[edge] (spawn) -- (f);
    \path[edge] (f) -- (g);
    \path[edge] (spawn) -- (cond);
    \path[wedge] (f) -- (cond);
    \path[edge] (cond) -- (join);
    \path[edge] (g) -- (join);
  \end{tikzpicture}

  (c)
  \end{minipage}
  \caption{DAGs in which the main thread reads a valid
    thread handle (a) and NULL (b), and a DAG with a weak edge representing
    a read of a valid thread handle (c). Vertices are labeled with the line of
    code they represent and threads are arranged in columns.}
  \label{fig:dags1}
\iffull\else \vspace{-5mm} \fi
\end{figure}

The DAG for this program, in particular whether there is an edge from~\cd{g}
to~\cd{main} representing the {\touch} on line \ref{code1:join},
depends crucially on whether~\cd{f} performs the {\create} and
assignment to~\cd{t} before~\cd{main} reads~\cd{t}, that is, on whether the
conditional on line~\ref{code1:cond} returns true or false.
In fact, depending on the outcome of the condition, this program gives
rise to one of two DAGs, both shown in Figure~\ref{fig:dags1}: one in which the
conditional is true and one in which it is false.
Applying the traditional separation between DAGs and schedules, given
DAG (a), the scheduler could execute the vertices in the following
order: 8, 9, 5, 3, 10.
But under this schedule, the read on line~\ref{code1:cond} should
read~\cd{NULL}, and thus line 10 should not be executed at all!
Similarly, the scheduler could execute DAG (b) in the order
8, 5, 3, 9, in which case the read would read a valid thread handle.

%% \ur{Drop the negative first sentence.  I am not sure if it is correct.  What is a conventional DAG?  Our DAGs seem pretty conventional to me as well.  Say something like to capture the dependence between x and y in the DAG, we will use a *weak edge*...}
%% %
%% Conventional DAGs do not have the ability to encode the fact that the choice
%% of DAG for a particular execution depends on the schedule, 

The issue is that each DAG is valid for only certain schedules but not all.
%
%The fact that DAG (a) only makes sense in schedules that execute line 5 before
%line 9 and DAG (b) only makes sense in schedules that execute line 9 before
%line 5.
%
To encode this information, we extend the traditional notion
of DAGs with a new type of edge we call a \term{weak edge}.
A weak edge from~$u$ to~$u'$ records the fact that the given DAG makes sense
only for schedules where~$u$ is executed before~$u'$.
We call such a schedule \term{admissible}.
%
%, and our results will quantify over
%only admissible schedules for a given DAG.
%
%
As an example, DAG (c) of Figure~\ref{fig:dags1} includes a weak edge
(shown as a dotted line) from 5 to 9.
The schedule 8, 5, 9, 3, 10 is an admissible schedule of DAG (c),
but 8, 9, 5, 3, 10 is not.

At first sight, the reader may feel that we can replace a weak edge
with an ordinary (\term{strong}) edge.
This is not quite correct, as strong and weak edges are treated differently
in determining whether a schedule is prompt.
Recall that a schedule is prompt if it assigns ready vertices in priority
order.
In the presence of weak edges, we define a vertex~$u$ to be ready when all of
its \term{strong parents}, that is, vertices~$u'$ such that there exists
a strong edge~$(u', u)$, have executed.

Consider again DAG (c) from Figure~\ref{fig:dags1}, but now suppose we
wish to construct a prompt schedule on two cores.
By the above definition, a prompt schedule must execute vertex 8, followed
by 5 and 9 in parallel, followed by 3, followed by 10.
This is, in fact, the only prompt schedule of DAG (c), but it is not admissible
because it does not execute 5 before 9.
We thus conclude that there are no prompt admissible schedules of DAG
(c) on two cores and DAG (b) is the only valid DAG for a
two-core execution of this program (as DAG (b) has no weak edges,
any prompt schedule of it is admissible).
If we were to replace the weak edge~$(5, 9)$ with a strong edge,
there would be a prompt schedule of DAG (c) that executes 8, followed by
5, followed by 9 and 3, followed by 10.
As always, a strong edge forces vertex~9 to wait for vertex~5, but this
violates the intended semantics of the program as a simple read operation
should not have to block waiting for a write.

In summary, strong edges determine what schedules are valid for a given DAG,
while weak edges determine whether a DAG is valid for a given schedule.
That is, weak edges internalize information about schedules into the DAG,
breaking what would otherwise be a circular dependency between constructing
a DAG and constructing a schedule of it.

We extend the notions of ancestors and paths to distinguish between
weak and strong edges.
We say that a path is \term{strong} if it contains no weak edges.
If~$\anc{u}{u'}$ and all paths from~$u$ to~$u'$ are strong, then we say that~$u$
is a \term{strong ancestor} of~$u'$ and write~$\sanc{u}{u'}$.
On the other hand, if there exists a weak path (i.e., a path with a weak edge)
from~$u$ to~$u'$, we say~$u$ is a \term{weak ancestor} of~$u'$
and write~$\wanc{u}{u'}$.
We will continue to drop the
superscript if it not important whether~$u$ is a weak or strong ancestor.

In formal notation, we represent a DAG~$\graph$ as a
quadruple $\dagq{\gthreads}{\spawns}{\syncs}{\reads}$.
The first component of the quadruple is a mapping from thread symbols,
for which we will use the metavariables~$a$,~$b$ and variants, to a
pair of that thread's priority and the vertices it comprises.
We use the notation~$\uthread$ for a sequence of
vertices~$u_1 \tscomp{} \dots \tscomp{u_n}$ making up a thread,
and write~$\ethread$ when~$n = 0$.
Such a sequence implies that~$g$ contains the
edges~$(u_1, u_2), \dots (u_{n - 1}, u_n)$.
We will refer to such edges as \term{continuation edges}.
For a thread with priority~$\prio$ and vertices~$\uthread$, we
write~$\gthread{a}{\prio}{\uthread} \in \gthreads$.
We write~$\uprio{\graph}{u}$ to refer to the priority of the thread containing
vertex~$u$ in~$\graph$.

The remaining three components are sets of edges.
The set~$\spawns$ contains \term{{\create} edges}~$(u, a)$ indicating
that vertex~$u$ creates thread~$a$.
It is shorthand for~$(u, s)$ where~$s$ is the first vertex of~$a$.
The set~$\syncs$ contains \term{{\touch} edges}~$(a, u)$ indicating
that vertex~$u$ touches thread~$a$.
It is shorthand for~$(t, u)$ where~$t$ is the last vertex of~$a$.
Finally, the set~$\reads$ contains weak edges.

\subsection{Well-Formedness and Response Time}
Our goal is to bound the \term{response time}~$\resptimeof{a}$ of a thread~$a$
in a DAG.
If~$\gthread{a}{\prio}{s \tscomp{} \dots \tscomp{} t} \in \graph$,
for a particular schedule of~$\graph$,
we define~$\resptimeof{a}$ to be the number of time steps between when~$s$
becomes ready and when~$t$ is executed, inclusive.

Intuitively, in a well-designed program and an appropriate schedule,
if thread~$a$ has priority~$\prio$, its response time should depend only
on parts of the graph that may happen in parallel with~$a$ (i.e. are not
ancestors or descendants of~$a$) and have priority not less than~$\prio$.
This is known as the
\term{competitor work}~$\prioworkof{\compwork{}{a}}{\psnlt{\prio}}$
of a thread~$a$ and is defined formally:
\[
\prioworkof{\compwork{}{a}}{\psnlt{\prio}} \defeq
|\{u \in \graph \mid \nanc{u}{s} \land \nanc{t}{u} \land
\pnlt{\uprio{\graph}{u}}{\prio}\}|\]
We must also define a metric corresponding to the critical path of~$a$.
We will call this metric the $a$-span, because it corresponds to the
traditional notion of span in a parallel cost DAG, but we will defer
its formal definition for now, because
we will need other definitions first.

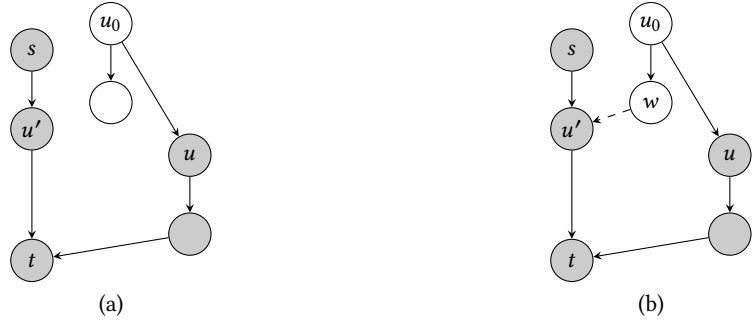
\begin{figure}
\begin{minipage}{0.47\columnwidth}
\centering
\begin{tikzpicture}[scale=0.70]
  \tikzstyle{n} = [draw,circle,fill=white,inner sep = 2pt,minimum width=16pt];
  \tikzstyle{pn} = [draw,circle,fill=black!20,inner sep = 2pt,minimum width=16pt];
  \tikzstyle{edge} = [draw,-stealth];
  \tikzstyle{wedge} = [draw,dashed,-stealth];

  \node[pn] (s) at (0, 1) {$s$};
  \node[n] (u0) at (1.5, 1.5) {$u_0$};
  \node[n] (u0') at (1.5, 0) {};
  \node[pn] (u) at (3, -1) {$u$};
  \node[pn] (u') at (3, -2.5) {};
  \node[pn] (s') at (0, -0.5) {$u'$};
  \node[pn] (t) at (0, -3) {$t$};

  \path[edge] (s) -- (s');
  \path[edge] (s') -- (t);
  \path[edge] (u0) -- (u);
  \path[edge] (u) -- (u');
  \path[edge] (u') -- (t);
  \path[edge] (u0) -- (u0');
\end{tikzpicture}

(a)
\end{minipage}%
\begin{minipage}{0.48\columnwidth}
\centering
\begin{tikzpicture}[scale=0.70]
  \tikzstyle{n} = [draw,circle,fill=white,inner sep = 2pt,minimum width=16pt];
  \tikzstyle{pn} = [draw,circle,fill=black!20,inner sep = 2pt,minimum width=16pt];
  \tikzstyle{edge} = [draw,-stealth];
  \tikzstyle{wedge} = [draw,dashed,-stealth];

  \node[pn] (s) at (0, 1) {$s$};
  \node[n] (u0) at (1.5, 1.5) {$u_0$};
  \node[n] (u0') at (1.5, 0) {$w$};
  \node[pn] (u) at (3, -1) {$u$};
  \node[pn] (u') at (3, -2.5) {};
  \node[pn] (s') at (0, -0.5) {$u'$};
  \node[pn] (t) at (0, -3) {$t$};

  \path[edge] (s) -- (s');
  \path[edge] (s') -- (t);
  \path[edge] (u0) -- (u);
  \path[edge] (u) -- (u');
  \path[edge] (u') -- (t);
  \path[edge] (u0) -- (u0');
  \path[wedge] (u0') -- (s');
\end{tikzpicture}

(b)
\end{minipage}
\caption{(a) a DAG that is not well-formed because of the strong path from~$u_0$
to~$t$ (b) a well-formed version of the DAG with a weak path from~$u_0$ to~$t$.}
\label{fig:well-formed}
\iffull\else \vspace{-5mm} \fi
\end{figure}

Bounding the response time of~$a$ in terms of only the competitor work
and~$a$-span is not possible for all DAGs: if~$a$ depends on lower-priority
code along its critical path, this code must be included in the response time
of~$a$.
This situation essentially corresponds to the well-known idea of a
\term{priority inversion}.
Our response time bound guarantees efficient scheduling of any DAG that
is \term{well-formed}, that is,
free of this type of priority inversion.
Well-formedness must, at a minimum, require that no {\touch} edges go from
lower- to higher-priority threads.
This requirement is formalized in the first bullet point of
Definition~\ref{def:well-formed}.
There is another, more subtle, way in which priority inversions could arise.
Consider the DAG in Figure~\ref{fig:well-formed}(a), in which shaded
vertices represent high-priority work.
Although no {\touch} edges violate the first requirement of the definition,
it would be possible, in a prompt schedule of the DAG, for high-priority
vertex~$t$ to be delayed indefinitely waiting for low-priority vertex~$u_0$
to execute due to the chain of strong dependences through~$u$.
Note that the problem is not that~$u$ depends on a lower-priority vertex---as
this is a {\create} edge, such a dependence is allowed.
The issue is that~$u$'s thread is then {\touch}ed by~$t$ with no other
dependence relation between~$u_0$ and~$t$.
The second bullet point of Definition~\ref{def:well-formed} requires that,
in such a situation, this dependence be mitigated by, e.g., the weak edge
added in Figure~\ref{fig:well-formed}(b).

We note that this second requirement actually places no additional restrictions
on programs.
DAGs such as the one in Figure~\ref{fig:well-formed}(a) could not arise from
real programs because in order for~$t$ to {\touch}~$u$'s thread, it must have
access to its thread handle, which will have been returned by the
{\create} call represented by~$u_0$.
This thread handle must be propagated to~$t$ through a chain of dependences
including at least one dependence through memory effects.
There must therefore be a weak path from~$u_0$ to~$t$, as in
DAG~\ref{fig:well-formed}(b), which reflects a write ($w$) of the thread handle
followed by a read ($u'$).

Definition~\ref{def:well-formed} formalizes the above intuitions.

\begin{defn}\label{def:well-formed} A DAG~$\graph = \dagq{\gthreads}{\spawns}{\syncs}{\reads}$ is
  \term{well-formed} if for all
  threads~$\gthread{a}{\prio}{s \tscomp{} \dots \tscomp{} t} \in \gthreads$,
  \begin{closeitemize}
  \item For all~$u \in \graph$, if~$\sanc{u}{t}$ and~$\nanc{u}{s}$,
    then~$\ple{\prio}{\uprio{\graph}{u}}$.
  \item For all strong edges~$(u_0, u)$ such that~$\sanc{u}{t}$
    and~$\nanc{u_0}{s}$
    and~$\nple{\uprio{\graph}{u}}{\uprio{\graph}{u_0}}$,
    there exists~$u'$ such that~$\wanc{u_0}{\sanc{u'}{t}}$
    and~$\nanc{u}{u'}$.
  \end{closeitemize}
\end{defn}

To a first approximation, we may define the~$a$-span of a
thread~$s \tscomp{} \dots \tscomp{} t$ as the longest path ending at~$t$
consisting of non-ancestors of~$s$ (i.e., the longest chain of vertices that
might delay the completion of~$a$).
In the presence of weak edges, however, the definition is not so simple.
Consider the DAG on the left of Figure~\ref{fig:strengthen}, in which
shaded nodes are high-priority.
Under the above definition, the~$a$-span includes low-priority node~$u_0$,
but in any admissible schedule,~$u'$ runs after~$u_0$, so~$u_0$ is not
actually on the critical path.
We thus transform the DAG into one, the \term{strengthening}, that
reflects this implicit dependence.

\begin{figure}
\begin{minipage}{0.47\columnwidth}
\centering
\begin{tikzpicture}[scale=0.70]
  \tikzstyle{n} = [draw,circle,fill=white,inner sep = 2pt,minimum width=16pt];
  \tikzstyle{pn} = [draw,circle,fill=black!20,inner sep = 2pt,minimum width=16pt];
  \tikzstyle{edge} = [draw,-stealth];
  \tikzstyle{wedge} = [draw,dashed,-stealth];

  \node[pn] (s) at (0, 1) {$s$};
  \node[n] (u0) at (1.5, 1.5) {$u_0$};
  \node[n] (u0') at (1.5, 0) {};
  \node[pn] (u) at (3, -1) {$u$};
  \node[pn] (u') at (3, -2.5) {};
  \node[pn] (s') at (0, -0.5) {$u'$};
  \node[pn] (t) at (0, -3) {$t$};

  \path[edge] (s) -- (s');
  \path[edge] (s') -- (t);
  \path[edge] (u0) -- (u);
  \path[edge] (u) -- (u');
  \path[edge] (u') -- (t);
  \path[edge] (u0) -- (u0');
  \path[wedge] (u0') -- (s');
\end{tikzpicture}

(a)
\end{minipage}%
\begin{minipage}{0.48\columnwidth}
\centering
\begin{tikzpicture}[scale=0.70]
  \tikzstyle{n} = [draw,circle,fill=white,inner sep = 2pt,minimum width=16pt];
  \tikzstyle{pn} = [draw,circle,fill=black!20,inner sep = 2pt,minimum width=16pt];
  \tikzstyle{edge} = [draw,-stealth];
  \tikzstyle{wedge} = [draw,dashed,-stealth];

  \node[pn] (s) at (0, 1) {$s$};
  \node[n] (u0) at (1.5, 1.5) {$u_0$};
  \node[n] (u0') at (1.5, 0) {};
  \node[pn] (u) at (3, -1) {$u$};
  \node[pn] (u') at (3, -2.5) {};
  \node[pn] (s') at (0, -0.5) {$u'$};
  \node[pn] (t) at (0, -3) {$t$};

  \path[edge] (s) -- (s');
  \path[edge] (s') -- (t);
  \path[edge] (u) -- (u');
  \path[edge] (u') -- (t);
  \path[edge] (u0) -- (u0');
  \path[edge] (s') -- (u);
\end{tikzpicture}

(b)
\end{minipage}
\caption{(a) a DAG; (b) its strengthening}
\label{fig:strengthen}
\iffull\else \vspace{-5mm} \fi
\end{figure}
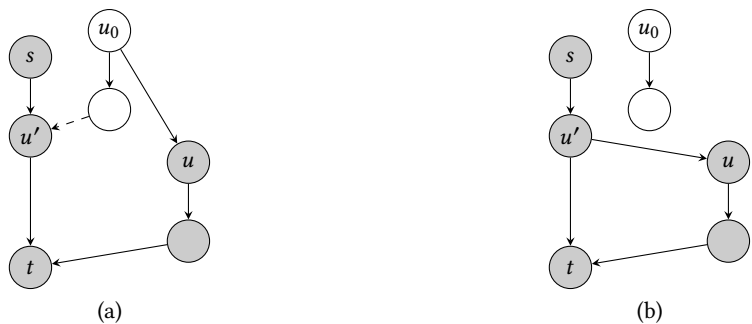

\begin{defn}
  Let~$\graph$ be a well-formed DAG with a thread
  $\gthread{a}{\prio_a}{s \tscomp{} \dots \tscomp{} t}$.
  We derive the $a$-\term{strengthening},
  written~$\strengthen{a}{\graph}$, from~$\graph$ as follows.
  For every strong edge~$(u_0, u)$ such that~$\sanc{u}{t}$
  and~$\nple{\uprio{\graph}{u}}{\uprio{\graph}{u_0}}$
  and~$\nanc{u}{s}$,
  \begin{closeitemize}
  \item Remove the edge~$(u_0, u)$.
  \item Let~$u' \in \graph$ such that~$\sanc{u'}{t}$ and~$\wanc{u_0}{u'}$.
    If~$\nanc{u'}{s}$, then add the edge~$(u', u)$ in place of the
    weak edge between~$u_0$ and~$u$.
  \end{closeitemize}
\end{defn}

%Note that~$\strengthen{a}{\graph}$ is guaranteed to exist and to be a DAG
%because the definition of well-formedness guarantees the existence of
%such a~$u'$ and that~$\nanc{u}{u'}$.
The strengthening of the example DAG is shown in the right side of the
figure.
For a thread~$\gthread{a}{\prio}{s \tscomp{} \dots \tscomp{} t} \in \graph$,
we define the~$a$-span, written~$\longsp{\compwork{}{a}}{a}$, to be the
length of the longest path in~$\strengthen{a}{\graph}$ ending at~$t$ consisting
only of vertices that are not ancestors of~$s$.
More generally, we write~$\longsp{V}{a}$ to be the
length of the longest path in~$\strengthen{a}{\graph}$ ending at~$t$ consisting
only of vertices in~$V$.
Intuitively, the~$a$-span corresponds to the critical path of~$a$ because, in
a valid and admissible schedule,
it is possible that all of the vertices along this path may need to be executed
sequentially while~$a$ is being executed.
%
%The strong edges added to replace weak edges reflect the worst-case admissible
%scheduling of the relevant threads: using the terminology from the definition
%above, thread~$c$ might become ready at the same time that~$u$ executes (it
%cannot become ready later since this would imply that thread~$b$ finishes
%after~$u$, violating admissibility).
\ifproofs
Lemma~\ref{lem:strengthen-ready} formalizes this intuition.

\begin{lemma}\label{lem:strengthen-ready}
  If, at any step of an admissible schedule, a vertex~$u$ is ready
  in~$\strengthen{a}{\graph}$, then it is ready in~$\graph$.
\end{lemma}
\begin{proof}
  For this to not hold, it would need to be the case that there exists
  a strong edge~$(u_0, u) \in \graph$ such that~$u_0$ has not been executed
  and~$(u_0, u) \not\in \strengthen{a}{\graph}$.
  Then~$(u_0, u)$ has been replaced
  in~$\strengthen{a}{\graph}$ with an edge~$(u', u)$ where,
  by assumption,~$u'$
  has been executed (because~$u$ is ready in~$\strengthen{a}{\graph}$).
  By construction,~$\wanc{u_0}{u'}$.
  Because the schedule is admissible,~$u_0$ has executed and~$u$ is ready
  in~$\graph$.
\end{proof}

Lemma~\ref{lem:strong-path} is another auxiliary result about
$a$-strengthening that will be necessary to prove the bound.
It shows that, for a thread~$a$ ending at~$t$,
\emph{any strong path} ending at~$t$ in the~$a$-strengthening
starts at a high-priority vertex
(in a well-formed but not strengthened DAG, this is true for strong ancestors
of~$t$---the strengthening essentially removes weak ancestors of~$t$).

\begin{lemma}\label{lem:strong-path}
  Let~$\gthread{a}{\prio}{s \tscomp{} \dots \tscomp{} t} \in \graph$.
  If there exists a strong path from~$u$ to~$t$ in~$\strengthen{a}{\graph}$
  and~$\nanc{u}{s}$ in~$\strengthen{a}{\graph}$,
  then~$\ple{\prio}{\uprio{\graph}{u}}$.
\end{lemma}
\begin{proof}
  %% Suppose that~$\anc{u}{s}$ in~$\graph$.
  %% %
  %% Then there must be an edge removed from~$\graph$ on a path from~$u$ to~$s$.
  %% %
  %% But the construction of~$\strengthen{a}{\graph}$ does not remove any
  %% such edges, so~$\nanc{u}{s}$ in~$\graph$.
  %% %
  %% If it is also the case that~$\sanc{u}{t}$ in~$\graph$, then the result
  %% follows from well-formedness of~$\graph$.
  %% %
  %% Suppose~$\nanc{u}{t}$ in~$\graph$.
  %% %
  %% Then, in constructing~$\strengthen{a}{\graph}$, we must have added a strong
  %% edge~$(u_1, u_2)$ such that~$\sanc{u_2}{t}$ and~$\nanc{u_1}{t}$ in~$\graph$.
  %% %
  %% But the construction process will never add such an edge.
  %% %
  If~$\nple{\prio}{\uprio{\graph}{u}}$,
  there must be some strong edge~$(u_1, u_2)$ on the path from~$u$ to~$t$
  in~$\strengthen{a}{\graph}$ such
  that~$\nple{\uprio{\graph}{u_2}}{\uprio{\graph}{u_1}}$.
  But all such edges are removed in the construction
  of~$\strengthen{a}{\graph}$.
\end{proof}
%\else
%\vspace{0.75in}

%\textbf{TODO: Add example giving intuition for a-span}

%\vspace{0.75in}
\fi

Theorem~\ref{thm:gen-brent} gives a bound on the response times of threads
in admissible, prompt schedules of well-formed DAGs.
The intuitive explanation of the bound also gives a sketch of the proof:
at every time step, such a schedule is doing one of two types of work:
(1) executing~$P$ vertices of competitor work or
(2) executing all available vertices on the~$a$-span.
The amount of work of type (1) to be done is bounded by the competitor work
divided by~$P$.
Work of type (2) can only be done during~$\longsp{\compwork{}{a}}{a}$
time steps, during which~$P-1$ of the~$P$ cores might be idle.
Adding these amounts of work together gives the bound on response time.

\begin{theorem}\label{thm:gen-brent}
  Let~$\graph$ be a well-formed DAG
  and let~$a$ be a thread of priority~$\prio$ in~$\graph$.
  For any admissible prompt schedule on~$P$ processing cores,
  %and any~$\ple{\prio'}{\prio}$,
  \[
  \resptimeof{a} \leq %\frac{1}{\fc(\psnlt{\prio'})}
  \frac{1}{P}\left[
  \prioworkof{\compwork{}{a}}{\psnlt{\prio}} +
      (P-1)\longsp{\compwork{}{a}}{a}\right]
      \]
\end{theorem}
%% \noindent{\em Proof intuition.}
%% %
%% This proof and several others in this chapter use a ``bucket-and-token''
%% analogy, which is fairly common in proofs of this form
%% (e.g.,~\citep{abp-multi-01}), to account for the actions of cores at each
%% time step.
%% %
%% We visualize a core at a time step by a token that each core spends
%% at each time step.
%% %
%% We divide processor-time-steps, and therefore tokens,
%% into two categories or ``buckets'', depending on whether or not
%% the core is busy with high-priority work at that time step.
%% %
%% We can bound the number of tokens in each category individually, and then add
%% the quantities together to find the total number of tokens spent.
%% %
%% Because~$P$ tokens are spent per time step, this in turn gives the number of
%% time steps.
\ifproofs
\begin{proof}
  Let~$s$ and~$t$ be the first and last vertices of~$a$, respectively.
  Consider the portion of the schedule from the step in which~$s$ becomes ready
  (exclusive) to the step in which~$t$ is executed (inclusive).
  For each core at each step, place a token in one of two buckets.
  If the core is working on a vertex of a priority not less than~$\prio$,
  place a token in the ``high'' bucket;
  otherwise, place a token in the ``low'' bucket.
  Because~$P$ tokens are placed per step,
  we have $\resptimeof{a} = \frac{1}{P}(B_l + B_h)$,
  where~$B_l$ and $B_h$
  are the number of tokens in the respective buckets
  after~$t$ is executed.

  Each token in~$B_h$ corresponds to work done at priority~$\not\prec\prio$, and
  thus $B_h \leq \prioworkof{\compwork{}{a}}{\psnlt{\prio}}$, so
  $\resptimeof{a} \leq %\frac{1}{\Sigma}\left(
  \frac{\prioworkof{\graph}{\psnlt{\prio}}}{P}
  + \frac{B_l}{P}$.
  We now need only bound~$B_l$ by~$(P-1)\cdot \longsp{\compwork{}{a}}{a}$.

  Let step 0 be the first step after~$s$ is ready, and let~$\exec{j}$ be the
  set of vertices that have been executed at the start of step~$j$.
  Consider a step~$j$ in which a token is added to~$B_l$.
  Consider a strong path in~$\strengthen{a}{\graph}$ ending at~$t$ consisting of
  vertices in~$\graph \setminus \exec{j}$.
  Any such path, by definition, begins at a vertex that is ready
  in~$\strengthen{a}{\graph}$ at the start of step~$j$.
  By Lemma~\ref{lem:strengthen-ready}, this vertex is also ready in~$\graph$,
  and by Lemma~\ref{lem:strong-path},
  it has priority greater than or equal to~$\prio$.
  By promptness, this vertex must therefore be executed during step~$j$.
  Thus, the length of the path decreases by 1 and so
  $\longsp{\graph \setminus \exec{j+1}}{a} =
  \longsp{\graph \setminus \exec{j}}{a} - 1$.
  Therefore, the maximum number of steps in which~$B_l$ increases
  is~$\longsp{\graph \setminus \exec{0}}{a}$, and
  at most~$P-1$ cores are idle while~$a$ is active,
  so~$B_l \leq (P-1)\cdot \longsp{\graph \setminus \exec{0}}{a}$.
  Because~$\exec{0}$ contains all ancestors of~$s$,
  any path excluding vertices
  in~$\exec{0}$ also excludes all ancestors of~$s$, and therefore
  $\longsp{\graph \setminus \exec{0}}{a} \leq
  \longsp{\compwork{}{a}}{t}$, so
  $B_l \leq (P-1)\cdot \longsp{\compwork{}{a}}{a}$.
\end{proof}
\else
\fi

\secput{lang}{Type System for Responsiveness}
We describe a type system that can be used to ensure that a
program results in a well-formed cost graph, by way of a core
calculus {\calcname}, which extends~$\lambda^4$~\citep{mah-priorities-2018},
with the key addition of mutable references (memory locations).
Section~\ref{sec:calculus} presents the calculus and type system.
Section~\ref{sec:cost} equips {\calcname} with a \term{cost semantics} that
evaluates a {\calcname} program to produce a cost graph of the form
described in Section~\ref{sec:dag}.
We prove that, for a well-typed program,
the resulting graph is well-formed, and thus the program is free of priority
inversions.

\subsection{The {\calcname} Core Calculus}\label{sec:calculus}
The syntax of {\calcname} is shown in Figure~\ref{fig:syn},
in A-normal form (for most expressions, any subexpressions that are not under
binders are values; computations can be sequenced using let-bindings).
We differentiate between
\term{expressions}, language constructs that do not depend on
the state of memory or threads, and \term{commands}, which do.
%
%% We differentiate pure from impure constructs in the syntax.
%% %
%% We refer to pure language constructs as {\em expressions}.
%% %
%% Impure operations (those that relate to threads and mutation)
%% are separated into a monadic
%% layer; we call these operations {\em commands}.
%
% Language constructs that are {\em pure} in that they do not depend on the
% state of memory or threads are contained in the {\em expression} layer.
%
%% This separation is made primarily so that the static semantics for expressions
%% need not refer to priorities (thus reducing the need for priority annotations)
%% and that the dynamic semantics
%% for expressions need not refer to global state, heavily
%% simplifying the technical presentation in the next section.
%
%Thread operations and memory operations are conceptually orthogonal and are
%included together in the command layer simply because they both reference
%the global execution state.
%
%Memory operations could be moved into the expression layer with no fundamental
%difficulties beyond the additional complication in the presentation of the
%dynamics.

\begin{figure}
\langfigsize
\[
\begin{array}{l l r l }
  %\mathit{Priorities} & \prio & \bnfdef &
  %\prioc \bnfalt  \vprio
  %\\

  \mathit{Constraints} & \cons & \bnfdef &
  \ple{\prio}{\prio} \bnfalt  \cconj{\cons}{\cons}
  \\

  \mathit{Types} & \tau & \bnfdef &
  \kwunit
  \bnfalt  \kwnat
  \bnfalt  \kwarr{\tau}{\tau}
  \bnfalt  \kwprod{\tau}{\tau}
  \bnfalt  \kwsum{\tau}{\tau}\\
  & & \bnfalt & \kwreft{\tau}
  \bnfalt  \kwat{\tau}{\prio}
  \bnfalt  \kwcmdt{\tau}{\prio}\\
  %\bnfalt  \kwforall{\vprio}{\cons}{\tau}\\

  \mathit{Values} & v & \bnfdef &
  x
  \bnfalt \kwtriv
  \bnfalt  \kwnumeral{n}
  \bnfalt  \kwfun{x}{e}
  \bnfalt  \kwepair{v}{v}
  \bnfalt  \kweinl{v}
  \bnfalt  \kweinr{v}\\
  & & \bnfalt & \kwref{\kwassn}
  \bnfalt  \kwtid{a}
  \bnfalt  \kwcmd{\prio}{\cmd}
  %\bnfalt  \kwwlam{\vprio}{\cons}{e}
  \\

  \mathit{Expressions} & e & \bnfdef &
  v
  \bnfalt  \kwlet{x}{e}{e}
  \bnfalt  \kwifz{v}{e}{x}{e}\\
  & & \bnfalt & \kwapply{v}{v}
  \bnfalt  \kwfst{v}
  \bnfalt  \kwsnd{v}\\
%  & & \bnfalt & \kweinl{v} \bnfalt
%  & & \bnfalt & \kweinr{v} \bnfalt
  & & \bnfalt & \kwcase{v}{x}{e}{y}{e}
%  & & \bnfalt & \kwoutput{v} \bnfalt
%  & & \bnfalt & \kwinput \bnfalt
 % \bnfalt & \kwwapp{v}{\prio}
  \bnfalt  \kwfix{x}{\tau}{e}
  \\

  \mathit{Commands} & \cmd & \bnfdef &
  \kwspawn{\prio}{\tau}{\cmd}
  \bnfalt  \kwsync{e}\\
  & & \bnfalt & \kwdcl{\tau}{\kwassn}{e}{\cmd}\\
  & & \bnfalt & \kwderef{e}
  \bnfalt  \kwassign{e}{e}
  \bnfalt  \kwbind{e}{x}{\cmd}
  \bnfalt  \kwret{e}
\end{array}
\]
\iffull\else \vspace{-2mm} \fi
\caption{Syntax of {\calcname}}
\label{fig:syn}
\iffull\else \vspace{-3mm} \fi
\end{figure}

%% Priorities may be either priority constants~$\prioc$, which are drawn from
%% a given partially-ordered set~$\worlds$, or variables~$\vprio$.
%% %
%% %As in the previous section,
%% %the set~$\worlds$ is equipped with a partial ordering~$\preceq$
%% %and its related strict order~$\prec$.
%% %
%% A priority constraint is a conjunction of priority
%% relations~$\ple{\prio_1}{\prio_2}$.

The non-standard types of~{\calcname} are
a type~$\kwreft{\tau}$ indicating references
to memory locations holding values of type~$\tau$;
a type~$\kwat{\tau}{\prio}$ representing
handles to running threads of type~$\tau$ at priority~$\prio$ and
a type~$\kwcmdt{\tau}{\prio}$ representing encapsulated commands which run
at priority~$\prio$ and have return type~$\tau$.
Priorities are drawn from a given fixed (partially ordered) set~$\worlds$.
%The calculus also admits for priority-polymorphic
%types~$\kwforall{\vprio}{\cons}{\tau}$.
%
%The variable~$\vprio$ must be instantiated with a priority meeting the
%constraint~$\cons$.

%To streamline the presentation in this and the next section,
%We present the expression layer 
%
%% Non-values subexpressions must be evaluated to a value and bound to a
%% variable using let-bindings.
%% %
%% For example, (allowing ourselves the use of the infix binary operation~$+$)
%% the expression~$1 + 2 + 3 + 4$
%% could be written
%% \[\kwlet{x}{1 + 2}{\kwlet{y}{3 + 4}{x + y}}\]
%% %
%% This A-normal form presentation therefore loses no generality and comes with
%% the benefit that the order of evaluation of subexpressions is explicit.
%
The novel values of the calculus are
references~$\kwref{\kwassn}$, which allow access to
a memory location~$\kwassn$;
thread handles~$\kwtid{a}$, which reference a running
thread referred to by~$a$; and $\kwcmd{\prio}{\cmd}$, which encapsulates
the command~$\cmd$
at priority~$\prio$.
%and priority-polymorphic
%abstractions~$\kwwlam{\vprio}{\cons}{e}$.
%
%These abstractions are instantiated with a priority~$\prio$ using
%the priority application
%expression~$\kwwapp{(\kwwlam{\vprio}{\cons}{e})}{\prio}$.
%
The expression layer is otherwise standard.

Commands include operations to manipulate threads and state,
including commands to create and touch threads\footnote{while the
syntax for \create and \touch is drawn from the fact that our threading model
is based on \textbf{f}utures, we simply use the term ``threads'' to refer to
running asynchronous threads of control and ``thread handles'' to refer to
the first-class values that refer to threads. This avoids some terminological
confusion frequently associated with futures.}
%
%% The command $\kwspawn{\prio}{\tau}{\cmd}$ creates a new thread at
%% priority~$\prio$ to run the command~$\cmd$.
%% %
%% %For simplicity, the command is also annotated with the thread's return
%% %type~$\tau$.
%% %
%% The fcreate command returns a handle to the new thread, which may be used to
%% touch the thread using the command~$\kwsync{e}$, which evaluates~$e$
%% to a thread handle~$\kwtid{a}$ and then waits for the thread~$a$ to complete
%% before returning its return value.
%
The command~$\kwdcl{\tau}{\kwassn}{e}{\cmd}$ declares a new mutable
memory location~$\kwassn$, initialized with the expression~$e$, in the
scope of~$\cmd$.
The read command~$\kwderef{e}$ evaluates~$e$ to a reference~$\kwref{\kwassn}$
and returns the current contents of~$\kwassn$.
The assignment command~$\kwassign{e_1}{e_2}$ evaluates~$e_1$ to a
reference~$\kwref{\kwassn}$ and writes the value of~$e_2$ to~$\kwassn$;
the command also returns the new value.

Commands are sequenced with an operator~$\kwbind{e}{x}{\cmd}$, which
evaluates~$e$ to an encapsulated command, executes the command, binds its
return value to~$x$ and continues as~$\cmd$.
Expressions may be embedded into the command layer using the
command~$\kwret{e}$ which evaluates~$e$ and returns its value.
These commands may be thought of as the monadic bind and return
operators, respectively.

%\subsection{Static Semantics}\label{sec:lang-statics}

  \iffull
\begin{figure}
\langfigsize
    \centering
  \def \MathparLineskip {\lineskip=\mylineskip}
\begin{mathpar}

\Rule{var}
{
  \strut
}
{
  \etyped{\sig}{\ctx, \hastype{x}{\tau}}{x}{\tau}
}
\and
\Rule{\kwunit I}
{
\strut
}
{
\etyped{\sig}{\ctx}{\kwtriv}{\kwunit}
}
\and
\Rule{Tid}
      {
        \strut
  %\etyped{\sig}{\ctx}{v}{\tau}\\
  %\ple{\prio}{\prio'}\\
  %\tau\mobile
}
{
  \etyped{\sig, \sigtype{a}{\tau}{\prio'}}{\ctx}{\kwtid{a}}
        {\kwat{\tau}{\prio'}}
}
\and
\Rule{Ref}
     {
       \strut
     }
     {
       \etyped{\sig,\sigrtype{\kwassn}{\tau}}{\ctx}{\kwref{\kwassn}}
              {\kwreft{\tau}}
     }
\and
\Rule{\kwnat I}
{
\strut
}
{
\etyped{\sig}{\ctx}{\kwn}{\kwnat}
}
\and
\Rule{\kwnat E}
{
  \etyped{\sig}{\ctx}{v}{\kwnat}\\
  \etyped{\sig}{\ctx}{e_1}{\tau}\\
  \etyped{\sig}{\ctx,\hastype{x}{\kwnat}}{e_2}{\tau}
}
{
  \etyped{\sig}{\ctx}{\kwifz{v}{e_1}{x}{e_2}}{\tau}
}
\and
\Rule{\arrsym I}
{
  \etyped{\sig}
        {\ctx, \hastype{x}{\tau_1}}{e}{\tau_2}
}
{
  \etyped{\sig}{\ctx}{\kwfun{x}{e}}{\kwarr{\tau_1}{\tau_2}}
}
\and
\Rule{\arrsym E}
{
\etyped{\sig}{\ctx}{v_1}{\kwarr{\tau_1}{\tau_2}}\\
\etyped{\sig}{\ctx}{v_2}{\tau_1}
}
{
\etyped{\sig}{\ctx}{\kwapply{v_1}{v_2}}{\tau_2}
}
\and
\Rule{\prodsym $I$}
{
  \etyped{\sig}{\ctx}{v_1}{\tau_1}\\
  \etyped{\sig}{\ctx}{v_2}{\tau_2}
}
{
  \etyped{\sig}{\ctx}{\kwepair{v_1}{v_2}}{\kwprod{\tau_1}{\tau_2}}
}
\and
\Rule{\prodsym $E_1$}
{
  \etyped{\sig}{\ctx}{v}{\kwprod{\tau_1}{\tau_2}}
}
{
  \etyped{\sig}{\ctx}{\kwfst{v}}{\tau_1}
}
\and
\Rule{\prodsym $E_2$}
{
  \etyped{\sig}{\ctx}{v}{\kwprod{\tau_1}{\tau_2}}
}
{
  \etyped{\sig}{\ctx}{\kwsnd{v}}{\tau_2}
}
\and
\Rule{\sumsym $I_1$}
{
  \etyped{\sig}{\ctx}{v}{\tau_1}
}
{
  \etyped{\sig}{\ctx}{\kweinl{v}}{\kwsum{\tau_1}{\tau_2}}
}
\and
\Rule{\sumsym $I_2$}
{
  \etyped{\sig}{\ctx}{v}{\tau_2}
}
{
  \etyped{\sig}{\ctx}{\kweinr{v}}{\kwsum{\tau_1}{\tau_2}}
}
\and
\Rule{\sumsym E}
{
  \etyped{\sig}{\ctx}{v}{\kwsum{\tau_1}{\tau_2}}\\
  \etyped{\sig}{\ctx,\hastype{x}{\tau_1}}{e_1}{\tau'}\\
  \etyped{\sig}{\ctx,\hastype{y}{\tau_2}}{e_2}{\tau'}
}
{
  \etyped{\sig}{\ctx}{\kwcase{v}{x}{e_1}{y}{e_2}}{\tau'}
}
\and
%% \Rule{Output}
%% {
%%   \etyped{\sig}{\ctx}{v}{\kwnat}
%% }
%% {
%%   \etyped{\sig}{\ctx}{\kwoutput{v}}{\kwunit}
%% }
%% \and
%% \Rule{Input}
%% {
%%   \strut
%% }
%% {
%%   \etyped{\sig}{\ctx}{\kwinput}{\kwnat}
%% }
%% \and
\Rule{\cmdsym I}
{
  \cmdtyped{\sig}{\ctx}{\cmd}{\tau}{\prio}
}
{
  \etyped{\sig}{\ctx}{\kwcmd{\prio}{\cmd}}{\kwcmdt{\tau}{\prio}}
}
\and
\Rule{\fasym I}
{
  \etyped{\sig}{\ctx, \vprio \isprio, \cons}{e}{\tau}
}
{
  \etyped{\sig}{\ctx}{\kwwlam{\vprio}{\cons}{e}}{\kwforall{\vprio}{\cons}{\tau}}
}
\and
\Rule{\fasym E}
{
  \etyped{\sig}{\ctx}{v}{\kwforall{\vprio}{\cons}{\tau}}\\
  \meetc{\ctx}{[\prio'/\vprio]\cons}
}
{
  \etyped{\sig}{\ctx}{\kwwapp{v}{\prio'}}{[\prio'/\vprio]\tau}
}
\and
\Rule{fix}
{
  \etyped{\sig}{\ctx, \hastype{x}{\tau}}{e}{\tau}
}
{
  \etyped{\sig}{\ctx}{\kwfix{x}{\tau}{e}}{\tau}
}
\and
\Rule{let}
{
  \etyped{\sig}{\ctx}{e_1}{\tau_1}\\
  \etyped{\sig}{\ctx, \hastype{x}{\tau_1}}{e_2}{\tau_2}
}
{
  \etyped{\sig}{\ctx}{\kwlet{x}{e_1}{e_2}}{\tau_2}
}
%% \Rule{\boxsym I}
%% {
%%   \etyped{\sig}{\ctx,\vprio \isprio}{e}{\kwcmdt{\tau}{\vprio}}
%% }
%% {
%%   \etyped{\sig}{\ctx}{\kwbox{\vprio}{e}}{\kwboxt{\tau}}
%% }
%% \and
%% \Rule{\diasym I}
%% {
%%   \etyped{\sig}{\ctx}{e}{\kwcmdt{\tau}{\prio}}
%% }
%% {
%%   \etyped{\sig}{\ctx}{\kwthere{\prio}{e}}{\kwdia{\tau}}
%% }
\end{mathpar}
\caption{Expression typing rules.}
\label{fig:statics}
\end{figure}
\fi

\begin{figure}
\langfigsize
  \centering
  \begin{mathpar}
    \iffull
\Rule{Bind}
{
  \etyped{\sig}{\ctx}{e}{\kwcmdt{\tau}{\prio}}\\\\
  \cmdtyped{\sig}{\ctx, \hastype{x}{\tau}}{\cmd}{\tau'}{\prio}
}
{
  \cmdtyped{\sig}{\ctx}{\kwbind{e}{x}{\cmd}}{\tau'}{\prio}
}
\fi
\and
\Rule{Create}
{
  \cmdtyped{\sig}{\ctx}{\cmd}{\tau}{\prio'}
}
{
  \cmdtyped{\sig}{\ctx}{\kwspawn{\prio'}{\tau}{\cmd}}{\kwat{\tau}{\prio'}}{\prio}
}
\and
\Rule{Touch}
{
  \etyped{\sig}{\ctx}{e}{\kwat{\tau}{\prio'}}\\
  \meetc[\worlds]{\ctx}{\ple{\prio}{\prio'}}
}
{
  \cmdtyped{\sig}{\ctx}{\kwsync{e}}{\tau}{\prio}
}
\and
\Rule{Dcl}
     {
       \etyped{\sig}{\ctx}{e}{\tau}\\
       \cmdtyped{\sig, \sigrtype{\kwassn}{\tau}}{\ctx}{\cmd}{\tau'}{\prio}
     }
     {
       \cmdtyped{\sig}{\ctx}{\kwdcl{\tau}{\kwassn}{e}{\cmd}}{\tau'}{\prio}
     }
\and
\Rule{Get}
     {
       \etyped{\sig}{\ctx}{e}{\kwreft{\tau}}
     }
     {
       \cmdtyped{\sig}{\ctx}{\kwderef{e}}{\tau}{\prio}
     }
\and
\Rule{Set}
     {
       \etyped{\sig}{\ctx}{e_1}{\kwreft{\tau}}\\
       \etyped{\sig}{\ctx}{e_2}{\tau}
     }
     {
       \cmdtyped{\sig}{\ctx}{\kwassign{e_1}{e_2}}{\tau}{\prio}
     }
%% \Rule{Poll}
%% {
%%   \etyped{\sig}{\ctx}{e}{\kwat{\tau}{\prio'}}
%% }
%% {
%%   \cmdtyped{\sig}{\ctx}{\kwpoll{e}}{\tau}{\prio}
%% }
%% \and
%% \Rule{\boxsym E}
%% {
%%   \etyped{\sig}{\ctx}{e}{\kwboxt{\tau}}\\
%%   \cmdtyped{\sig}{\ctx, \hastype{x}{\kwcmdt{\tau}{\prio}}}{\cmd}{\tau'}{\prio}
%% }
%% {
%%   \cmdtyped{\sig}{\ctx}{\kwletbox{x}{e}{\cmd}}{\tau'}{\prio}
%% }
%% \and
%% \Rule{\diasym E}
%% {
%%   \etyped{\sig}{\ctx}{e}{\kwdia{\tau}}\\
%%   \cmdtyped{\sig}{\ctx, \vprio \isprio, \hastype{x}{\kwcmdt{\tau}{\vprio}}}
%%            {\cmd}{\tau'}{\prio}
%% }
%% {
%%   \cmdtyped{\sig}{\ctx}{\kwletdia{\vprio}{x}{e}{\cmd}}{\tau'}{\prio}
     %% }
     \iffull
\and
\Rule{Ret}
{
  \etyped{\sig}{\ctx}{e}{\tau}
}
{
  \cmdtyped{\sig}{\ctx}{\kwret{e}}{\tau}{\prio}
}
\fi
%% \and
%% \Rule{\cmdsym E}
%% {
%%   \etyped{\sig}{\ctx}{e}{\kwcmdt{\tau}{\prio}}
%% }
%% {
%%   \cmdtyped{\sig}{\ctx}{\kwdo{e}}{\tau}{\prio}
%% }
%% \and
%% \Rule{Post}
%% {
%%   \etyped{\sig}{\ctx}{e}{\tau}\\
%%   \cmdtyped{\sig}{\ctx}{\cmd}{\tau}{\prio}
%% }
%% {
%%   \cmdtyped{\sig}{\ctx}{\kwpost{e}{\cmd}}{\tau}{\prio}
%% }
%\and
%
%\Rule{letm}
%{
%  \etyped{\sig}{\ctx}{e_1}{\tau_1}{\prio}\\
%  \tau_1 \mobile\\
%  \etyped{\sig,\hastype{x}{\tau_1}{\prio}}{\ctx}{e_2}{\tau_2}{\prio}
%}
%{
%  \etyped{\sig}{\ctx}{\kwletm{x}{e_1}{e_2}}{\tau_2}{\prio}
%}
%% \and
%% \infer
%% {
%% }
%% {
%%   \mtyped{(\worlds, \prel)}{\esig}{\emem}
%% }
%% \and
%% \infer
%% {
%%   \cmdtyped{\sig}{\ctxify{(\worlds, \prel)}}{\cmd}{\tau}{\prio}\\
%%   \mtyped{(\worlds, \prel)}{\sig}{\mem}
%% }
%% {
%%   \mtyped{(\worlds, \prel)}{\sig, \sigtype{a}{\tau}{\prio}}
%%          {\dthread{a}{\cmd} \mcp \mem}
%% }
\end{mathpar}
\iffull\else \vspace{-2mm} \fi
\caption{\iffull Command \else Selected command \fi typing rules.}
\label{fig:thread-statics}
\iffull\else \vspace{-3mm} \fi
\end{figure}

\iffull
\begin{figure}
\langfigsize
  \centering
  \begin{mathpar}
\Rule{hyp}
{
  \strut
}
{
  \meetc{\ctx, \ple{\prio_1}{\prio_2}}{\ple{\prio_1}{\prio_2}}
}
\and
\Rule{assume}
{
  \plt{\prioc_1}{\prioc_2} \in \worlds
}
{
  \meetc[\worlds]{\ctx}{\ple{\prioc_1}{\prioc_2}}
}
\and
\Rule{refl}
{
  \strut
}
{
  \meetc{\ctx}{\ple{\prio}{\prio}}
}
\and
\Rule{trans}
{
  \meetc{\ctx}{\ple{\prio_1}{\prio_2}}\\
  \meetc{\ctx}{\ple{\prio_2}{\prio_3}}
}
{
  \meetc{\ctx}{\ple{\prio_1}{\prio_3}}
}
\and
\Rule{conj}
{
  \meetc{\ctx}{\cons_1}\\
  \meetc{\ctx}{\cons_2}
}
{
  \meetc{\ctx}{\cconj{\cons_1}{\cons_2}}
}
\end{mathpar}
\caption{Constraint entailment}
\label{fig:constraint-statics}
\end{figure}
\fi

\iffull Figures~\ref{fig:statics}--\ref{fig:constraint-statics} show
the typing rules for {\calcname}.
The key concepts of the type system are those that deal with the types of
references.
\else Figure~\ref{fig:thread-statics} shows the key rules of the type
system for {\calcname}, namely the rules for threads and references.
Due to space constraints, we omit more standard features and present them
in the full version~\citep{arxiv}.
\fi

\iffull
Figure~\ref{fig:statics} gives the rules for
the typing judgment for expressions,~$\etyped[\worlds]{\sig}{\ctx}{e}{\tau}$.
Expressions do not manipulate memory, but the types of memory locations are
still necessary in order to type references~$\kwref{\kwassn}$.
We track these types in the signature~$\sig$, along with the return types and
priorities of thread symbols.
The elements of~$\sig$ are of two forms: the entry~$\sigtype{a}{\tau}{\prio}$
indicates that thread~$a$ runs at priority~$\prio$ and has return type~$\tau$,
and the entry~$\sigrtype{\kwassn}{\tau}$ indicates that memory
location~$\kwassn$ holds values of type~$\tau$.

The typing judgment is also parameterized by a partially-ordered set~$\worlds$
of priorities and a typing context~$\ctx$.
The context~$\ctx$, as usual, contains premises of the form~$\hastype{x}{\tau}$,
indicating that the variable~$x$ has type~$\tau$.
In addition,~$\ctx$ contains premises of the form~$\vprio\isprio$, introducing
the priority variable~$\vprio$, as well as priority constraints~$\cons$.
Under these assumptions, the judgment~$\etyped[\worlds]{\sig}{\ctx}{e}{\tau}$
indicates that expression~$e$ has type~$\tau$.
Note that expressions do not have priorities: because they do not create or
touch threads, expressions may be freely evaluated at any priority.

The non-standard expression typing rules are those that deal with
thread handles, references, priorities and commands.
A thread handle~$\kwtid{a}$ has type~$\kwat{\tau}{\prio'}$
if thread~$a$ has priority~$\prio'$ and return type~$\tau$
(rule~\rulename{Tid}).
A reference~$\kwref{\kwassn}$ has type~$\kwref{\tau}$ if~$\kwassn$ holds
values of type~$\tau$ (rule~\rulename{Ref}).
These are the only rules that inspect the signature~$\sig$; it is otherwise
simply threaded through the typing derivations.
The introduction form for priority-polymorphic
types,~$\kwwlam{\vprio}{\cons}{e}$
introduces the premise~$\vprio\isprio$ and the constraint~$\cons$ into the
context in order to type~$e$ (rule~\rulename{$\forall$ I}).
The corresponding elimination form applies such an abstraction to a
priority~$\prio'$ (rule~\rulename{$\forall$ E}).
This rule requires that~$\prio'$ meets the constraints imposed by~$\cons$
(the judgment~$\meetc{\ctx}{\cons}$ will be described below) and
substitutes~$\prio'$ for~$\vprio$ in the polymorphic type.
Finally, an encapsulated command~$\kwcmd{\prio}{\cmd}$ has
type~$\kwcmdt{\tau}{\prio}$ if~$\cmd$ is well-typed with priority~$\prio$
and type~$\tau$ under the command typing judgment, which is described below.
\fi

\iffull
The rules for typing commands are given in Figure~\ref{fig:thread-statics}.
These rules define the
judgment~$\cmdtyped[\worlds]{\sig}{\ctx}{\cmd}{\tau}{\prio}$.
\else
The command typing rules in the figure define the
judgment~$\cmdtyped[\worlds]{\sig}{\ctx}{\cmd}{\tau}{\prio}$.
The signature~$\sig$ tracks type information for threads and memory locations,
as well as the priorities of threads.
%, along with the return types and
%priorities of thread symbols.
%
%The elements of~$\sig$ are of two forms: the entry~$\sigtype{a}{\tau}{\prio}$
%indicates that thread~$a$ runs at priority~$\prio$ and has return type~$\tau$,
%and the entry~$\sigrtype{\kwassn}{\tau}$ indicates that memory
%location~$\kwassn$ holds values of type~$\tau$.
%
The typing judgment is also parameterized by a partially-ordered set~$\worlds$
of priorities and a typing context~$\ctx$.
The context~$\ctx$, as usual, contains premises of the form~$\hastype{x}{\tau}$,
indicating that the variable~$x$ has type~$\tau$.
\fi
%
%In addition,~$\ctx$ contains premises of the form~$\vprio\isprio$, introducing
%the priority variable~$\vprio$, as well as priority constraints~$\cons$.
%
In addition to the return type~$\tau$ of the command, the typing judgment
indicates that the command may run at priority~$\prio$.
The rules~\rulename{Create} and~\rulename{Touch} contain notable features
relating to priorities.
In particular,~\rulename{Touch} requires that~$e$ be a handle to a
thread running at priority~$\prio'$ and that this priority be higher than
or equal to the priority~$\prio$ of the current thread.
It is this requirement that prevents priority inversion.
The~\rulename{Create} rule requires that a command run in a new thread at
priority~$\prio'$ indeed be able to run at priority~$\prio'$.
Note, however, that the {\create} command itself may run at any priority; the
language does not enforce any priority relationship between a thread and its
parent.
We refer the reader to the presentation
of~$\lambda^4$~\citep{mah-priorities-2018} for a more thorough description of
these rules.

We describe the rules for allocating and accessing references in more detail.
Rule~\rulename{Dcl} types the initialization expression~$e$ at type~$\tau$
and introduces a new location~$\kwassn$ in typing~$\cmd$.
Rule~\rulename{Get} requires that its subexpression have reference type.
Rule~\rulename{Set} requires that~$e_1$ have type~$\kwreft{\tau}$ and
that~$e_2$ have type~$\tau$.
Note that this requires memory locations to have a consistent type throughout
execution.
The return type of an assignment to a~$\tau$ reference is~$\tau$.
All of these commands may type at any priority as state operations and
priorities are orthogonal.

The judgment~$\meetc[\worlds]{\ctx}{\cons}$ indicates that the premises
contained in~$\ctx$ entail the priority constraints~$\cons$.
\iffull
The rules for this judgment appear in Figure~\ref{fig:constraint-statics}.
A constraint is entailed by~$\ctx$ if it appears in~$\ctx$ (\rulename{hyp}),
is directly implied by the ordering on~$\worlds$ (\rulename{assume}),
follows from reflexivity (\rulename{refl}) or transitivity (\rulename{trans})
or is the conjunction of two entailed constraints (\rulename{conj}).
\else
The rules (which follow standard rules of logic) are found in the
full version~\citep{arxiv}
\fi

\ifproofs
The typing rules admit several forms of substitution.
%
%We may substitute expressions for variables in expressions and commands,
%and may substitute priorities for priority variables in
%expressions, commands and constraints.
%
Lemma~\ref{lem:subst} shows that all of these substitutions preserve typing.
The proof largely follows from that of the equivalent substitution lemma
  for~$\lambda^4$~\citep{mah-priorities-2018}.
\begin{lemma}[Substitution]\label{lem:subst}\strut\\[-1em]
  \begin{closeenum}
\item If $\etyped{\sig}{\ctx, \hastype{x}{\tau}}{e_1}{\tau'}$ and
  $\etyped{\sig}{\ctx}{e_2}{\tau}$, then
  $\etyped{\sig}{\ctx}{[e_2/x]e_1}{\tau'}$.
\item If $\cmdtyped{\sig}{\ctx, \hastype{x}{\tau}}{\cmd}{\tau'}{\prio}$ and
  $\etyped{\sig}{\ctx}{e}{\tau}$, then
  $\cmdtyped{\sig}{\ctx}{[e/x]\cmd}{\tau'}{\prio}$.
\item If $\etyped{\sig}{\ctx, \vprio \isprio}{e}{\tau}$,
  then $\etyped{\sig}{[\prio/\vprio]\ctx}{[\prio/\vprio]e}{[\prio/\vprio]\tau}$.
\item If $\cmdtyped{\sig}{\ctx, \vprio \isprio}{\cmd}{\tau}{\prio}$,
  then $\cmdtyped{\sig}{[\prio'/\vprio]\ctx}{[\prio'/\vprio]\cmd}
  {[\prio'/\vprio]\tau}{[\prio'/\vprio]\prio}$.
\item If $\meetc{\ctx, \vprio \isprio}{\cons}$, then
  $\meetc{[\prio/\vprio]\ctx}{[\prio/\vprio]\cons}$.
\end{closeenum}
\end{lemma}
%% \begin{proof}
%%   This proof largely follows from that of the equivalent substitution lemma
%%   for~$\lambda^4$~\citep{mah-priorities-2018}.
%%   %
%%   The new cases follow directly from induction.
%%   %% The only proof with nontrivial new cases is that of (2).
%%   %% %
%%   %% This proof is by induction on the derivation
%%   %% of~$\cmdtyped{\sig}{\ctx, \hastype{x}{\tau}}{\cmd}{\tau'}{\prio}$
%%   %% and the new cases are those for the commands dealing with state.
%%   %% \begin{itemize}
%%   %% \item \rulename{Dcl}.
%%   %%   Then~$\cmd = \kwdcl{\tau''}{\kwassn}{e'}{\cmd'}$
%%   %%   and~$\etyped{\sig}{\ctx,\hastype{x}{\tau}}{e'}{\tau'}$
%%   %%   and~$\cmdtyped{\sig,\sigrtype{\kwassn}{\tau''}}{\ctx,\hastype{x}{\tau}}
%%   %%   {\cmd'}{\tau'}{\prio}$.
%%   %%   By part (1), we have~$\etyped{\sig}{\ctx}{[e/x]e'}{\tau'}$.
%%   %%   By induction,~$\cmdtyped{\sig,\sigrtype{\kwassn}{\tau''}}{\ctx}
%%   %%   {[e/x]\cmd'}{\tau'}{\prio}$.
%%   %%   Apply rule~\rulename{Dcl}.
%%   %% \item \rulename{Get}.
%%   %% \item \rulename{Set}.
%%   %% \end{itemize}
%% \end{proof}
\fi

\subsection{Cost Semantics and Time Bounds}\label{sec:cost}
In this section, we equip~{\calcname} with a small-step dynamic semantics
that tracks two notions of cost.
First,
in a straightforward sense, the number of steps taken by the semantics to
execute a program gives an abstract measure of execution time.
Second, we equip the dynamic semantics to construct a cost graph for the
program that captures the parallelism opportunities in the execution, and also
uses weak edges to record happens-before relations as
described in Section~\ref{sec:dag}.
%

%% We use a small-step semantics rather than a big-step cost semantics (as is
%% more traditional for generating cost graphs from programs) because a big-step
%% semantics deliberately abstracts away execution order.
%% %
%% In the presence of potential data races, however, execution order is crucial
%% and must be captured in the DAG using weak edges to prove the soundness of
%% the type system.
%% %
%% Furthermore, unifying the construction of the DAG and the operational semantics
%% of {\calcname} into one semantics drastically simplifies the process of
%% proving that the operational semantics respects the predicted cost bounds.
%% %
%% We further discuss this design choice and its relation to prior work
%% in Section~\ref{sec:related}.

%% In Section~\ref{sec:dyn}, we present the dynamic semantics and show
%% type safety (progress and preservation) with respect to the type system
%% presented in the previous section.
%% %
%% We then, in Section~\ref{sec:cost-corr}
%% show that well-typed programs produce well-formed cost graphs, allowing
%% us to apply the cost bound of Theorem~\ref{thm:gen-brent} to
%% {\calcname} executions.

%\subsection{A Cost-Tracking Dynamic Semantics}\label{sec:dyn}
We present the dynamic semantics of~{\calcname}
as a stack-based parallel abstract machine that serves
as a rough model of the program's execution time on realistic
parallel hardware.
%
%As discussed above, the semantics also constructs a cost graph of the form
%described in Section~\ref{sec:dag}.
%
A {\em stack}~$\stack$ consists of a sequence of {\em stack frames}~$f$
(or is the empty stack,~$\estack$).
Each frame is a command or expression with a hole, written~$\shole$,
to be filled with the result of the next frame.
The stack thus represents the continuation of the current computation.
At each step, each thread active in the
machine is either executing a command or expression from the
top of the stack (``popping'') or returning a resulting value to the
stack (``pushing'').
These states are represented by~$\ssend{\stack}{e}$ and~$\sreturn{\stack}{v}$,
respectively, for expressions (and similar syntax with filled triangles for
commands).
The syntax of stack frames, stacks, and stack states is given in
Figure~\ref{fig:stack-syn}.

\begin{figure}
\footnotesize
  \[
  \begin{array}{l l r l}
    \textit{Frames} & f & \bnfdef & \kwlet{x}{\shole}{e} \bnfalt
    \kwbind{\shole}{x}{\cmd} \bnfalt
    \kwsync{\shole}\\
     & & \bnfalt &
    \kwdcl{\tau}{\kwassn}{\shole}{\cmd} \bnfalt
    \kwderef{\shole} \bnfalt
    \kwassign{\shole}{e}\\
    & & \bnfalt &
    \kwassign{v}{\shole} \bnfalt
    \kwret{\shole}\\

    \textit{Stacks} & \stack & \bnfdef & \estack \bnfalt \scp{\stack}{f}\\

    \textit{States} & \stackstate & \bnfdef & \ssend{\stack}{e} \bnfalt
    \sreturn{\stack}{v} \bnfalt
    \scsend{\stack}{\cmd} \bnfalt
    \screturn{\stack}{\kwret{v}}
  \end{array}
  \]
\iffull\else \vspace{-3mm} \fi
\caption{Stack, frame and state syntax.}
\label{fig:stack-syn}
\iffull\else \vspace{-3mm} \fi
\end{figure}

\begin{figure*}
\langfigsize
    \centering
  \def \MathparLineskip {\lineskip=\mylineskip}
  \begin{mathpar}
    \iffull
\Rule{D-Bind1}
     {u\fresh}
     {
       \lconfig{\mem}{\tp}{\cthread{a}{\prio}{\sig}
         {\scsend{\stack}{\kwbind{e}{x}{\cmd_2}}}}
       \mstep
       \rconfig{\cthread{a}{\prio}{\sig}
         {\ssend{\scp{\stack}{\kwbind{\shole}{x}{\cmd_2}}}{e}}}
               {\etp}{\esig}{\mem}{\tgraph{a}{\prio}{\sthread{u}}}
     }
\and
\Rule{D-Bind2}
     {u\fresh}
     {
       \phantom{\mstep}\lconfig{\mem}{\tp}{\cthread{a}{\prio}{\sig}
         {\sreturn{\scp{\stack}{\kwbind{\shole}{x}{\cmd_2}}}{\kwcmd{\prio}{\cmd}}}}
       \mstep
       \rconfig{\cthread{a}{\prio}{\sig}
         {\scsend{\scp{\stack}{\kwbind{\shole}{x}{\cmd_2}}}{\cmd}}}
               {\etp}{\esig}{\mem}{\tgraph{a}{\prio}{\sthread{u}}}
     }
\and
\Rule{D-Bind3}
     {u\fresh}
     {
       \lconfig{\mem}{\tp}{\cthread{a}{\prio}{\sig}
         {\screturn{\scp{\stack}{\kwbind{\shole}{x}{\cmd_2}}}{\kwret{v}}}}
       \mstep
       \rconfig{\cthread{a}{\prio}{\sig}
         {\scsend{\stack}{[v/x]\cmd_2}}}
               {\etp}{\esig}{\mem}
               {\tgraph{a}{\prio}{\sthread{u}}}
     }
     \and
     \fi

\iffull
\Rule{D-Create}
     {u\fresh\\
       b\fresh}
     {
       \phantom{\mstep} \lconfig{\mem}{\tp}
               {\cthread{a}{\prio}{\sig}
                 {\scsend{\stack}{\kwspawn{\prio'}{\tau}{\cmd}}}}
       \\\mstep
       \rconfig{\cthread{a}{\prio}{\sig, \sigtype{b}{\tau}{\prio'}}
         {\screturn{\stack}{\kwret{\kwtid{b}}}}}
               {\cthread{b}{\prio'}{\sig}{\scsend{\estack}{\cmd}}}
               {\esig}{\mem}
               {\dagq{\gthread{a}{\prio}{\sthread{u}}}
                 {\{(u, b)\}}{\emptyset}{\emptyset}}
     }
\else
\Rule{D-Create}
     {u\fresh\\
       b\fresh}
     {
       \phantom{\mstep} \lconfig{\mem}{\tp}
               {\cthread{a}{\prio}{\sig}
                 {\scsend{\stack}{\kwspawn{\prio'}{\tau}{\cmd}}}}
       \mstep
       \rconfig{\cthread{a}{\prio}{\sig, \sigtype{b}{\tau}{\prio'}}
         {\screturn{\stack}{\kwret{\kwtid{b}}}}}
               {\cthread{b}{\prio'}{\sig}{\scsend{\estack}{\cmd}}}
               {\esig}{\mem}
               {\dagq{\gthread{a}{\prio}{\sthread{u}}}
                 {\{(u, b)\}}{\emptyset}{\emptyset}}
     }
\fi

\and
\iffull
\Rule{D-Touch1}
     {u \fresh}
     {
       \lconfig{\mem}{\tp}
               {\cthread{a}{\prio}{\sig}{\scsend{\stack}{\kwsync{e}}}}
       \mstep
       \rconfig{\cthread{a}{\prio}{\sig}
         {\ssend{\scp{\stack}{\kwsync{\shole}}}{e}}}
               {\etp}{\esig}{\mem}
               {\tgraph{a}{\prio}{\sthread{u}}}
     }
\and
\fi
\Rule{D-Touch2}
     {u \fresh}
     {
       \phantom{\mstep}
       \lconfig{\mem}{\tp \tpcp \cthread{b}{\prio'}{\sig'}
         {\screturn{\estack}{\kwret{v}}}}
               {\cthread{a}{\prio}{\sig, \sigtype{b}{\prio'}{\tau'}}
                 {\sreturn{\scp{\stack}{\kwsync{\shole}}}{\kwtid{b}}}}
       \\\mstep
       \rconfig{\cthread{a}{\prio}{\sig, \sigtype{b}{\prio'}{\tau'}, \sig'}
         {\screturn{\stack}{\kwret{v}}}}
               {\etp}{\esig}{\mem}
               {\dagq{\gthread{a}{\prio}{\sthread{u}}}
                 {\emptyset}{\{(b, u)\}}{\emptyset}}
     }
\and
\iffull
\Rule{D-Dcl1}
     {u \fresh}
     {
       \phantom{\mstep}
       \lconfig{\mem}{\tp}
               {\cthread{a}{\prio}{\sig}
                 {\scsend{\stack}{\kwdcl{\tau'}{\kwassn}{e}{\cmd}}}}
       \mstep
       \rconfig{\cthread{a}{\prio}{\sig}
         {\ssend{\scp{\stack}{\kwdcl{\tau'}{\kwassn}{\shole}{\cmd}}}{e}}}
               {\etp}{\esig}{\mem}
               {\tgraph{a}{\prio}{\sthread{u}}}
     }
\and
\Rule{D-Dcl2}
     {u \fresh}
     {
       \phantom{\mstep}
       \lconfig{\mem}{\tp}{\cthread{a}{\prio}{\sig}
         {\sreturn{\scp{\stack}{\kwdcl{\tau'}{\kwassn}{\shole}{\cmd}}}{v}}}
       \mstep
       \rconfig{\cthread{a}{\prio}{\sig}
         {\scsend{\stack}{\cmd}}}
               {\etp}{\sigrtype{\kwassn}{\tau'}}
               {\mem[\mement{\kwassn}{v}{u}{\sig}]}
               {\tgraph{a}{\prio}{\sthread{u}}}
     }
\and
\Rule{D-Get1}
     {u \fresh}
     {
       \lconfig{\mem}{\tp}
               {\cthread{a}{\prio}{\sig}
                 {\scsend{\stack}{\kwderef{e}}}}
       \mstep
       \rconfig{\cthread{a}{\prio}{\sig}
         {\ssend{\scp{\stack}{\kwderef{\shole}}}{e}}}
               {\etp}{\esig}{\mem}
               {\tgraph{a}{\prio}{\sthread{u}}}
     }
\and
\fi
\Rule{D-Get2}
     {u \fresh\\
       \mem(\kwassn) = \memrent{v}{u'}{\sig'}
     }
     {
       \phantom{\mstep}
       \lconfig{\mem}{\tp}
               {\cthread{a}{\prio}{\sig, \sigrtype{\kwassn}{\tau'}}
                 {\sreturn{\scp{\stack}{\kwderef{\shole}}}{\kwref{\kwassn}}}}
       \mstep
       \rconfig{\cthread{a}{\prio}{\sig, \sigrtype{\kwassn}{\tau'}, \sig'}
         {\screturn{\stack}{\kwret{v}}}}
               {\etp}{\esig}{\mem}
               {\dagq{\gthread{a}{\prio}{\sthread{u}}}
                 {\emptyset}{\emptyset}{\{(u', u)\}}}
     }
\iffull
\end{mathpar}
\caption{Cost semantics for threads (1 of 2)}
\label{fig:cost1}
\end{figure*}

\begin{figure*}
\langfigsize
  \centering
  \begin{mathpar}
\Rule{D-Set1}
     {u \fresh}
     {
       \lconfig{\mem}{\tp}
               {\cthread{a}{\prio}{\sig}
                 {\scsend{\stack}{\kwassign{e_1}{e_2}}}}
       \mstep
       \rconfig{\cthread{a}{\prio}{\sig}
         {\ssend{\scp{\stack}{\kwassign{\shole}{e_2}}}{e_1}}}
               {\etp}{\esig}{\mem}
               {\tgraph{a}{\prio}{\sthread{u}}}
     }
\and
\Rule{D-Set2}
     {u \fresh}
     {
       \lconfig{\mem}{\tp}
               {\cthread{a}{\prio}{\sig}
                 {\sreturn{\scp{\stack}{\kwassign{\shole}{e}}}{\kwref{\kwassn}}}}
       \mstep
       \rconfig{\cthread{a}{\prio}{\sig}
         {\scsend{\scp{\stack}{\kwassign{\kwref{\kwassn}}{\shole}}}{e}}}
               {\etp}{\esig}{\mem}
               {\tgraph{a}{\prio}{\sthread{u}}}
     }
\fi
\and
\Rule{D-Set3}
     {u \fresh}
     {
       \phantom{\mstep}
       \lconfig{\mem}{\tp}{\cthread{a}{\prio}{\sig,\sigrtype{\kwassn}{\tau'}}
         {\sreturn{\scp{\stack}{\kwassign{\kwref{\kwassn}}{\shole}}}{v}}}
       \mstep
       \rconfig{\cthread{a}{\prio}{\sig,\sigrtype{\kwassn}{\tau'}}
         {\screturn{\stack}{\kwret{v}}}}
               {\etp}{\esig}
               {\mem[\mement{\kwassn}{v}{u}{\sig}]}
               {\tgraph{a}{\prio}{\sthread{u}}}
     }
\and
\iffull
\Rule{D-Ret1}
     {u \fresh}
     {
       \lconfig{\mem}{\tp}{\cthread{a}{\prio}{\sig}
         {\scsend{\stack}{\kwret{e}}}}
       \mstep
       \rconfig{\cthread{a}{\prio}{\sig}
         {\ssend{\scp{\stack}{\kwret{\shole}}}{e}}}
               {\etp}{\esig}{\mem}{\tgraph{a}{\prio}{\sthread{u}}}
     }
\and
\Rule{D-Ret2}
     {u \fresh}
     {
       \lconfig{\mem}{\tp}{\cthread{a}{\prio}{\sig}
         {\sreturn{\scp{\stack}{\kwret{\shole}}}{v}}}
       \mstep
       \rconfig{\cthread{a}{\prio}{\sig}
         {\screturn{\stack}{\kwret{v}}}}
               {\etp}{\esig}{\mem}
               {\tgraph{a}{\prio}{\sthread{u}}}
     }
\and
\Rule{D-Exp}
     {\stackstate \estep \stackstate'\\
       u\fresh}
     {
       \lconfig{\mem}{\tp}{\cthread{a}{\prio}{\sig}{\stackstate}}
       \mstep
       \rconfig{\cthread{a}{\prio}{\sig}{\stackstate'}}
               {\etp}{\esig}
               {\mem}{\tgraph{a}{\prio}{\sthread{u}}}
     }
\fi
\end{mathpar}
%%   \\[4ex]
%% \begin{mathpar}
%% \Rule{D-Par}
%%      {
%%        \tp = \tp_0 \tpcp
%%        \cthread{a_1}{\prio_1}{\sig_1}{\stackstate_1} \tpcp \dots
%%        \tpcp \cthread{a_n}{\prio_n}{\sig_n}{\stackstate_n}\\
%%          \lconfig{\mem}{\tp}{\cthread{a_i}{\prio_i}{\sig_i}{\stackstate_i}}
%%          \mstep
%%          \rconfig{\cthread{a_i}{\prio_i}{\sig_i'}{\stackstate_i'}}
%%                  {\tp_i'}{\sig_i''}{\mem_i'}
%%                  {\graph_i'}\\
%%        \tp' = \tp_0 \tpcp
%%                  \cthread{a_1}{\prio_1}{\sig_1'}{\stackstate_1'} \tpcp \tp_1'
%%                  \tpcp \dots \tpcp
%%                  \cthread{a_n}{\prio_i}{\sig_n'}{\stackstate_n'} \tpcp \tp_n'
%%      }
%%      {
%%        \gconfig{\sig}{\mem}{\graph}{\tp}
%%        \gstep
%%        \gconfig{\sig,\sig_1'',\dots, \sig_n''}
%%                {\mem,\mem_1', \dots, \mem_n'}
%%                {\graph \scomp{a_1} \graph_1' \dots \scomp{a_n} \graph_n'}
%%                {\tp'}
%%      }
  %% \end{mathpar}
  \iffull
  \caption{Cost semantics for threads (2 of 2)}
  \else
  \vspace{-4mm}
  \caption{Cost semantics for threads}
  \fi
\label{fig:cost2}
\iffull\else \vspace{-3mm} \fi
\end{figure*}

A full configuration of the stack machine includes the current heap~$\mem$
and set of threads~$\tp$.
A heap, essentially, is a mapping from memory locations to values.
For technical reasons, we also record two pieces of metadata in the heap
at each location: the DAG vertex that performed the last
write to that memory location (which will be used to add weak edges to the
cost graph) and a signature containing threads that one might ``learn about''
by reading this memory location.
For example, suppose thread~$a$ creates thread~$b$ and writes~$\kwtid{b}$
into a memory location~$\kwassn$.
If thread~$c$ later reads from~$\kwassn$, it must ``learn about'' the existence
of thread~$b$ in order to preserve typing.
We write an element of the heap as~$\mement{\kwassn}{v}{u}{\sig}$.
We denote the empty heap~$\emem$, and
let~$\mem[\mement{\kwassn}{v}{u}{\sig}]$ be the extension of~$\mem$ with
the binding~$\mement{\kwassn}{v}{u}{\sig}$.
If~$\kwassn \in \dom{\mem}$, the new binding is assumed to overwrite the
existing binding.
\iffull

\else%
\fi
A {\em thread pool}~$\tp$ maps thread symbols~$a$ to a triple consisting
of thread~$a$'s priority, its stack state and a signature~$\sig$
consisting of the threads that~$a$ ``knows about,'' as motivated above.
%
%\fi
\iffull
We require that~$a$'s stack state be well-typed with
signature~$\sig$.
This requirement will be preserved by the transition rules.
We write an element of the thread pool
as~$\cthread{a}{\prio}{\sig}{\stackstate}$.
We notate the empty thread pool as~$\etp$,
and notate the union of two (disjoint) thread pools by~$\tp_1 \tpcp \tp_2$.
For notational simplicity, we will
write~$\cthread{a}{\prio}{\sig}{\stackstate}$ for the
singleton thread pool consisting of only this binding.

Formally, a configuration of the stack machine consists of
a 4-tuple of a signature~$\sig$ containing signatures for the current heap,
the heap~$\mem$, a cost graph~$\graph$, and a thread pool~$\tp$.
Such a configuration is written as~$\gconfig{\sig}{\mem}{\graph}{\tp}$.
The judgment~$\gconfig{\sig}{\mem}{\graph}{\tp}
\gstep \gconfig{\sig'}{\mem'}{\graph'}{\tp'}$
indicates a single (parallel) step of the stack machine.
This judgment has one rule,~\rulename{D-Par}, shown below:
\[
\Rule{D-Par}
     {
       n \geq 1\\
       \tp = \tp_0 \tpcp
       \cthread{a_1}{\prio_1}{\sig_1}{\stackstate_1} \tpcp \dots
       \tpcp \cthread{a_n}{\prio_n}{\sig_n}{\stackstate_n}\\
         \lconfig{\mem}{\tp}{\cthread{a_i}{\prio_i}{\sig_i}{\stackstate_i}}
         \mstep
         \rconfig{\cthread{a_i}{\prio_i}{\sig_i'}{\stackstate_i'}}
                 {\tp_i'}{\sig_i''}{\mem_i'}
                 {\graph_i'}\\
       \tp' = \tp_0 \tpcp
                 \cthread{a_1}{\prio_1}{\sig_1'}{\stackstate_1'} \tpcp \tp_1'
                 \tpcp \dots \tpcp
                 \cthread{a_n}{\prio_i}{\sig_n'}{\stackstate_n'} \tpcp \tp_n'
     }
     {
       \gconfig{\sig}{\mem}{\graph}{\tp}
       \gstep
       \gconfig{\sig,\sig_1'',\dots, \sig_n''}
               {\mem,\mem_1', \dots, \mem_n'}
               {\graph \scomp{a_1} \graph_1' \dots \scomp{a_n} \graph_n'}
               {\tp'}
     }
\]
The rule selects an arbitrary subset~$a_1, \dots, a_n$ of the threads of~$\tp$
and steps them using an auxiliary
judgment~$\lconfig{\mem}{\tp}{\cthread{a_i}{\prio_i}{\sig_i}{\stackstate_i}}
\mstep \rconfig{\cthread{a_i}{\prio_i}{\sig_i'}{\stackstate_i'}}{\tp_i'}
       {\sig_i''}{\mem_i'}{\graph_i'}$.
In this judgment,~$\stackstate_i'$ is the new state of thread~$a_i$,~$\sig_i''$
contains the memory locations allocated by the step
and~$\mem_i'$ contains any
heap writes performed by the step.
The graph~$\graph_i'$ contains a vertex corresponding to this step as well as
any additional \create{}, \touch{} or weak edges added by this step.
We combine the signatures~$\sig, \sig_1'', \dots, \sig_n''$ freely
since the allocated memory locations are assigned to be disjoint.
We combine the heaps~$\mem_1', \dots, \mem_n'$, allowing writes
by~$a_j$ to overwrite writes by~$a_i$ for~$j > i$.
This corresponds to non-deterministic resolution of
write-write data races in the program.
We add each of the generated cost graphs to the existing graph~$\graph$
using thread-sequential composition:
\[
\begin{array}{c}
  \dagq{\gthread{a}{\prio}{\uthread_1} \uplus \gthreads_1}{\spawns_1}{\syncs_1}
       {\reads_2}
\scomp{a}
\dagq{\gthread{a}{\prio}{\uthread_2} \uplus \gthreads_2}{\spawns_2}{\syncs_2}
     {\reads_2}\\
\defeq\\
\dagq{\gthread{a}{\prio}{\uthread_1 \tscomp{} \uthread_2} \uplus
  \gthreads_1 \uplus \gthreads_2}{\spawns_1 \cup \spawns_2}
     {\syncs_1 \cup \syncs_2}{\reads_1 \cup \reads_2}
\end{array}
\]
Finally, we concatenate the resulting thread pools.

Figures~\ref{fig:cost1} and~\ref{fig:cost2} define the rules for the
auxiliary judgment~$\lconfig{\mem}{\tp}{\cthread{a_i}{\prio_i}{\sig_i}{\stackstate_i}}
\mstep \rconfig{\cthread{a_i}{\prio_i}{\sig_i'}{\stackstate_i'}}{\tp_i'}
       {\sig_i''}{\mem_i'}{\graph_i'}$.
\else

Figure~\ref{fig:cost2} presents a subset of the rules for the
command transition
judgment
\[\lconfig{\mem}{\tp}{\cthread{a_i}{\prio_i}{\sig_i}{\stackstate_i}}
\mstep \rconfig{\cthread{a_i}{\prio_i}{\sig_i'}{\stackstate_i'}}{\tp_i'}
       {\sig_i''}{\mem_i'}{\graph_i'}\]
In this judgment,~$\stackstate_i'$ is the new state of thread~$a_i$,~$\sig_i''$
contains the memory locations allocated by the step
and~$\mem_i'$ contains any
heap writes performed by the step.
The graph~$\graph_i'$ contains a vertex corresponding to this step as well as
any additional \create{}, \touch{} or weak edges added by this step.
The full semantics for the abstract machine also includes a single rule that
steps some number of threads in parallel and combines the resulting states
and graphs.
The full set of rules can be found in the full version~\citep{arxiv}
\fi

\iffull
If the command at the top of the stack is a bind, we first push the binding
on the stack and evaluate the expression~$e$ (\rulename{D-Bind1}).
When the resulting encapsulated command is returned, it in turn is evaluated
(\rulename{D-Bind2}) and finally the resulting value is substituted into
the new command and the bind command is popped from the stack
(\rulename{D-Bind3}).
Each of the above rules produces a graph consisting of one vertex.
\fi
An {\create} command simply creates a new thread symbol~$b$ and adds a thread~$b$
to the thread pool to execute the command~$m$.
It returns the thread handle, and adds a {\create} edge to the graph.
An {\touch} command first evaluates its subexpression
(\rulename{D-Touch1}).
When the thread handle~$\kwtid{b}$ is returned, rule~\rulename{D-Touch2}
inspects the thread pool for the
entry~$\cthread{b}{\prio'}{\sig'}{\screturn{\estack}{\kwret{v}}}$
(if~$b$'s stack is not of this form,~$b$ has not finished executing and
the \touch{} will block until it does).
The command returns the value~$v$ and adds the appropriate \touch{} edge.
It also adds~$\sig'$ to the set of threads that~$a$ ``knows about,'' because~$v$
might contain handles to threads in~$\sig'$.

\iffull
The rules \rulename{D-Dcl1}, \rulename{D-Get1}, \rulename{D-Set1},
\rulename{D-Set2} and \rulename{D-Ret1} are similar in that they
push a frame onto the stack and set aside a subexpression for evaluation.
Rule~\rulename{D-Dcl2} adds a binding to the heap for the new memory
location, initialized with the returned value~$v$, the new graph vertex~$u$
and the signature~$\sig$.
The signature~$\sigrtype{\kwassn}{\tau'}$ is returned as well.
The rule~\rulename{D-Set3} is similar, but returns the new value
while~\rulename{D-Dcl2} continues as the new command~$\cmd$.

Rule~\rulename{D-Get2} inspects the heap for the binding of~$\kwassn$,
returns its value, adds a weak edge~$(u', u)$ (recall that~$u'$
is the vertex corresponding to most recent write
to~$\kwassn$) and adds~$\sig'$ to the signature of~$a$.
Rule~\rulename{D-Ret2} returns the value of a command that has completed
execution.

\begin{figure}
\langfigsize
  \[
  \begin{array}{r l l}
    \ssend{\stack}{v} & \estep & \sreturn{\stack}{v}\\

    \ssend{\stack}{\kwifz{\kwnumeral{0}}{e_1}{x}{e_2}}
    & \estep & \ssend{\stack}{e_1}\\

    \ssend{\stack}{\kwifz{\kwnumeral{n + 1}}{e_1}{x}{e_2}}
    & \estep & \ssend{\stack}{[\kwn/x]e_2}\\

    \ssend{\stack}{\kwapply{(\kwfun{x}{e})}{v}} & \estep &
    \ssend{\stack}{[v/x]e}\\

    \ssend{\stack}{\kwfst{\kwepair{v_1}{v_2}}} & \estep &
    \sreturn{\stack}{v_1}\\

    \ssend{\stack}{\kwsnd{\kwepair{v_1}{v_2}}} & \estep &
    \sreturn{\stack}{v_2}\\

    \ssend{\stack}{\kwcase{\kweinl{v}}{x}{e_1}{y}{e_2}} & \estep &
    \ssend{\stack}{[v/x]e_1}\\

    \ssend{\stack}{\kwcase{\kweinr{v}}{x}{e_1}{y}{e_2}} & \estep &
    \ssend{\stack}{[v/y]e_2}\\

    \ssend{\stack}{\kwwapp{(\kwwlam{\vprio}{\cons}{e})}{\prio}} & \estep &
    \ssend{\stack}{[\prio/\vprio]e}\\

    \ssend{\stack}{\kwfix{x}{\tau}{e}} & \estep &
    \ssend{\stack}{[\kwfix{x}{\tau}{e}/x]e}\\

    \ssend{\stack}{\kwlet{x}{e_1}{e_2}} & \estep &
    \ssend{\scp{\stack}{\kwlet{x}{\shole}{e_2}}}{e_1}\\

    \sreturn{\scp{\stack}{\kwlet{x}{\shole}{e_2}}}{v} & \estep &
    \ssend{\stack}{[v/x]e_2}
  \end{array}
  \]
  \caption{Stack dynamics for expressions.}
  \label{fig:exp-stack}
\end{figure}

Finally, the rule~\rulename{D-Exp} handles cases in which the stack state
is evaluating an expression.
This rule uses the auxiliary judgment~$\stackstate_1 \estep \stackstate_2$
which expresses the transition relation between stack states that only
require expression evaluation.
Because expressions do not alter or depend on global state, this judgment
requires access only to the stack state itself and thus its rules
are drastically simpler and relatively standard for a stack-based semantics.
These rules are presented in Figure~\ref{fig:exp-stack}.
Because expressions are in A-normal form, only the rules for let-binding
require evaluating a subexpression.
All other rules immediately pop the expression from the stack,
perform its associated computation, and either return the resulting value or
(e.g., in the case of a function application) continue as a new expression.
\else
Rule~\rulename{D-Set3} adds a binding to the heap for the new value of the
memory location, and includes as metadata
the new graph vertex~$u$ and the signature~$\sig$.
%
%The signature~$\sigrtype{\kwassn}{\tau'}$ is returned as well.
%
%The rule~\rulename{D-Set3} is similar, but returns the new value
%while~\rulename{D-Dcl2} continues as the new command~$\cmd$.
%
Rule~\rulename{D-Get2} inspects the heap for the binding of~$\kwassn$,
returns its value, adds a weak edge~$(u', u)$ (recall that~$u'$
is the vertex corresponding to most recent write
to~$\kwassn$) and adds~$\sig'$ to the signature of~$a$.
\fi

\iffull
\begin{figure}
\langfigsize
  \centering
  \def \MathparLineskip {\lineskip=\mylineskip}
\begin{mathpar}
\Rule{KS-Empty}
     {\strut}
     {\stackacceptsc{\sig}{\estack}{\tau}{\tau}{\prio}}
\and
\Rule{KS-Let}
     {\stackaccepts{\sig}{\stack}{\tau_2}{\tau'}{\prio}\\
       \etyped{\sig}{\hastype{x}{\tau_1}}{e}{\tau_2}
     }
     {
       \stackaccepts{\sig}{\scp{\stack}{\kwlet{x}{\shole}{e}}}{\tau_1}{\tau'}{\prio}
     }
\iffull
\and
\Rule{KS-Bind1}
     {\stackacceptsc{\sig}{\stack}{\tau_2}{\tau'}{\prio}\\
       \cmdtyped{\sig}{\hastype{x}{\tau_1}}{\cmd}{\tau_2}{\prio}
     }
     {
       \stackaccepts{\sig}{\scp{\stack}{\kwbind{\shole}{x}{\cmd}}}
                    {\kwcmdt{\tau_1}{\prio}}{\tau'}{\prio}
     }
\and
\Rule{KS-Bind2}
     {\stackacceptsc{\sig}{\stack}{\tau_2}{\tau'}{\prio}\\
       \cmdtyped{\sig}{\hastype{x}{\tau_1}}{\cmd}{\tau_2}{\prio}
     }
     {
       \stackacceptsc{\sig}{\scp{\stack}{\kwbind{\shole}{x}{\cmd}}}
                     {\tau_1}{\tau'}{\prio}
     }
\and
\Rule{KS-Sync}
     {
       \stackacceptsc{\sig}{\stack}{\tau}{\tau'}{\prio}\\
       \ple{\prio}{\prio'}
     }
     {
       \stackaccepts{\sig}{\scp{\stack}{\kwsync{\shole}}}{\kwat{\tau}{\prio'}}
                    {\tau'}{\prio}
     }
\and
\Rule{KS-Dcl}
     {
       \stackacceptsc{\sig}{\stack}{\tau_2}{\tau'}{\prio}\\
       \cmdtyped{\sig, \sigrtype{\kwassn}{\tau_1}}{\ectx}{\cmd}{\tau_2}{\prio}
     }
     {
       \stackaccepts{\sig}{\scp{\stack}{\kwdcl{\tau_1}{\kwassn}{\shole}{\cmd}}}
         {\tau_1}{\tau'}{\prio}
     }
\and
\Rule{KS-Get}
     {
       \stackacceptsc{\sig}{\stack}{\tau}{\tau'}{\prio}
     }
     {
       \stackaccepts{\sig}{\scp{\stack}{\kwderef{\shole}}}{\kwreft{\tau}}
                    {\tau'}{\prio}
     }
\and
\Rule{KS-Set1}
     {
       \stackacceptsc{\sig}{\stack}{\tau}{\tau'}{\prio}
     }
     {
       \stackaccepts{\sig}{\scp{\stack}{\kwassign{\shole}{e}}}{\kwreft{\tau}}
                    {\tau'}{\prio}
     }
\and
\Rule{KS-Set2}
     {
       \stackacceptsc{\sig}{\stack}{\tau}{\tau'}{\prio}
     }
     {
       \stackaccepts{\sig}{\scp{\stack}{\kwassign{e}{\shole}}}{\tau}
                    {\tau'}{\prio}
     }
\and
\Rule{KS-Ret}
     {
       \stackacceptsc{\sig}{\stack}{\tau}{\tau'}{\prio}
     }
     {
       \stackaccepts{\sig}{\scp{\stack}{\kwret{\shole}}}{\tau}{\tau'}{\prio}
     }
\end{mathpar}
\\[4ex]
\begin{mathpar}
\Rule{KS-PopExp}
     {
       \stackaccepts{\sig}{\stack}{\tau}{\tau'}{\prio}\\\\
       \etyped{\sig}{\ectx}{e}{\tau}
     }
     {
       \sstyped{\sig}{\ssend{\stack}{e}}{\tau'}{\prio}
     }
\and
\Rule{KS-PopCmd}
     {
       \stackacceptsc{\sig}{\stack}{\tau}{\tau'}{\prio}\\\\
       \cmdtyped{\sig}{\ectx}{\cmd}{\tau}{\prio}
     }
     {
       \sstyped{\sig}{\scsend{\stack}{\cmd}}{\tau'}{\prio}
     }
\and
\Rule{KS-PushExp}
     {
       \stackaccepts{\sig}{\stack}{\tau}{\tau'}{\prio}\\\\
       \etyped{\sig}{\ectx}{v}{\tau}
     }
     {
       \sstyped{\sig}{\sreturn{\stack}{v}}{\tau'}{\prio}
     }
\and
\Rule{KS-PushCmd}
     {
       \stackacceptsc{\sig}{\stack}{\tau}{\tau'}{\prio}\\\\
       \etyped{\sig}{\ectx}{v}{\tau}
     }
     {
       \sstyped{\sig}{\screturn{\stack}{\kwret{v}}}{\tau'}{\prio}
     }
\fi
\end{mathpar}
\iffull
\caption{Typing rules for stacks.}
\else
\caption{Selected typing rules for stacks.}
\fi
\label{fig:stack-statics}
\end{figure}

We can relate the typing rules of the previous section with the transition
judgment in a type safety proof showing that, in the typical parlance,
``well-typed programs don't go wrong'' (``go wrong'' here refers simply to
entering a ``stuck'' state; we also wish to show that well-typed programs
lead to well-formed cost graphs, but this will be the subject of the next
subsection).
In order to state such a property formally, we need static semantics for
stacks, and stack states.
These are given by Figure~\ref{fig:stack-statics} in three judgments.
The judgments~$\stackaccepts{\sig}{\stack}{\tau_1}{\tau_2}{\prio}$
and~$\stackacceptsc{\sig}{\stack}{\tau_1}{\tau_2}{\prio}$
indicate that the stack~$\stack$ ``accepts'' a value of type~$\tau_1$ and
transforms it to a value of type~$\tau_2$ at priority~$\prio$.
The former judgment is used for stacks that are evaluating an expression
and the latter for stacks that are evaluating a command.
The rules for these judgments largely follow the corresponding rules
of Figures~\ref{fig:statics} and~\ref{fig:thread-statics}.
The judgment~$\sstyped{\sig}{\stackstate}{\tau'}{\prio}$
states that a stack state is well-formed and computes a value of type~$\tau'$
at priority~$\prio$.
The rules for this judgment  require that the stack takes a value of
type~$\tau$ to one of~$\tau'$ (using one of the two judgments described above)
and that the expression or command have type~$\tau$.

As is typical, the proof of type safety will be split into two theorems:
Preservation and Progress.
The statement of the preservation theorem requires two additional definitions.
Intuitively, we say that a thread pool~$\tp$ is {\em compatible} with
a signature~$\sig$ if, whenever a thread~$a$ ``knows about'' a thread~$b$,
the thread~$b$ exists in~$\tp$ and its stack state is well-typed.

\begin{defn}
  A thread pool~$\tp$ is {\em compatible} with a signature~$\sig$
  if for all~$\cthread{a}{\prio_a}{\sig_a}{\stackstate_a} \in \tp$,
  and all~$\sigtype{b}{\tau_b}{\prio_b} \in \sig_a$,
  we have~$\cthread{b}{\prio_b}{\sig_b}{\stackstate_b} \in \tp$
  and~$\sstyped{\sig,\sig_b}{\stackstate_b}{\tau_b}{\prio_b}$.
\end{defn}

A heap~$\mem$ is well-typed with respect to a signature~$\sig$,
written~$\mtyped{\worlds}{\mem}{\sig}$
if for all~$\sigrtype{\kwassn}{\tau} \in \sig$,
we have~$\mement{\kwassn}{v}{u}{\sig'} \in \mem$
and~$\etyped{\sig,\sig'}{\ectx}{v}{\tau}$.

The preservation theorem itself is split into three cases, corresponding
to expression steps\iffull(Figure~\ref{fig:exp-stack})\fi, thread steps
%\iffull
(Figures~\ref{fig:cost1} and~\ref{fig:cost2})
%\else (Figure~\ref{fig:cost2})
%\fi
and parallel steps
(Rule~\rulename{D-Par}).

\begin{theorem}[Preservation]\label{lem:preservation}
  \strut\\
  \begin{enumerate}
  \item If~$\sstyped{\sig}{\stackstate}{\tau}{\prio}$
    and~$\stackstate \estep \stackstate'$
    then~$\sstyped{\sig}{\stackstate'}{\tau}{\prio}$.
  \item If~$\tp$ is compatible with~$\sig$
    and~$\sstyped{\sig,\sig_a}{\stackstate}{\tau}{\prio}$ and
    \[\lconfig{\mem}{\tp}{\cthread{a}{\prio}{\sig_a}{\stackstate}}
      \mstep
      \rconfig{\cthread{a}{\prio}{\sig_a'}{\stackstate'}}
              {\tp'}{\sig'}{\mem'}
              {\graph}
    \]
    then~$\sstyped{\sig,\sig_a'}{\stackstate'}{\tau}{\prio}$
    and~$\mtyped{\worlds}{\mem'}{\sig,\sig'}$ and~$\tp'$ is
    compatible with~$\sig, \sig'$ and for
    all~$\cthread{b}{\prio_b}{\sig_b}{\stackstate_b} \in \tp' \setminus \tp$,
    we have~$\sstyped{\sig,\sig_b}{\stackstate_b}{\tau_b}{\prio_b}$.
  \item If~$\mtyped{\worlds}{\mem}{\sig}$
    and~$\tp$ is compatible with~$\sig$
    and~$\sstyped{\sig,\sig_a}{\stackstate}{\tau}{\prio}$
    for all~$\cthread{a}{\prio}{\sig_a}{\stackstate} \in \tp$
    and
    \[\gconfig{\sig}{\mem}{\graph}{\tp} \gstep
    \gconfig{\sig'}{\mem'}{\graph'}{\tp'}
    \]
    then~$\mtyped{\worlds}{\mem'}{\sig'}$
    and~$\tp'$ is compatible with~$\sig'$
    and~$\sstyped{\sig,\sig_a'}{\stackstate'}{\tau}{\prio}$
    for all~$\cthread{a}{\prio}{\sig_a'}{\stackstate'} \in \tp'$.
  \end{enumerate}
\end{theorem}
\begin{proof}
  \begin{enumerate}
  \item
    By cases on the rules for~$\stackstate \estep \stackstate'$.
  \item
    By cases on the rules for
    \[\lconfig{\mem}{\tp}{\cthread{a}{\prio}{\sig_a}{\stackstate}}
      \mstep
      \rconfig{\cthread{a}{\prio}{\sig_a'}{\stackstate'}}
              {\tp'}{\sig'}{\mem'}
              {\graph}
              \]
%\iffull
    \begin{itemize}
    \item \rulename{D-Bind1}.
      Then~$\stackstate = \ssend{\stack}{\kwbind{e}{x}{\cmd_2}}$.
      By inversion on~\rulename{KS-PopCmd} and~\rulename{Bind},
      $\stackaccepts{\sig,\sig_a}{\stack}{\tau_2}{\tau}{\prio}$
      and~$\etyped{\sig,\sig_a}{\ectx}{e}{\kwcmdt{\tau_1}{\prio}}$
      and~$\cmdtyped{\sig,\sig_a}{\hastype{x}{\tau_1}}{\cmd_2}{\tau_2}{\prio}$.
      Apply~\rulename{KS-Bind1} and~\rulename{KS-PopExp}.
    \item \rulename{D-Bind2}.
      Then~$\stackstate = \sreturn{\scp{\stack}{\kwbind{\shole}{x}{\cmd_2}}}
      {\kwcmd{\prio}{\cmd}}$.
      By inversion on~\rulename{KS-PushExp},
      $\stackaccepts{\sig,\sig_a}{\scp{\stack}{\kwbind{\shole}{x}{\cmd_2}}}
      {\kwcmdt{\tau_1}{\prio}}{\tau}{\prio}$
      and~$\etyped{\sig,\sig_a}{\ectx}{\kwcmd{\prio}{\cmd}}{\kwcmdt{\tau_1}{\prio}}$.
      By inversion on~\rulename{\cmdsym I},
      $\cmdtyped{\sig,\sig_a}{\ectx}{\cmd}{\tau_1}{\prio}$.
      By inversion on~\rulename{KS-Bind1},
      $\stackacceptsc{\sig,\sig_a}{\stack}{\tau_2}{\tau}{\prio}$
      and~$\cmdtyped{\sig,\sig_a}{\ectx}{\hastype{x}{\tau_1}}{\cmd_2}{\tau_2}{\prio}$.
      Apply~\rulename{KS-Bind2} and~\rulename{KS-PopCmd}.
    \item \rulename{D-Bind3}.
      Then~$\stackstate = \screturn{\scp{\stack}{\kwbind{\shole}{x}{\cmd_2}}}
      {\kwret{v}}$.
      By inversion on~\rulename{KS-PushCmd},
      $\stackacceptsc{\sig,\sig_a}{\scp{\stack}{\kwbind{\shole}{x}{\cmd_2}}}
      {\tau}{\tau'}{\prio}$
      and~$\etyped{\sig,\sig_a}{\ectx}{v}{\tau}$.
      By inversion on~\rulename{KS-Bind2},
      $\stackacceptsc{\sig,\sig_a}{\stack}{\tau_2}{\tau'}{\prio}$
      and~$\cmdtyped{\sig,\sig_a}{\ectx}{\hastype{x}{\tau_1}}{\cmd_2}{\tau_2}{\prio}$.
      By Lemma~\ref{lem:subst},
      $\cmdtyped{\sig,\sig_a}{\ectx}{[v/x]\cmd_2}{\tau_2}{\prio}$.
      Apply~\rulename{KS-PopCmd}.
    \item \rulename{D-Create}.
      Then~$\stackstate = \scsend{\stack}{\kwspawn{\prio'}{\tau}{\cmd}}$.
      By inversion on~\rulename{KS-PopCmd} and~\rulename{Create},
      $\stackacceptsc{\sig,\sig_a}{\stack}{\kwat{\tau}{\prio'}}{\tau'}{\prio}$
      and~$\cmdtyped{\sig,\sig_a}{\ectx}{\cmd}{\tau}{\prio'}$.
      By~\rulename{Tid},~$\etyped{\sig,\sig_a'}{\ectx}{\kwtid{b}}{\kwat{\tau}{\prio'}}$.
      By~\rulename{KS-PushCmd},
      $\sstyped{\sig,\sig_a'}{\sreturn{\stack}{\kwret{\kwtid{b}}}}{\tau'}{\prio}$.
      By~\rulename{KS-PopCmd} and~\rulename{KS-Empty},
      $\sstyped{\sig,\sig_a}{\ssend{\estack}{\cmd}}{\tau'}{\prio'}$.
    \item \rulename{D-Touch2}.
      Then~$\stackstate = \sreturn{\scp{\stack}{\kwsync{\shole}}}{\kwtid{b}}$.
      By inversion on~\rulename{KS-PushExp} and~\rulename{KS-Touch},
      $\stackacceptsc{\sig,\sig_a}{\stack}{\tau}{\tau'}{\prio}$
      and~$\etyped{\sig,\sig_a}{\ectx}{\kwtid{b}}{\kwat{\tau}{\prio'}}$.
      By inversion on~\rulename{Tid}, we
      have~$\sigtype{b}{\tau}{\prio'} \in \sig_a$,
      so by
      compatibility,~$\sstyped{\sig,\sig_b}{\screturn{\estack}{\kwret{v}}}{\tau}{\prio'}$
      and for all~$\sigtype{c}{\tau_c}{\prio_c} \in \sig'$,
      we have~$\cthread{c}{\prio_c}{\sig_c}{\stackstate_c} \in \tp$
      and~$\sstyped{\sig,\sig_c}{\stackstate_c}{\tau_c}{\prio_c}$
      so~$\tp$ is compatible with~$\sig$.
      By inversion on~\rulename{KS-PushCmd} and \rulename{KS-Empty},
      $\etyped{\sig,\sig_a}{\ectx}{v}{\tau}$.
      Apply~\rulename{Ret} and~\rulename{KS-PushCmd}.
    \item \rulename{D-Dcl2}.
      Then~$\stackstate = \sreturn{\scp{\stack}
        {\kwdcl{\tau'}{\kwassn}{\shole}{\cmd}}}{v}$.
      By inversion on \rulename{KS-PushExp} and \rulename{KS-Dcl},
      $\stackacceptsc{\sig,\sig_a}{\stack}{\tau_0}{\tau}{\prio}$
      and~$\cmdtyped{\sig,\sig_a'}{\ectx}
      {\cmd}{\tau_0}{\prio}$
      and~$\etyped{\sig,\sig_a'}{\ectx}{v}{\tau'}$.
      By \rulename{KS-PopCmd},
      $\sstyped{\sig,\sig_a'}{\scsend{\stack}{\cmd}}{\tau}{\prio}$.
    \item \rulename{D-Get2}.
      Then~$\stackstate = \sreturn{\scp{\stack}{\kwderef{\shole}}}
      {\kwref{\kwassn}}$.
      By inversion on \rulename{KS-PushExp} and \rulename{KS-Get},
      $\stackacceptsc{\sig,\sig_a}{\stack}{\tau_0}{\tau}{\prio}$
      and~$\etyped{\sig,\sig_a}{\ectx}{\kwref{\kwassn}}{\kwreft{\tau_0}}$.
      By inversion on~\rulename{Ref},
      $\sigrtype{\kwassn}{\tau} \in \sig,\sig_a$.
      By heap typing,~$\etyped{\sig,\sig'}{\ectx}{v}{\tau_0}$.
      Apply \rulename{KS-PushCmd}.
    \end{itemize}
    %\fi
  \item
    By (2), for
    all~$\cthread{a}{\prio}{\sig_a}{\stackstate} \in \tp'$
    we have~$\sstyped{\sig,\sig_1'',\dots,\sig_n'',\sig_a}{\stackstate}{\tau}{\prio}$.
    We also have for all~$i$,
    that~$\mtyped{\worlds}{\mem_i'}{\sig,\sig_1'',\dots, \sig_n''}$
    and that~$\tp \tpcp \cthread{a_i}{\prio_i}{\sig_i'}{\stackstate_i'}
    \tpcp \tp_i'$ is compatible with~$\sig,\sig_i''$.
    Thus,~$\tp'$ is compatible with~$\sig,\sig_1'',\dots,\sig_n''$.
  \end{enumerate}
\end{proof}

The progress theorem for~{\calcname} is somewhat different from those of
sequential calculi.
One might expect this theorem to state that any machine configuration
consisting of well-typed threads can take a step with Rule~\rulename{D-Par}.
However, this theorem would hold as long as even one thread can take
a step---all other threads could be stuck!
Instead, we show that if all threads of a machine configuration are
well-typed, {\em each thread} can take a step (cases (2) and (4)) in the
statement below unless it is finished executing (case (1)) or is blocked
waiting on another thread to finish (case (3)).

\begin{theorem}[Progress]\label{thm:progress}
  If~$\cthread{a}{\prio}{\sig_a}{\stackstate} \in \tp$
  and~$\sstyped{\sig}{\stackstate}{\tau}{\prio}$
  and~$\mtyped{\worlds}{\mem}{\sig}$
  and~$\tp$ is compatible with~$\sig$
  then
  \begin{enumerate}
  \item $\stackstate = \screturn{\estack}{\kwret{v}}$ or
  \item $\stackstate \estep \stackstate'$ or
  \item $\stackstate = \sreturn{\scp{\stack'}{\kwsync{\shole}}}{\kwtid{b}}$
  and~$\cthread{b}{\prio'}{\sig_b}{\stackstate_b} \in \tp$
  where~$\stackstate_b \neq \screturn{\estack}{\kwret{v}}$
  or
  \item
  $\lconfig{\mem}{\tp}{\cthread{a}{\prio}{\sig_a}{\stackstate}}
  \mstep
  \rconfig{\cthread{a}{\prio}{\sig_a'}{\stackstate'}}
          {\tp'}{\sig'}{\mem'}{\graph'}
          $
  \end{enumerate}
\end{theorem}
\begin{proof}
By induction on the derivation
of~$\sstyped{\sig}{\stackstate}{\tau}{\prio}$.
\iffull
    \begin{itemize}
    \item \rulename{KS-PopExp}.
      Then~$\stackstate = \ssend{\stack}{e}$
      and~$\stackaccepts{\sig}{\stack}{\tau'}{\tau}{\prio}$
      and~$\etyped{\sig}{\ectx}{e}{\tau'}$. Proceed by induction on the
      latter derivation.
    \item \rulename{KS-PushExp}.
      Then~$\stackstate = \sreturn{\stack}{v}$
      and~$\stackaccepts{\sig}{\stack}{\tau'}{\tau}{\prio}$
      and~$\etyped{\sig}{\ectx}{v}{\tau'}$.
      Proceed by induction on the former derivation.
      \begin{itemize}
      %\item \rulename{KS-Let}
      \item \rulename{KS-Bind1}.
        Then~$\stack = \scp{\stack'}{\kwbind{\shole}{x}{\cmd}}$
        and~$\tau' = \kwcmdt{\tau_1}{\prio}$,
        so by canonical forms,~$v = \kwcmd{\prio}{\cmd_1}$.
        Apply~\rulename{D-Bind2}.
      \item \rulename{KS-Touch}.
        Then~$\stack = \scp{\stack'}{\kwsync{\shole}}$
        and~$\tau' = \kwat{\tau_1}{\prio'}$,
        so by canonical forms,~$v = \kwtid{b}$.
        and~$\sigtype{b}{\tau_1}{\prio'} \in \sig$.
        By compatibility,~$\cthread{b}{\prio'}{\sig_b}{\stackstate_b} \in \tp$.
        If~$\stackstate_b = \screturn{\estack}{\kwret{v}}$, then
        apply~\rulename{D-Touch2}.
        Otherwise, we are in case (3) of the conclusion.
      %\item \rulename{KS-Set1}.
      %\item \rulename{KS-Set2}.
      %\item \rulename{KS-Ret}.
      \end{itemize}
    \item \rulename{KS-PopCmd}.
      Then~$\stackstate = \scsend{\stack}{\cmd}$
      and~$\stackacceptsc{\sig}{\stack}{\tau'}{\tau}{\prio}$
      and~$\cmdtyped{\sig}{\ectx}{\cmd}{\tau'}{\prio}$.
      Proceed by induction on the latter derivation.
    \item \rulename{KS-PushCmd}.
      Then~$\stackstate = \screturn{\stack}{\kwret{v}}$
      and~$\stackacceptsc{\sig}{\stack}{\tau'}{\tau}{\prio}$
      and~$\etyped{\sig}{\ectx}{v}{\tau'}{\prio}$.
      Proceed by induction on the former derivation.
      \begin{itemize}
      \item \rulename{KS-Empty}.
        Then we are in case (1) of the conclusion.
      \item \rulename{KS-Bind2}.
        Then~$\stack = \scp{\stack'}{\kwbind{\shole}{x}{\cmd_2}}$.
        Apply rule~\rulename{D-Bind3}.
      \end{itemize}
    \end{itemize}
    \fi
\end{proof}

\subsubsection{Soundness of Type System and Cost Bounds}\label{sec:cost-corr}

We have now shown that the operational semantics and the type system respect
each other in the standard sense of type safety: that well-typed programs
don't become ``stuck''.
In this section, we show that the type system also meets the goal for which
it was designed: that it ensures that well-typed programs lead to well-formed
cost graphs, to which we can apply the cost bound results of
Theorem~\ref{thm:gen-brent}.
In fact, we will show a slightly stronger property that is easier to prove.
A DAG is {\em strongly well-formed} if, whenever there is a {\touch} edge
from~$a$ to~$b$, it is the case that~$a$'s priority is higher than~$b$'s
priority and, in the intuition we have been using above,~$b$ ``knows about''~$a$.
This latter property is formalized by requiring that there exists a path from
the creation of~$a$ to the {\touch} point, other than the one that goes through~$a$.

\begin{defn}\label{def:strongly-well-formed}
  A DAG~$\graph= \dagq{\gthreads}{\spawns}{\syncs}{\reads}$ is
  {\em strongly well-formed} if
  for all~$(a, u) \in \syncs \cup \reads$,
  we have that
  \begin{enumerate}
  \item $\gthread{a}{\prio_a}{\uthread},
  \gthread{b}{\prio_b}{\uthread_1 \tscomp{} u \tscomp{} \uthread_2} \in \gthreads$,
  \item If~$(a, u) \in \syncs$ then $\ple{\prio_b}{\prio_a}$ and
  \item If~$(u', a) \in \spawns$, then there exists a path from~$u'$ to~$u$
    where the first and last edges are continuation edges.
  \end{enumerate}
\end{defn}

\begin{lemma}\label{lem:wf-alt}
  If~$\graph$ is strongly well-formed, then~$\graph$ is well-formed.
\end{lemma}
\iffull
\begin{proof}
  Let $\gthread{a}{\prio}{u_1 \tscomp{} \dots \tscomp{} u_n} \in \gthreads$ and
  let~$\anc{u'}{u_n}$.
  We first show that either~$\anc{u'}{u_1}$ or~$\wanc{u'}{u_n}$ or
  $\ple{\prio}{\uprio{\graph}{u}}$.
  Since the graph is finite and acyclic, we can proceed by well-founded
  induction on~$\anc{}{}$.
  If~$u' = u_n$, the result is clear.
  Otherwise, assume that
  for all~$u''$ such that~$\neqanc{u'}{\anc{u''}{u_n}}$,
  we have~$\anc{u''}{u_1}$ or~$\wanc{u''}{u_n}$ or
  $\ple{\prio}{\uprio{\graph}{u''}}$.
  If~$\anc{u''}{u_1}$ for any such~$u''$, then~$\anc{u'}{u_1}$,
  and if~$\wanc{u''}{u_n}$ for any such~$u''$, then~$\wanc{u'}{u_n}$,
  so suppose that for all such~$u''$ it is the case that
  $\ple{\prio}{\uprio{\graph}{u''}}$.
  Let~$E$ be the set of outgoing edges of~$u'$.
  If any edge in~$E$ is a continuation or {\touch} edge, then we have
  $\ple{\prio}{\ple{\uprio{\graph}{u''}}{\uprio{\graph}{u'}}}$.
  So, there must exist an~{\create} edge~$(u, u'') \in E$ where~$u''$ is
  the first vertex of a thread~$b$.
  If there exists a corresponding {\touch} or weak edge~$(b, u_t)$ in
  the path from~$u''$ to~$u_n$,
  then by assumption there exists a path from~$u'$ to~$u_t$ where the
  first edge is a continuation edge, but this is a contradiction because
  there were assumed to be no continuation edges in~$E$.
  If no corresponding {\touch} edge~$(b, u_t)$ is in the path,
  then~$u_n$ must be in~$b$ or a thread transitively {\create}d by it,
  so~$\anc{u'}{u_1}$, also a contradiction.

  Next, we show the second condition of well-formedness.
  Let~$(u_0, u) \in \graph$ be a strong edge such that~$\sanc{u}{t}$
  and~$\nanc{u_0}{s}$ and $\nple{\uprio{\graph}{u}}{\uprio{\graph}{u_0}}$.
  By strong well-formedness,~$(u_0, u)$ must be an {\create} edge.
  If there is no {\touch} or weak edge on the path from~$u$ to~$t$,
  then~$\anc{u_0}{s}$, a contradiction.
  So there exists an {\touch} or weak edge~$(u_1, u_2)$ on the path
  from~$u$ to~$t$.
  By assumption, there exists a path from~$u_0$ to~$u_2$ where the first
  and last edges are continuation edges.
  By the first condition of well-formedness, along this path,
  there exists~$u'$ such that~$\wanc{u_0}{\sanc{u'}{u_2}}$.
  Because~$u'$ is an ancestor of the join point of the thread containing~$u$,
  we have~$\nanc{u}{u'}$.
\end{proof}
\else
The proof is similar to that of the corresponding theorem of
prior work~\citep{mah-priorities-2018}, but with an extra case to handle
weak edges.
Full details are available in the supplementary material.
\fi

We now prove the {\em soundness} result for {\calcname}: a well-typed
program produces a strongly well-formed cost graph under the operational
semantics.
Lemma~\ref{lem:soundness-mstep} states some key properties of the
step relation on individual threads.
In particular, it states that such a step does not create any cycles in the
larger DAG (item 2),
that it does not create a priority-inverted {\touch} edge (item 3),
that it ``knows about'' any threads it adds to its signature through
creating them itself or through {\touch} or weak edges (item 4),
and that any locations it allocates or writes to on the heap are tagged with
the correct vertex and signature (item 5).

\begin{lemma}\label{lem:soundness-mstep}
  Let~$\sig$,~$\sig_a$, $\stackstate$ and~$\prio$ be such that
  we have~$\sstyped{\sig_a}{\stackstate}{\tau}{\prio}$.
  If
  \[\lconfig{\mem}{\tp}{\cthread{a}{\prio}{\sig_a}{\stackstate}}
         \mstep
         \rconfig{\cthread{a}{\prio}{\sig_a'}{\stackstate'}}
                 {\tp'}{\sig'}{\mem'}
                 {\dagq{\gthreads}{\spawns}{\syncs}{\reads}}
  \]
  then
  \begin{enumerate}
  \item There exists~$\gthread{a}{\prio}{\sthread{u}} \in \gthreads$,
    where~$u$ is a freshly created vertex.
  \item For any edge~$(u, u') \in \spawns \cup \syncs \cup \reads$,
    the vertex~$u'$ is freshly created or is a placeholder for a freshly
    created thread. Furthermore, the collection of edges contains no cycles.
  \item If $(b, u) \in \syncs$, then
    $\sigtype{b}{\tau_b}{\prio_b} \in \sig_a$ and~$\ple{\prio}{\prio_b}$.
  \item If~$\sigtype{b}{\tau_b}{\prio_b} \in \sig_a' \setminus \sig_a$,
    then (a)~$(u, b) \in \spawns$
    or (b)~$(c, u) \in \syncs$
    and~$\cthread{c}{\prio_c}{\sig_c}{\stackstate_c} \in \tp$
    and~$\sigtype{b}{\tau_b}{\prio_b} \in \sig_c$
    or (c)~$(u', u) \in \reads$
    and~$\mement{\kwassn}{v}{u'}{\sig_{u'}} \in \mem$
    and~$\sigtype{b}{\tau_b}{\prio_b} \in \sig_{u'}$.
  \item For all~$\mement{\kwassn}{v}{u'}{\sig_{u'}} \in \mem' \setminus \mem$,
    we have that~$u' = u$ and~$\sig_{u'} \subset \sig_a$.
  \end{enumerate}
\end{lemma}
\iffull
\begin{proof}
  \begin{enumerate}
  \item By inspection of the cost semantics rules.
  \item By inspection of the cost semantics rules.
  \item Only rule~\rulename{D-Touch2} applies.
    The requirements follow from inversion on~\rulename{D-Touch2}
    and~\rulename{KS-Touch}.
    %XXX TODO technically we have to match up the
    %prio in the signature with the prio in the thread pool.
  \item By cases on the step rule applied.
    \begin{itemize}
    \item \rulename{D-Create}. Then we
      have~$\sig_a' \setminus \sig_a = \sigtype{b}{\prio'}{\tau}$
      and~$(u, b) \in \spawns$.
    \item \rulename{D-Touch2}. Then we
      have~$\sig_a' \setminus \sig_a = \sig'$
      and~$\cthread{b}{\prio'}{\sig'}{\stackstate_b} \in \tp$.
    \item \rulename{D-Get2}. Then we
      have~$\sig_a' \setminus \sig_a = \sig'$
      and~$\mement{\kwassn}{v}{u'}{\sig'} \in \mem$.
    \end{itemize}
  \item By cases on the step rule applied.
    \begin{itemize}
    \item \rulename{D-Dcl2}. Then we
      have~$\mem' \setminus \mem = \mement{\kwassn}{v}{u}{\sig_a}$.
    \item \rulename{D-Set3}. Then we
      have~$\mem' \setminus \mem = \mement{\kwassn}{v}{u}{\sig}$
      and~$\sig \subset \sig_a$.
    \end{itemize}
  \end{enumerate}
\end{proof}
\else
The proof of each statement follows largely from inspection of the
corresponding cost rules.
\fi

To show that the operational semantics preserves the guarantees of
strong well-formedness, we must introduce two definitions formally connecting
concepts of the DAG model and the operational semantics.
In particular, in the context of graphs, we have said that a vertex~$u$
``knows about'' a thread~$a$ if there exists a path from the creation of~$a$
to~$u$ other than the path that goes through~$a$.
In the context of the operational semantics, we have said that a thread~$b$
``knows about'' another thread~$a$ if~$b$ is in the domain of the signature
carried by~$a$.
The key to showing the soundness theorem is maintaining that these two
properties match up throughout execution.
This is captured by the definitions of two additional compatibility results.
Intuitively, a graph is compatible with a thread pool~$\tp$ if
for all~$\cthread{a}{\prio_a}{\sig_a}{\stackstate_a} \in \tp$,
the last vertex of~$a$ in the graph ``knows about'' all threads in~$\sig_a$
through a path in the graph.
A graph is compatible with a heap~$\mem$ if
for all~$\mement{\kwassn}{v}{u}{\sig} \in \mem$, the vertex~$u$ ``knows about''
all threads in~$\sig$ through a path in the graph.

\begin{defn}
  A graph~$\dagq{\gthreads}{\spawns}{\syncs}{\reads}$
  is {\em compatible} with a thread pool~$\tp$ if
  for all~$\cthread{a}{\prio_a}{\sig_a}{\stackstate_a} \in \tp$,
  there exists~$\gthread{a}{\prio_a}{\uthread_a \tscomp u_a} \in \gthreads$
  such that for all~$\sigtype{b}{\tau_b}{\prio_b} \in \sig_a$,
  there exists~$(u, b) \in \spawns$ and either~$u = u_a$ or
  there exists a path from~$u$ to~$u_a$ in~$\graph$
  whose first edge is not a {\create} edge.
\end{defn}

\begin{defn}
  A graph~$\dagq{\gthreads}{\spawns}{\syncs}{\reads}$
  is {\em compatible} with a heap~$\mem$ if
  for all~$\mement{\kwassn}{v}{u}{\sig} \in \mem$,
  there exists~$\gthread{a}{\prio_a}{\uthread_a} \in \gthreads$
  such that~$u \in \uthread_a$
  and for all~$\sigtype{b}{\tau_b}{\prio_b} \in \sig$,
  there exists~$(u', b) \in \spawns$ and
  there exists a path from~$u'$ to~$u$ in~$\graph$
  whose first edge is not a {\create} edge.
\end{defn}

Lemma~\ref{lem:soundness-gstep} shows that one parallel step of the operational
semantics preserves strong well-formedness of a graph and does not introduce
cycles into the graph.
The lemma also assumes both of the relevant compatibility properties hold,
and ensures that they are preserved.

\begin{lemma}\label{lem:soundness-gstep}
  Let~$\sig$ and~$\tp$ be such that for
  all~$\cthread{a}{\prio_a}{\sig_a}{\stackstate_a} \in \tp$,
  we have~$\sstyped{\sig_a,\sig}{\stackstate_a}{\tau}{\prio_a}$.
  Let~$\graph$ be strongly well-formed, acyclic and compatible
  with both~$\tp$ and~$\mem$.
  If~$\gconfig{\sig}{\mem}{\graph}{\tp} \gstep
  \gconfig{\sig'}{\mem'}{\graph'}{\tp'}$,
  then~$\graph'$ is strongly well-formed, acyclic and compatible
  with~$\tp'$ and~$\mem'$.
\end{lemma}
\begin{proof}
  By inversion on~\rulename{D-Par},
  we have~$\graph' = \graph \scomp{a_1} \graph_1' \dots \scomp{a_n} \graph_n'$
  where for each~$i \in [1, n]$,
  \[
  \lconfig{\mem}{\tp}{\cthread{a_i}{\prio_i}{\sig_i}{\stackstate_i}}
         \mstep
         \rconfig{\cthread{a_i}{\prio_i}{\sig_i'}{\stackstate_i'}}
                 {\tp_i'}{\sig_i''}{\mem_i'}
                 {\graph_i'}
  \]
  Let~$\graph_i' = \dagq{\gthreads_i}{\spawns_i}{\syncs_i}{\reads_i}$.
  We first show that~$\graph'$ is strongly well-formed.
  Because~$\graph$ is well-formed, it suffices to show that
  all {\touch} edges in~$\syncs_1 \cup \dots \cup \syncs_n$
  obey the requirements of Definition~\ref{def:strongly-well-formed}.
  Let~$(b, u) \in \syncs_i$.
  By Lemma~\ref{lem:soundness-mstep},
  $\gthread{a_i}{\prio_i}{\sthread{u}} \in \gthreads_i$
  and~$\sigtype{b}{\tau_b}{\prio_b} \in \sig_i$ and~$\ple{\prio_i}{\prio_b}$.
  By compatibility and the definition of sequential graph composition,
  there exists~$(u', b) \in \graph$ and a path from~$u'$
  to~$u$ whose first edge is a continuation edge.

  We next show that~$\graph'$ is acyclic.
  Because~$\graph$ is acyclic, it suffices to show that for all
  added edges~$(u, u')$, we have~$\nanc{u'}{u}$.
  This is immediate because, by Lemma~\ref{lem:soundness-mstep},
  for all such edges,~$u' \not\in \graph$ and the introduced edges contain
  no cycles.

  Next, we show that~$\graph'$ is compatible with~$\tp'$.
  Let~$\cthread{a_i}{\prio_i}{\sig_i'}{\stackstate_i'} \in \tp'$.
  By Lemma~\ref{lem:soundness-mstep},
  $\gthread{a_i}{\prio_i}{\sthread{u}} \in \gthreads_i$
  and if~$\sigtype{b}{\tau_b}{\prio_b} \in \sig_i' \setminus \sig_i$
  then (1)~$(u, b) \in \spawns_i$ or
  (2)~$(c, u) \in \syncs_i$
  and~$\sigtype{b}{\tau_b}{\prio_b} \in \sig_c$
  or (3)~$(u', u) \in \reads_i$
  and~$\mement{\kwassn}{v}{u'}{\sig_{u'}} \in \mem$
  and~$\sigtype{b}{\tau_b}{\prio_b} \in \sig_{u'}$.
  In case (1), compatibility is immediate.
  In case (2), let~$u_c$ be the last vertex of thread~$c$ in~$\graph$.
  By compatibility of~$\graph$, there exists
  $(u', b) \in \graph$ where~$u' = u_c$ or there is a path from~$u'$
  to~$u_c$ whose first edge is not a {\create} edge.
  The edge~$(c, u)$ creates a path from~$u'$ to~$u$
  in~$\graph'$ whose first edge is not a {\create} edge.
  In case (3), by compatibility of~$\graph$, there exists
  $(u'', b) \in \graph$ where~$u'' = u'$ or there is a path from~$u''$ to~$u'$
  whose first edge is not a {\create} edge.
  The edge~$(u', u)$ creates a path from~$u''$ to~$u$ whose first edge is not
  a {\create} edge.

  Finally, we show that~$\graph'$ is compatible with~$\mem'$.
  Let~$\mement{\kwassn}{v}{u}{\sig_u} \in \mem_i'$
  and~$\sigtype{b}{\tau_b}{\prio_b} \in \sig_u$.
  By Lemma~\ref{lem:soundness-mstep},
  $\gthread{a_i}{\prio_i}{\sthread{u}} \in \gthreads_i$
  and~$\sigtype{b}{\tau_b}{\prio_b} \in \sig_i$.
  By compatibility of~$\graph$,
  there exists~$(u', b) \in \spawns$ and a path from~$u'$ to~$u$ in~$\graph'$
  whose first edge is not a {\create} edge.
\end{proof}

The soundness theorem is then simply an inductive application of
Lemma~\ref{lem:soundness-gstep}.

\begin{theorem}\label{thm:soundness}
  Let~$\cmd$ be such that~$\cmdtyped{\esig}{\ectx}{\cmd}{\tau}{\prio}$.
  If
  \[
  \gconfig{\esig}{\emem}{\egraph}
          {\cthread{a}{\prio}{\esig}{\ssend{\estack}{\cmd}}}
  \gstep^*
  \gconfig{\sig}{\mem}{\graph}
          {\tp}
  \]
  then~$\graph$ is strongly well-formed and acyclic.
\end{theorem}
%% \begin{proof}
%%   By~\rulename{KS-Empty} and~\rulename{KS-PopCmd}, we
%%   have~$\sstyped{\esig}{\ssend{\estack}{\cmd}}{\tau}{\prio}$.
%%   %
%%   The requirements for strong well-formedness and compatibility are
%%   trivially met by the empty graph, heap and thread pool.
%%   %
%%   The result then follows from repeated applications of
%%   Lemma~\ref{lem:soundness-gstep}.
%% \end{proof}

Having shown that a well-typed program yields a well-formed cost graph, we
can now use the cost bounds of Section~\ref{sec:dag} to bound the length of
a program execution by considering an execution of the program using the
operational semantics to be a schedule of the resulting DAG.
There is one complication here: our bounds apply to prompt schedules, and
rule~\rulename{D:Par}, by construction, can choose any set of threads to
execute at a given time; it is not required to be prompt.
%
%Scheduling techniques, that is, techniques for selecting threads to execute at
%runtime, are outside the scope of this paper.
%
A na{\"\i}ve scheduling algorithm that would be easy to implement in our
formal semantics would be to simply select as many threads as possible,
ordered by priority, up to a given number of processors~$P$.
Such a scheduler would meet the formal bounds of promptness but a real-world
implementation of it would require too much synchronization to be practical.
These concerns are outside the scope of this paper, and so we instead add as
a condition of the theorem that, in an execution that is under
consideration,~\rulename{D:Par} chooses threads to execute in a prompt manner.

\begin{theorem}\label{thm:cost-corr}
  Let~$\cmd$ be such that~$\cmdtyped{\esig}{\ectx}{\cmd}{\tau}{\prio}$.
  Suppose
  \[
  \gconfig{\esig}{\emem}{\egraph}
          {\cthread{a}{\prio}{\esig}{\ssend{\estack}{\cmd}}}
  \gstep^*
  \gconfig{\sig}{\mem}{\graph}
          {\tp}
  \]
  and at each step in this execution, threads are chosen in a prompt
  manner.
  Let~$\prioworkof{\compwork{}{a}}{\psnlt{\prio}}$ be the competitor
  work of~$a$ in~$\graph$,
  and~$\longsp{\compwork{}{a}}{a}$ be the~$a$-span.
  If thread~$a$ is active for~$\resptimeof{a}$ steps during the
  execution, then
  \[
  \resptimeof{a} \leq %\frac{1}{\fc(\psnlt{\prio'})}
  \frac{1}{P}(
  \prioworkof{\compwork{}{a}}{\psnlt{\prio}} +
      (P-1)\longsp{\compwork{}{a}}{t})
      \]
\end{theorem}
\begin{proof}
  We construct a schedule of~$\graph$ in which a time step of the schedule
  corresponds to a parallel step of the execution and, at each step, the
  vertices added to the graph by the operational semantics are the vertices
  executed by the schedule.
  By assumption, this schedule is prompt.
  We must show that it is admissible.
  Consider a step~$\gconfig{\sig_i}{\mem_i}{\graph_i}{\tp_i} \gstep
  \gconfig{\sig_{i+1}}{\mem_{i+1}}{\graph_{i+1}}{\tp_{i+1}}$
  that adds a weak edge~$(u, u')$ to~$\graph_i$.
  Such an edge is only created by rule~\rulename{D-Get2} and is created
  at the time that~$u'$ is executed.
  It thus suffices to show that at such a time,~$u$ will have already
  been executed.
  By inversion on~\rulename{D-Get2}, we
  have~$\mement{s}{v}{u}{\sig_{u}} \in \mem_i$.
  By inductive application of Lemma~\ref{lem:soundness-gstep},
  we have that~$\graph_i$ is compatible with~$\mem_i$, and so~$u \in \graph_i$,
  meaning that~$u$ has already been executed.

  The constructed schedule is thus prompt and admissible.
  By Theorem~\ref{thm:soundness},~$\graph$ is strongly well-formed
  (and thus well-formed by Lemma~\ref{lem:wf-alt}).
  The result then follows immediately from Theorem~\ref{thm:gen-brent}.
\end{proof}

\subsection{Extensions of {\calcname}}
The calculus of the last two sections represents the essence of an imperative
parallel language.
Of course, the implementations and case studies of the following sections use
additional features not present in {\calcname}, but many such features could
be added fairly conservatively.
For example, the case studies use atomic operations
like compare-and-swap (CAS) to implement concurrent data structures.
The CAS operation is non-blocking and so can be added to {\calcname} without
fear of creating priority inversions.
As a proof of concept, we present representative inference rules for
adding CAS to the dynamic semantics:
\begin{mathpar}
  \Rule{D-CAS1}
       {u \fresh\\
         \mem(\kwassn) = \memrent{v}{u'}{\sig'}\\
         v = v_{\text{old}}
       }
     {
       \phantom{\mstep}
       \lconfig{\mem}{\tp}{\cthread{a}{\prio}{\sig,\sigrtype{\kwassn}{\tau'}}
         {\sreturn{\scp{\stack}{\kw{cas}(\kwref{\kwassn}, v_{\text{old}},
               \shole)}}{v_{\text{new}}}}}
       \\\mstep
       \rconfig{\cthread{a}{\prio}{\sig,\sigrtype{\kwassn}{\tau'}}
         {\screturn{\stack}{\kwret{\kwnumeral{1}}}}}
               {\etp}{\esig}
               {\mem[\mement{\kwassn}{v_{\text{new}}}{u}{\sig}]}
               {\tgraph{a}{\prio}{\sthread{u}}}
     }
     \and
     \Rule{D-CAS2}
       {u \fresh\\
         \mem(\kwassn) = \memrent{v}{u'}{\sig'}\\
         v \neq v_{\text{old}}
       }
     {
       \phantom{\mstep}
       \lconfig{\mem}{\tp}{\cthread{a}{\prio}{\sig,\sigrtype{\kwassn}{\tau'}}
         {\sreturn{\scp{\stack}{\kw{cas}(\kwref{\kwassn}, v_{\text{old}},
               \shole)}}{v_{\text{new}}}}}
       \\\mstep
       \rconfig{\cthread{a}{\prio}{\sig,\sigrtype{\kwassn}{\tau'}}
         {\screturn{\stack}{\kwret{\kwnumeral{0}}}}}
               {\etp}{\esig}
               {\mem}
               {\tgraph{a}{\prio}{\sthread{u}}}
     }
\end{mathpar}
\else
The soundness theorem for the type system states that well-typed programs
have well-formed cost graphs.
\begin{theorem}\label{thm:soundness}
  Let~$\cmd$ be such that~$\cmdtyped{\esig}{\ectx}{\cmd}{\tau}{\prio}$.
  If
  \[
  \gconfig{\esig}{\emem}{\egraph}
          {\cthread{a}{\prio}{\esig}{\ssend{\estack}{\cmd}}}
  \gstep^*
  \gconfig{\sig}{\mem}{\graph}
          {\tp}
  \]
  then~$\graph$ is well-formed and acyclic.
\end{theorem}
The proof of this theorem consists of showing that all steps maintain two
invariants:
\begin{enumerate}
\item No strong edges go from lower to higher priority
\item $\sig$ correctly reflects the ``knows about'' relation motivated above.
\end{enumerate}
These invariants respectively imply the two well-formedness requirements
of Section~\ref{sec:dag}.
Full proof details are available in the supplementary material.
\fi

\vspace{2ex}
\section{Implementation of \SYS} \label{sec:impl}

This section presents the design and implementation of \SYS, our prototype
task-parallel platform that supports parallel interactive applications.
\SYS is based on an open-source implementation~\cite{CilkF-impl}
of Cilk (a parallel dialect of C/C++) called Cilk-F~\cite{SingerXuLe19}
that extends Cilk with support for futures.    
The implementation of \SYS consists of two main components, a type system to
rule out priority inversions (closely following the typing rules discussed in
\secref{lang}) and a runtime scheduler that automates load balancing while
prioritizing high-priority tasks over lower-priority ones.  
\iffull 
This section discusses the implementations of these components.  Before we
discuss each of these components, we first briefly discuss the programming
interface supported by \SYS for writing interactive parallel applications.
\fi

\subsection{Programming Interface}

% To support parallel interactive applications, \SYS extends Cilk-F with the
% notion of priorities in the style of \calcname{} and a special type of
% futures, called latency-hiding futures, that can be used to hide I/O
% latency. 

\paragraph{Thread creation.} In \SYS, like in \calcname{}, a
function $f$ can invoke
another function $g$ with \create, which indicates that the execution of $g$
is logically in parallel with the continuation of $f$ after \create.  A
function invocation prefixed with \create returns a handle to the new thread,
on which one can later invoke \touch to ensure that the thread
terminates before the control passes beyond the \touch statement.  Since a
thread handle can be stored in a data structure or global variable and
retrieved later, the use of \create and \touch can generate irregular
parallelism with arbitrary dependences.  In \SYS, as is common in C-like
languages, it is possible to allocate a variable of thread handle type
without associating it to a thread, and later pass this variable
by reference to \create, to associate it with the created thread.
This is in contrast to \calcname{}, where the allocation of the handle
and the creation of the thread happen simultaneously.
\footnote{\SYS
additionally supports \spawn and \sync for nested parallelism.
The use of \spawn and \sync can be
subsumed by \create and \touch from the type checking perspective and hence we
omit the discussion here.}

\paragraph{I/O Operations.} \SYS supports the use of I/O operations via a
special type of thread, called an \ioF, that performs an I/O operation in a
latency-hiding way.  Specifically, \SYS provides special versions of  the
\cilkRead and \cilkWrite functions that behave similarly to the
Linux \texttt{read} and \texttt{write} except that they return a \ioF
reference representing the I/O operation.  Upon invocation,
\cilkRead and \cilkWrite create a thread to perform the I/O without
occupying the processor, and the returned \ioF can be used to wait on the I/O
by calling \touch on it.

\subsection{Type System}

The type system in \SYS  does not provide full type safety guarantees,
as C++ is not type safe.  Nevertheless, provided that the programmer follows a
set of simple rules, the C++-based type system can ensure that a program that
type checks will result in strongly well-formed DAGs when executed.  The type
system enables us to type check moderately large benchmarks that implement
interesting functionalities involving the use of low-level system calls and
concurrent data structures (discussed in \secref{eval-app}).
% We first describe our system setup, explain how we use C++ language features
% to enforce typing rules, present the programming interface, and finally end
% with a discussion on type safety with the set of rules a programmer should
% follow. 

\subheading{Enforcing Typing Rules} 

We utilize templates and other C++11 language features to encode the type
system.  In the C++ encoding, each priority is represented as a \code{class}.
The relationship between two priorities is captured through the class
hierarchy via inheritance; if priority $\rho$ inherits from priority $\rho'$
or some descendant of $\rho'$, then $\rho \succ \rho'$ (i.e., $\rho$ has
higher priority than $\rho'$).  Such relationships can be tested at compile
time using \code{is\_base\_of}, which tests whether one \code{class} is either
the same as or the ancestor of another.  Unlike in \calcname{}, priorities are
thus user-defined types rather than a pre-defined set of constants.
%\sr{Commented out the sentence above. The calculus requires priorities to be
%  from a pre-defined set of constants but the PriML implementation allows them
%  to be user-defined.}

In \calcname{}, there is a separation between the command layer and expression
layer.  In \SYS, the separation is not as clear.  However, we must enforce
restrictions on which functions can be invoked with \create (generating a
handle that can be \touch{}ed later) and which function can execute \touch,
because the priority of such functions must be retrievable at compile time in
order to enforce the typing rules.  We require these functions to be wrapped
in a \code{command class} whose type relies on a template that specifies its
execution priority.  For ease of discussion, we will refer to such a function
as a \code{command} function.  Unlike in {\calcname}, \create is not a command
--- code at any priority may safely invoke a function with \create; this
causes no difficulties in enforcing the typing guarantees.  Also unlike in
\calcname{}, code in \SYS does not require special syntax for invoking an
expression (e.g., function that is not a \code{command}) within a command.
%\sr{Check that the above makes sense}

The encoding of the type system is realized by C++ macros that transform
\create, \touch, and declarations / invocations of command functions into the
necessary C++ encodings.\footnote{We additionally provide macros for declaring
and defining a \code{command} function to ease the use of \code{command}
functions.}  The templated types of \code{command} functions allow their
priority to be known at compile time, and the type system checks for priority
inversion at the execution of \touch.  First, a function invoked with \create
(which must be a \code{command} function) returns a thread handle whose type
is templated with its priority and return type (i.e., what its corresponding
thread returns when done executing, which may be void).  Second, an
\touch can only be executed from within a \code{command} function, and 
\touch on a thread handle \code{fptr} is translated to:
\begin{lstlisting}[language={[11]C++}, firstnumber=last]
fptr->touch();
static_assert(is_base_of<this->Priority, 
                         fptr->Priority>::value,
  "ERROR: priority inversion on future touch");
\end{lstlisting}
The static assert ensures that the thread invoking the
\touch has priority lower than or equal to that of the thread whose handle is
\touch{}ed, causing a compiler error otherwise.  

Lastly, we enforce that a \code{command} function $g$, if invoked by another
\code{command} function $f$, must be invoked with \create or
inherits the priority of $f$.\footnote{Currently this is enforced by 
name mangling \code{command} functions which can be circumvented, but in
principle this can be enforced with better compiler support.} Doing so
ensures that another \code{command} function $h$ joining with $f$ (with
lower priority than $f$ but higher priority than $g$) does not suffer from
priority inversion by waiting on $g$. 
In \calcname{} such an issue does not arise because call is an expression
whereas \create is a command, and therefore the two do not mix.  This issue is
an artifact of the fact that the distinction between the command and the
expression is not clear in \SYS.  

\subheading{Discussion: Type Safety}

Ideally we would like to guarantee that programs which type check
using our API will always generate strongly well-formed DAGs when executed.
However,
we cannot make this guarantee in full because C++ is not a type-safe language.
Nevertheless, provided that the programmer
follows a few simple rules, our type system can statically prevent
cases of priority inversions, and a program that type checks will result in
strongly well-formed DAGs when executed. 

\iffull
The first rule is that the programmer should not use unsafe type casts.  Type
casts circumvent the type system; the programmer can use \create to invoke a
\code{command} function \code{foo} with a low priority, but at the point of
\touch to join with \code{foo} can type cast the thread handle
to be of a higher priority.  The
code would type check, as at the point of \touch, the thread handle is
interpreted at a higher priority.  Similarly, one could allocate a thread
handle \code{thread_pointer} of a high priority, but then type cast it to a
lower priority at the point of \create to create a thread with low
priority.  Because the priority check is done against
the priority type of the \code{thread_pointer}, one can trick the type system
into thinking that a \touch does not cause a priority inversion when the
thread associated with the handle is actually of a lower priority.   
\else
The first rule is that the programmer should not use unsafe type casts,
which circumvent the type system and allow the programmer to modify 
priority types in ways that the type system cannot detect.
\fi

\iffull
The second rule is that the programmer should always ensure that a thread handle
is already associated with a thread (i.e., the handle has been used to
invoke a function via \create) before invoking \touch on it.  This rule is
important because a strongly well-formed DAG, by Definition
\ref{def:strongly-well-formed}, must have a path between the vertex that invokes
the \create and the vertex that invokes the \touch.  This requirement is
trivially satisfied in PriML because a thread handle
cannot come to existence separately from the \create Command.
Thus, if a thread handle exists and is being touched, the
thread must have been created (i.e., a path exists between them).
In C++, we allow for the thread handle allocation to be
separated from the creation of its thread (a C++ function invocation).
Thus, such a requirement is not trivially satisfied, and the programmer has to
ensure this is the case manually. 
\else
The second rule is that the programmer should always ensure that a thread
handle is already associated with a thread (via \create) before
invoking \touch on it.  This rule is important because a strongly well-formed
DAG must have a path between the vertex that invokes the \create and the
vertex that invokes the \touch. 
%\ar{I assume the path can include weak edges.}
This is trivially satisfied in \calcname{} because allocation and creation are
inextricably linked, but in \SYS a thread handle allocation can be separate from
its thread creation.  Thus, such a requirement is not trivially 
satisfied, and the programmer has to manually ensure the thread has 
been created before an \touch.
\fi

\subsection{Runtime Scheduler}

An execution of an \SYS program generates a computation DAG as described in
\secref{dag} that dynamically unfolds on the fly, and the underlying runtime
schedules the computation in a way that respects the dependences in the DAG.
\SYS, like Cilk-F, schedules the computation using proactive work
stealing~\cite{SingerXuLe19} but in addition, prioritizes threads.  

Recall from \secref{dag} that one can bound the response times of threads in a
well-formed DAG (\thmref{gen-brent}), provided that the schedule is admissible
and \term{prompt}, i.e., the schedule assigns a ready vertex~$u$ such that no
currently unassigned vertex is higher-priority than $u$.  Any schedule
produced by an actual execution is admissible by construction.
%A scheduler that ensures promptness precisely and efficiently,
%however, is difficult to achieve in practice.
Promptness, however, requires the scheduler to find ready vertices of
high-priority threads in the system to assign before vertices of
lower-priority threads.  Doing so requires maintaining centralized
information, which becomes inefficient in practice due to frequent
synchronizations.  Thus, \SYS implements a scheduler that approximates
promptness. 

Specifically, \SYS uses a two-level scheduling scheme, similar to the
scheme proposed by prior work A-STEAL~\cite{AgrawalHeLe06, AgrawalHeLe07}.
The top-level \term{master} scheduler determines how to best assign processing
cores to different priority levels, and threads within each priority level are
scheduled with a second-level \term{work-stealing}
scheduler~\cite{BlumofeLe94, blumofele99, arorablpl98,abp-multi-01}, known for
its decentralized scheduling protocol with low overhead and provably efficient
execution time bound.  \SYS utilizes a variant of work stealing called
\term{proactive work stealing}~\cite{SingerXuLe19} inherited from Cilk-F, the
baseline scheduler \SYS extends.

The master scheduler evaluates the cores-to-priority-level assignments in a
fixed scheduling interval, called the \term{scheduling quantum}.  The master
assigns cores based on the desired number of cores reported by the
work-stealing schedulers of each priority-level, but in a way that prioritizes
high-priority threads --- it always assigns cores in the order of priority.
Thus, the highest priority always gets its requested cores up to the limit of
what is available on the system, and the next levels get the left-over cores.

The work-stealing scheduler at each priority level maintains its
\term{desire}, the number of cores it wishes to get.  At the end of a quantum,
the scheduler for a given priority level determines its core utilization in
this quantum and re-evaluates its desire based on the measured utilization and
whether its desire was satisfied in this quantum.  Because a work-stealing
scheduler is either doing useful work (making progress on the computation), or
attempting to steal (which leads to load balancing), its \term{utilization} is
computed by the fraction of processing cycles that went into doing work.  If
its utilization exceeded a fixed threshold (e.g., $90\%$) and its desire was
satisfied (i.e., it got its desired number of cores), it increases its desire
by a multiplicative factor of the \term{growth parameter} $\gp$.  For
instance, if $\gp = 2$, double the desire.  On the other hand, if the
utilization exceeded the threshold but its desire was not met, it keeps the
same desire.  Finally, if the utilization did not meet the threshold, it
reduces its desire by a factor of $\gp$ (e.g., if $\gp = 2$, halve the
desire).

Prior work~\cite{AgrawalHeHs06, AgrawalHeLe07, AgrawalHeHs08} has analyzed
similar two-level strategies and shown that one can bound the wasted cycles
(i.e., due to low utilization) and the execution time of computations
scheduled by the second-level schedulers.  The prior analyses do not directly
apply in our case, however, for two reasons.  First, \SYS utilizes proactive
work stealing for the second-level schedulers, which differs from the ones
analyzed in prior work.  Second, in prior work, the computations scheduled by
the second-level schedulers are independent, whereas in our case, each
second-level scheduler corresponds to a priority level, and threads in
different priority levels can have dependences.  Nevertheless, in
\secref{eval}, we show that our scheduler does appropriately prioritize
high-priority threads over low-priority ones and provides better response time
for high-priority threads compared to the baseline system that does not account
for priorities.

%\clearpage
%\input{pml}
\secput{eval}{Evaluation of \SYS}

This section empirically evaluates \SYS.  To evaluate the practicality 
and usability of the type system, we wrote three moderately sized application
benchmarks: a proxy server (\proxy, ~1.5K LoC), a multi-user email client
(\emailapp, ~1.1K LoC), and a job server (\jserver, ~1.1K LoC).\footnote{LoC
exclude comments, system libraries, and runtime code.}
The type system helps the programmer ensure that there is no priority
inversion, which is not always easy to tell, as thread handles are
often used to coordinate interactions among different application components.
We also use the same applications to evaluate the efficiency of the scheduler 
by comparing \SYS against Cilk-F, the baseline system that utilizes proactive
work stealing but does not account for the priority of threads (and thus does
not incur the two-level scheduling overhead).  For fair comparison, Cilk-F is
also equipped with the same \ioF library that performs I/O operations in a
latency-hiding way.  We use this library for the I/O operations in the
benchmarks so that I/O-blocked threads do not hinder parallelism.
The empirical results indicate that \SYS was able to
prioritize high-priority threads and thus provide shorter response
times.

\paragraph{Experimental Setup.} Our experiments ran on a computer with 2 Intel
Xeon Gold 6148 processors with 20 2.40-GHz cores. Each core has a 32-kB L1 data
and 32-KB L1 instruction cache, and a private 1 MB L2 cache.
Hyperthreading was enabled, and each core had 2 hardware threads. Both
processors have a 27.5 MB shared L3 cache, and there are 768 GB of main
memory.  \SYS and all benchmarks were compiled using the Tapir
compiler~\cite{SchardlMoLe17} (based on clang 5.0.0), with -O3 and -flto.
Experiments ran in Linux kernel 4.15.

\subsection{Application Case Studies}\label{sec:eval-app}

We evaluate the type system with three applications representative of
interactive applications in the real world in that they utilize interesting
features commonly used to write such applications, such as low-level file
system and network libraries, and concurrent data structures implemented using
primitives such as fetch-and-add and compare-and-swap.  Due to space limitations, we
discuss the email client in detail but only summarize the other two
applications.

\paragraph{Proxy server.} 
The first application, \proxy, allows multiple clients to connect and request
websites by their URL. The server fetches the website on the client's behalf,
masking the client's IP address.  As an optimization, the server maintains a
cache of website contents using a concurrent hashtable.  If a website is
cached, the server can respond with it immediately.  The application utilizes
components with four priority levels, listed in order from highest to
lowest: a) the loop that accepts client connections and the per-client event
loop that handles the client requests, b) a component that fetches websites in
the event of a cache miss, c) a component that logs statistics, and the lowest is
d) the main function that performs server startup / shutdown.  The priority
specification favors response time for client requests.

\paragraph{Email client.} The second application \emailapp is a multi-user,
shared email client that allows users with individual mailboxes to sort
messages, send messages, and print messages; a background task also runs
periodically to reduce storage overhead by compressing each user's messages
using Huffman codes~\cite{CormenLeRi09}[Chp.~16.3].  The application contains
components with six priority levels, listed in order from highest to
lowest: a) an event loop to handle user requests, b) a send component
that sends email, c) a sort component that sorts emails, d) a compress
component to compress emails and a print component to uncompress and
send the uncompressed emails to the printer, e) a check component that
periodically checks for the need to compress and fires off compression,
f) the main function that performs shutdown.

One interesting feature is that the application requires the print
and compress to interact with one another --- if the user asks to print a
particular email but it is in the  midst of being compressed, the print
component needs to coordinate with the compress component and wait for it to
finish.  Similarly, the compress component may encounter an email
that it is about to compress, but it is in the midst of being printed, and
thus the compress needs to wait for the print to complete.

To enable this, within each user's inbox data structure is an array
indexed using the email ID where any thread attempting to print or
compress the email will store its own handle.  For instance, say there is an
ongoing print thread for an email.  The array slot corresponding to the email
stores (a pointer to) the handle of the print thread.  If a compress
thread for the same email is created, the first thing the compress thread does
is perform a compare-and-swap (CAS) on the same array slot, swapping out the
handle of the print thread and inserting a pointer to its own handle
into the slot.  Assuming that CAS returns a non-null
reference, the compress thread invokes \touch on the reference to ensure that
the printing is done before proceeding with the compress.

\iffull

The simplified pseudo code for the compress thread is shown below.
\begin{lstlisting}[firstnumber=last,language={[11]C++}]
int compress(int userID, int emailID, future<PrintOfComp, int> *thisFut) {
  future<PrintOrCompP, int> *prev = CAS(&userInbox[userID][emailID], thisFut);
  int emailState = DECOMPRESSED;
  if (prev) emailState = @\touch{}@ prev;
  if (emailState == DECOMPRESSED) compressEmail(userID, emailID);
  return COMPRESSED;
}
\end{lstlisting}

\fi
A print thread performs similar operations on the array to coordinate with an
ongoing compress thread for the same email.  Such an interaction is achieved
by utilizing the thread handles and mutable state in an interesting
way. 

\paragraph{Job server.} The \jserver application executes jobs that arrive in
the system using a smallest-work-first policy, i.e., given different types of
jobs, the server knows the amount of work entailed for each type, and it
prioritizes jobs with the least amount of work.  We simulate user inputs using
a Poisson process to generate jobs at random intervals and execute them.  The
priority levels correspond to the types of jobs.  We simulated four
different types of jobs with fixed input size $n$, listed in order of
priority (high to low): a) parallel divide-and-conquer matrix multiplication
(\texttt{matmul}, $n=1024$), b) fibonacci (\texttt{fib}, $n=36$), c) parallel
merge sort (\texttt{sort}, $n=1.1\times 10^7$),
and d) Smith-Waterman for sequence
alignment (\texttt{sw}, $n=1024$).  This application differs from the previous
two in that threads in different priority levels are independent of each
other, and it is constructed so that we can easily modify the workload to
simulate a server that is lightly loaded to heavily loaded.  

\paragraph{Compilation time.} Because the type system heavily utilizes templates,
we measure its effect by comparing the compilation time and resulting binary
sizes between code that uses priorities and code that does not.\footnote{The use of 
template can increase code size as each type instantiation of a given 
template gets its own code clone.}  As shown in \figref{comp}, 
the use of templates for enforcing the typing rules incurs acceptable 
overhead.

\begin{table}
\caption{The compilation times and resulting binary sizes of application code
without and with priority.  The compilation time is in seconds and the maximum
out of the three compile runs.  The binary size is in KB.  The numbers in
parentheses show overhead compared to the no priority version.}
\vspace{-2mm}
\begin{center}
\newcommand*{\oh}[1]{\scriptsize ($#1\times$)}
\footnotesize
\begin{tabular}{l|rr}
  \textit{case study} & \textit{compilation time} & \textit{binary size} \\
  \hline
  \proxy (w/out)    & 1.95 \oh{1.00} & 824.0 \oh{1.00} \\
  \proxy (with)     & 2.48 \oh{1.27} & 974.7 \oh{1.18} \\
  \hline
  \emailapp (w/out) & 4.66 \oh{1.00} & 1241.16 \oh{1.00} \\
  \emailapp (with)  & 5.40 \oh{1.16} & 1454.58 \oh{1.17} \\
  \hline
  \jserver (w/out)  & 2.10 \oh{1.00} & 851.2 \oh{1.00} \\
  \jserver (with)   & 2.67 \oh{1.27} & 987.7 \oh{1.16} \\
\end{tabular}
\end{center}
\label{fig:comp} 
\vspace{-2mm}
\end{table}

\subsection{Empirical Evaluation}

To evaluate the efficiency of our implementation, we compare the applications'
running times on \SYS
and on Cilk-F with the same latency-hiding I/O support.  The main distinctions
between the two systems are that a) \SYS prioritizes high-priority threads
whereas Cilk-F does not; and b) \SYS utilizes the two-level scheduling scheme
discussed in \secref{impl} whereas Cilk-F utilizes proactive work stealing
only.  For \SYS, we ran all applications with the following runtime 
parameters: utilization threshold of $90\%$, quantum length of $500$
microseconds, and growth parameter of $2$.  These parameter values seem to 
work well in general.   

Each of the applications represents different workload characteristics.  The
\proxy server has the most I/O latency and very little computation.  The
\emailapp has a fair amount of I/O latency and slightly more computation
than \proxy.  The \jserver has little I/O latency with
compute-intensive workloads.  We use one socket ($20$ cores) to run the
server and the second socket to simulate clients that generate inputs.
Each application is evaluated with multiple server load configurations that
range from lightly loaded to heavily loaded.  For \proxy and \emailapp, we ran
with $90$, $120$, $150$, and $180$ connections.  As we increase the number
of connections, each core needs to multiplex among more connections.
For \emailapp, the computation load also increases as the number of clients
increases.  For \jserver, we simulated the job generations so that the
workload results server machine utilization of $64\%$, $77\%$, $95\%$, and $>
95\%$ respectively.

For each application, we run the server for at least $15$ seconds, during
which tens of thousands of threads from various priority levels (which
correspond to different application components) are created, and we measure
their duration.  Specifically, we measure the \term{response} time of the
application, which corresponds to the time elapsed between when the user /
client sends the request to when the server handles the request (which is
always handled by the highest priority thread), and the \term{compute}
time for each thread of different priority levels.  

The standard deviation for such time measurements can be high for interactive
applications, due to multiple factors.  First, the timing includes the I/O
latency, which is not always uniform.  Second, the server is time-multiplexing
among multiple client connections, and thus the measured time of a thread
includes not only its computation time but also the time it took the server to
get to the threads.  As such, for many interactive applications, what one cares
about is the latency near the tail.  Thus, for all timing data, we show both
the average time and the $95^{th}$ percentile running time (i.e., $95\%$ of
the measured time is below that value). 

\figref{response} shows the response time ratio for \proxy and \emailapp (the
job server does not have a response time measurement as the jobs are generated
in the same process as the server).  We normalize the response time of Cilk-F
by that of \SYS, and thus higher means \SYS is more responsive.  As can be
seen, \SYS provides much better response time, appropriately prioritizing the
highest priority threads.  \SYS appears to be much more responsive for
\emailapp than for \proxy.  This is because \proxy is very lightly-loaded ---
most of the time cores are idling, as there isn't much computation in the
server execution (mostly I/O operations).  In contrast, \emailapp has more
computations to keep cores occupied, and thus high-priority threads can be
delayed much longer in Cilk-F as the cores are pre-occupied by computations
generated by lower-priority threads. 

\begin{figure}
\footnotesize
\iffull
\begin{minipage}[t]{0.35\columnwidth}
\else
\begin{minipage}[t]{0.48\columnwidth}
\fi
\centering
\includegraphics[width=.98\columnwidth]{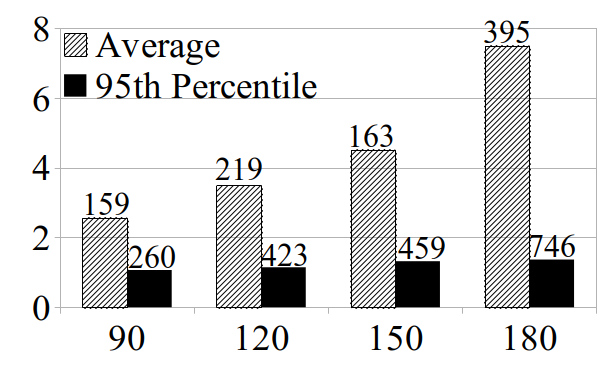}
\proxy responsiveness ratio
\end{minipage}
\iffull
\hspace{1cm}
\begin{minipage}[t]{0.35\columnwidth}
\else
\begin{minipage}[t]{0.48\columnwidth}
\fi
\centering
\includegraphics[width=.98\columnwidth]{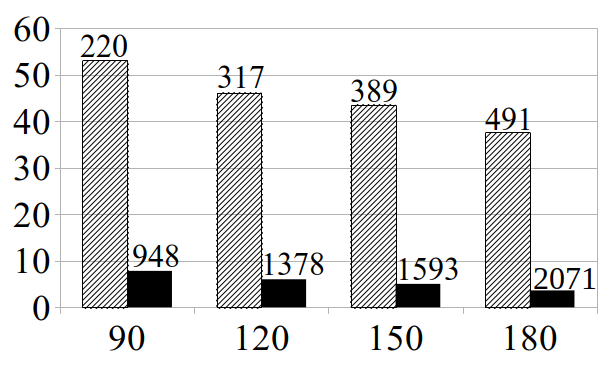}
\emailapp responsiveness ratio
\end{minipage}
\vspace{-2mm}
\caption{The relative responsiveness of \proxy and \emailapp, 
measured as the response time running on Cilk-F normalized by \SYS response
time, so higher means \SYS is more responsive.  Grey bars show the 
responsiveness calculated using average response time and black bars show
that calculated using the response time at $95^{th}$ percentile.  The x-axis
shows the number of client connections used.  The numbers shown on
top of the bars are the latency measured in microseconds for \SYS.}
\label{fig:response}
\vspace{-2mm}
\end{figure}

The high responsiveness is achieved by prioritizing the high-priority threads,
sometimes at the expense of the lower priority threads.  \figref{compute}
shows the computation times of threads from different components.  For a given
application and a given configuration (e.g., \proxy with $90$ clients), the
bars from left to right show the normalized compute time for threads from
higher to lower priority.  As the figures show, \SYS provides better compute
time than Cilk-F for the highest priority threads in the figures (which is the
second highest priority for \proxy and \emailapp).  However, the lower
priority threads can run slower.  This trend can be seen across different
server loads, where the compute time ratio for the higher priority threads
grows larger as the load gets heavier.  This is because the compute time for
the higher priority threads on Cilk-F degrades as the server gets more heavily
loaded, whereas \SYS is able to maintain similar level of quality of service.
For the lower-priority threads, compute time on both systems degrades, with
\SYS degrading more especially when the load gets heavy.
 
\begin{figure*}
\footnotesize
\begin{minipage}[t]{0.32\textwidth}
\centering
\includegraphics[width=.98\columnwidth]{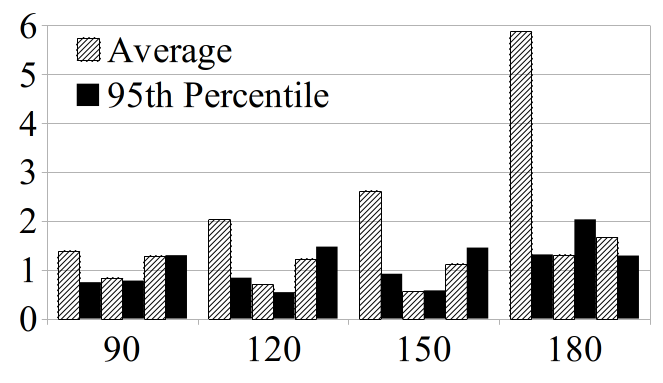}
\proxy compute time ratio
\end{minipage}
\begin{minipage}[t]{0.32\textwidth}
\centering
\includegraphics[width=.98\columnwidth]{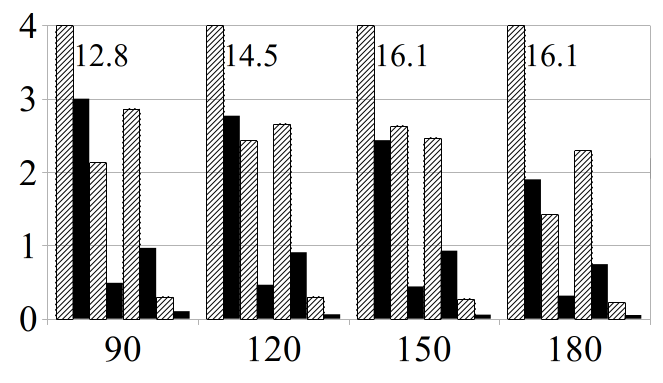}
\emailapp compute time ratio. Data labels indicate the values of bars that exceed the scale of the graph.
\end{minipage}
\begin{minipage}[t]{0.32\textwidth}
\centering
\includegraphics[width=.98\columnwidth]{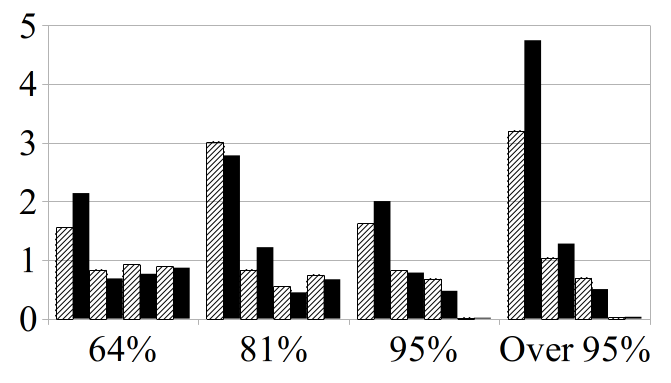}
\jserver compute time ratio
\end{minipage}
\vspace{-2mm}
\caption{The relative compute time for \proxy, \emailapp, and \jserver, measured
as the compute time running on Cilk-F normalized by that running on \SYS, so
higher means \SYS computes faster.  For a given application and a given
configuration (e.g., \proxy with $90$ clients), the bars from left to right
show the compute time ratio for threads from the highest to the lowest 
priority.  The grey bars show the compute time ratio calculated using 
average compute time and the black bars show that calculated using the 
compute time at $95^{th}$ percentile.  The x-axis shows the number of 
client connections used for \proxy and \emailapp, and the server 
utilization for \jserver.}
\label{fig:compute}
\vspace{-3mm}
\end{figure*}

\section{Related Work}\label{sec:related}

\subheading{Cooperative Parallelism}

Many languages and systems have been developed for cooperative parallelism over
the years.
A large number of these, such as Id~\citep{id-78}, Multilisp~\citep{halstead85},
NESL~\citep{nesl-94} and parallel versions of
Haskell~\citep{haskell-dp-2007,keller+2010}
and ML~\citep{manticore-implicit-11,jagannnathannasizi10,rmab-mm-2016},
have focused on functional programming languages, in which the issues of
races and deadlock do not arise or are limited.

Some parallel language extensions have, however, targeted popular imperative
programming languages such as
C~\citep{frigolera98} and Java~\citep{lea00,x10-2005,bad+09,is-habanero-14}.
Many papers have been devoted
over the years to taming races
(e.g.~\citep{RamanZhSa12,feng97efficient,LeeSc15,UtterbackAgFi16,XuLeAg18})
and deadlock (e.g.~\citep{VossCoSa19, CogumbreriroSuMaSaVaGr17, CogumbreiroHuMaYo15,AgarwalBaBoSaShRuYe07}).
None of these languages allow the cooperative threads to be
prioritized; doing so, as we do in this work, requires reasoning about
{\em priority inversions} in addition to the problems mentioned above.

\subheading{Scheduling for Responsiveness}

Responsiveness has long been a concern in the systems community,
as operating systems must
schedule processes and threads, many of which are interactive.
A thorough overview of this topic can be found in a text by~\citet{sgg-os-05}.
In contrast to cooperative parallel systems, OS schedulers deal with
relatively small numbers of
threads.

%Much of the work in this community is done in relation to low-level
%systems languages such as C and C++, and so must address the issues of
%data races and deadlock.
%
Many threading systems for which responsiveness is a concern
incorporate some notion of priority.
The problem of {\em priority inversion} has been noted in systems as
early as Mesa~\citep{lr-mesa-1980}.
\citet{bms-formal-1993} formalized the idea of priority inversions and
discussed some techniques by which they could be prevented.

Recent work~\citep{mah-responsive-2017,mah-priorities-2018} has introduced
thread priorities into a cooperative parallel system and developed type systems
for ruling out priority inversions that arise through touching a future.
That work, however, targets a purely functional subset of Standard ML, and
so future handles can essentially only be passed through calls and returns,
leading to a well-behaved DAG structure.
In this paper, future handles can additionally be passed through mutable state,
leading to much more complicated reasoning about priority inversions.

\subheading{Cost Semantics}

Cost semantics~\citep{Rosendahl89,Sands90} are used to reason statically
about the resource usage, broadly construed, of programs.
Cost semantics for parallel
programs~\citep{BlellochGr95,BlellochGr96,SpoonhowerBlHaGi08}
typically represent the parallel structure of the program as a DAG.
Offline scheduling results bound the time required to execute such a DAG
on~$P$ processors in terms of the work and span of the DAG.
Famous offline scheduling results have shown that a ``level-by-level''
schedule~\citep{brent74} and any greedy schedule~\citep{eagerzala89} are within
a factor of two of optimal.
Although the full formality of the DAG model is usually reserved for proofs,
the metrics of work and span and the scheduling results above are quite useful
in practice for analyzing programs by thinking in terms of the parallel
structure of the underlying algorithm (e.g., the branching factor and problem
size in a divide-and-conquer algorithm).
Even in cases where the input is unknown, one can reason asymptotically about
work and span, much like asymptotic reasoning in sequential algorithms.
Recent work has extended parallel cost semantics to reason about
I/O latency~\citep{ma-latency-2016} and responsivenesss~\citep{mah-responsive-2017,mah-priorities-2018}.
This paper further extends the state of the art by adding {\em weak edges}
that allow DAGs to reflect information passed between threads through
global state.

Much of the above prior work has drawn a distinction between the
cost semantics, which uses a very abstract evaluation model
to produce a cost DAG from a program, and a {\em provably-efficient}
or {\em bounded implementation}~\citep{BlellochGr96,mah-responsive-2017,mah-priorities-2018},
which counts the steps of an abstract machine.
A proof that the abstract machine actually meets the bounds promised by the
cost semantics can be quite technical and involved.
In this work,
we present one dynamic semantics that
both counts steps and produces a cost graph.
\iffull
This semantics reflects execution ordering, which is important in our
calculus, and simplifies the proof
that the steps of the abstract machine are bounded by the cost semantics
(this is part of the proof of Theorem~\ref{thm:cost-corr}).
\else
This semantics reflects execution ordering, which is important in our
calculus, and simplifies the proof
that the steps of the abstract machine are bounded by the cost semantics.
\fi

\section{Conclusion}

This paper bridges cooperative and competitive threading models by
bringing together a classic threading construct, futures, with
priorities and mutable state.
To facilitate reasoning about efficiency and responsiveness, the paper
extends the traditional graph-based cost models for parallelism to
account for priorities and mutable state.
The cost model applies only to computations that are free of priority
inversions.
To guarantee their absence, we present a formal calculus called
{\calcname} and a type system that disallows priority
inversions.
The cost model and the type system both rely on a novel technical
device, called {\em weak edges}, that represent run-time happens-before
dependencies that arise due to communication via mutable shared state.
We show that these theoretical results are practical by presenting a
reasonably faithful implementation that extends C++ with futures and priorities.
This extension offers an expressive substrate for writing interactive
parallel programs and is able to enforce the absence of priority
inversions if the programmer avoids certain unsafe constructs of C++.
Our empirical evaluation shows that the techniques work well in practice.

\bibliography{main,new,thisbib}
\end{document}